\documentclass[11pt]{article} 

\usepackage{natbib}

\usepackage{fullpage}
\usepackage{color}
\usepackage{booktabs}
\usepackage{amsfonts}
\usepackage{amssymb}
\usepackage{amsthm}
\usepackage{amsmath}
\usepackage{mathptmx}
\usepackage{soul}
\usepackage{graphicx} 
\usepackage{subcaption}

\usepackage{thmtools}
\usepackage{thm-restate}

\usepackage{ifpdf}
\newcommand{\mydriver}{hypertex}
\ifpdf
\renewcommand{\mydriver}{pdftex}
\fi
\usepackage[breaklinks,\mydriver]{hyperref}
\usepackage{svg}

\usepackage{url}
\usepackage{amsmath}
\usepackage{amsthm}
\usepackage{graphicx} 

\usepackage{tikz}
\usetikzlibrary{decorations.pathmorphing, calc}
\usepackage[margin=1in]{geometry}
\usepackage{subcaption}
\usepackage{soul}

\usepackage[ruled,vlined,linesnumbered]{algorithm2e}

\usepackage{hyperref}

\usepackage[capitalise]{cleveref}

\newtheorem{theorem}{Theorem}[section]
\newtheorem{lemma}{Lemma}[section]
\newtheorem{claim}{Claim}[section]
\newtheorem{observation}{Observation}[section]
\newtheorem{definition}{Definition}[section]
\newtheorem{corollary}{Corollary}[section]
\newtheorem{proposition}{Proposition}[section]
\newtheorem{remark}{Remark}[section]

\newcommand{\Nbr}{\texttt{Nbr}}

\newcommand{\eps}{\ensuremath{\varepsilon}}
\newcommand{\reg}{\ensuremath{\mathrm{reg}}}
\newcommand{\Det}{\ensuremath{\mathrm{Det}}}

\newcommand{\ER}{Erd\H{o}s-R\'{e}nyi~}
\newcommand{\Gnp}{G(n,p)}

\newcommand{\anc}{\mathrm{anc}}
\newcommand{\E}{\mathbb{E}}

\newcommand{\poly}{\mathrm{poly}}
\newcommand{\smallest}[1]{\text{Smallest}(#1)}
\newcommand{\largest}[1]{\text{Largest}(#1)}
\newcommand{\abs}[1]{\left| #1 \right|}

\newcommand{\w}{\textbf{w}}
\newcommand{\dist}{\mathrm{dist}}

\allowdisplaybreaks

\title{Local Computation Algorithms for (Minimum) Spanning Trees\\ on Expander Graphs}
\author{
Pan Peng\footnote{
School of Computer Science and Technology, University of Science and Technology of China. Email:
\href{mailto:ppeng@ustc.edu.cn}{ppeng@ustc.edu.cn}
}
\and
Yuyang Wang\footnote{
School of Computer Science and Technology, University of Science and Technology of China. Email:
\href{mailto:wangyvyang@mail.ustc.edu.cn}{wangyvyang@mail.ustc.edu.cn}
}
}

\date{}

\begin{document}

\maketitle
\begin{abstract}
We study \emph{local computation algorithms (LCAs)} for constructing spanning trees. In this setting, the goal is to locally determine, for each edge $ e \in E $, whether it belongs to a spanning tree $ T $ of the input graph $ G $, where $ T $ is defined implicitly by $ G $ and the randomness of the algorithm. It is known that LCAs for spanning trees do not exist in general graphs, even for simple graph families. We identify a natural and well-studied class of graphs -- \emph{expander graphs} -- that do admit \emph{sublinear-time} LCAs for spanning trees. This is perhaps surprising, as previous work on expanders only succeeded in designing LCAs for \emph{sparse spanning subgraphs}, rather than full spanning trees. We design an LCA with probe complexity $
O\left(\sqrt{n}\left(\frac{\log^2 n}{\phi^2} + d\right)\right)$  
for graphs with conductance at least $ \phi $ and maximum degree at most $ d $ (not necessarily constant), which is nearly optimal when $\phi$ and $d$ are constants, since $\Omega(\sqrt{n})$ probes are necessary even for expanders. Next, we show that for the natural class of \emph{\ER graphs} $ G(n, p) $ with $ np = n^{\delta} $ for any constant $ \delta > 0 $ (which are expanders with high probability), the $ \sqrt{n} $ lower bound can be bypassed. Specifically, we give an \emph{average-case} LCA for such graphs with probe complexity $ \tilde{O}(\sqrt{n^{1 - \delta}})$. 

Finally, we extend our techniques to design LCAs for the \emph{minimum spanning tree (MST)} problem on weighted expander graphs. Specifically, given a $d$-regular unweighted graph $\bar{G}$ with sufficiently strong expansion, we consider the weighted graph $G$ obtained by assigning to each edge an independent and uniform random weight from $\{1,\ldots,W\}$, where $W = O(d)$. We show that there exists an LCA that is consistent with an exact MST of $G$, with probe complexity $\tilde{O}(\sqrt{n}d^2)$.

\end{abstract}

\thispagestyle{empty}
\setcounter{page}{0}
\newpage

\section{Introduction}

Constructing a spanning tree is a fundamental problem in computer science and graph theory. In this paper, we study \emph{local computation algorithms (LCAs)} for constructing spanning trees. In this setting, the goal is to determine quickly whether a given edge $e$ belongs to some spanning tree, without computing the entire tree. Instead, the algorithm only access the local neighborhood of $e$ by probes\footnote{Some works refer to \emph{query access} to the graph. In our setting, however, since the algorithm is designed to answer queries, we reserve the term \emph{query} specifically for inputs received by the algorithm.} to graph. The key challenge is consistency:  ensuring that answers to every edge are all consistent with the same underlying spanning tree. Such LCAs are useful in scenarios where we do not need the full solution at once, or where multiple independent processes may wish to query edges in parallel.

The LCA model (also known as the centralized local model) was introduced by Rubinfeld et al. \cite{RubinfeldTVX11} and Alon et al. \cite{alon2012space}. It generalizes earlier frameworks such as graph property testing and sublinear-time algorithms for approximating graph parameters. All these models share a common motivation: inferring global properties of a graph while exploring only a small portion of the input through local probes. While property testing and approximation algorithms typically deal with large input and small output (e.g., estimating a global quantity or deciding a property), LCAs handle the more demanding large input and large output regime, where the goal is to construct a global structure, such as a spanning tree, locally and consistently on demand.

The study of LCAs with sublinear probe complexity for spanning trees was initiated by Levi et al. \cite{reut2014local,levi2020local} nearly a decade ago. They quickly identified a fundamental limitation: \emph{it is impossible to design a sublinear-time LCA for constructing spanning trees, even for very simple graph families!} For example, consider a path and a cycle, both of length $n$. In the path, every edge must be included in the spanning tree, while in the cycle, at least one edge must be omitted to break the cycle. Distinguishing between these two cases requires probing a linear number of edges, violating the sublinear probe constraint of LCAs. 

The above observation motivated a shift in focus toward LCAs for \emph{sparse spanning subgraphs} -- connected subgraphs that span all vertices and contain at most $(1+\varepsilon)n$ edges, for some small $\varepsilon > 0$. This relaxation permits sublinear probe complexity while still yielding meaningful global structures. Several works have since proposed LCAs for sparse spanning subgraphs in different graph families (see, e.g., 
\cite{reut2014local,levi2020local, levi2015quasi, levi2021testing,levi2015constructingnearspanningtrees,levi2024nearly}).
For example, in planar graphs, LCAs with $\poly(1/\varepsilon)$ probe complexity are known \cite{levi2020local};
in bounded-degree expander and well-clusterable graphs, LCAs with $\tilde{O}(\sqrt{n})$ probe complexity exist \cite{levi2024nearly}; 
and in general graphs with maximum degree at most $d$, LCAs with $\tilde{O}(n^{2/3} \cdot \poly(d/\varepsilon))$ probe complexity are available 
\cite{lenzen2018centralized},
with a known lower bound of $\Omega(\sqrt{n})$ \cite{reut2014local}.

\cite{reut2014local} studied LCAs for computing minimum-weight spanning subgraph 
in minor-free graphs, achieving a $(1+\varepsilon)$-approximation to the optimal weight (the weight of a minimum spanning tree), with probe complexity quasi-polynomial in $1/\varepsilon$, $d$, and $W$.

In this work, we identify a natural and well-studied class of graphs - \emph{expander graphs} - that do admit sublinear-time LCAs for constructing spanning trees. This is perhaps surprising, given that prior works have also studied expander graphs (e.g., \cite{levi2020local,levi2024nearly}), yet all of them were only able to design LCAs for sparse spanning subgraphs, rather than full spanning trees. On the other hand, the cycle-versus-path example discussed earlier only rules out sublinear LCAs on graphs with poor expansion (i.e., tree-like structures), but does not apply to well-connected graphs. We achieve nearly optimal probe complexity for our LCA on expander graphs. We further focus on a specific subclass of expander graphs, namely \ER graphs, and design even faster average-case LCAs for constructing spanning trees in these settings, thereby bypassing the worst-case lower bounds that hold for general expanders. Finally, we extend our techniques to obtain an LCA for constructing an \emph{exact} minimum spanning tree (MST) on certain weighted expander graphs, where edge weights are chosen independently and uniformly at random from the set $[W] := \{1, \ldots, W\}$ for some integer $W$.

\subsection{Basic Definitions}
To formally describe our results, we first introduce some basic definitions. We have the following definition of LCA for a graph problem.

\begin{definition}[Local Computation Algorithm]\label{def:LCA}

A \textbf{\emph{Local Computation Algorithm (LCA)}} for a problem $\Pi$ is an oracle algorithm $\mathcal{A}$ that answers problem-specific type of query about a solution $X$ to $\Pi$ on input graph $G$, satisfying following property:
\begin{itemize}
    \item $\mathcal{A}$ use only probe access to $G$, a sequence of random bits $R$, and a local memory to respond any admissible query $q$ to $X$. 
    \item After answering any query $q$, $\mathcal{A}$ erases its local memory entirely, including the query and the response.
    \item Answers to any sequence of queries returned by $\mathcal{A}$ must be consistent with $X$.
\end{itemize}

Let $T_{\mathcal{A}}(G, q)$ denote the expected (over the choice of the random bits $R$) number of probes it takes for the LCA $\mathcal{A}$ to answer query $q$ on input graph $G$, and set $T_{\mathcal{A}}(G) = \max_q T_{\mathcal{A}}(G, q)$. We say the LCA has \emph{(worst-case) probe complexity $T(n)$} if the maximum of $T_{\mathcal{A}}(G)$ over all possible $n$-vertex input graphs $G$ is $T(n)$.

\end{definition}

When specialized for spanning trees, we have the following definition. 

\begin{definition}[LCA for (Minimum) Spanning Tree]
    An algorithm $\mathcal{A}$ is a \textbf{\emph{Local Computation Algorithm for (Minimum) Spanning Tree}} if for input graph $G$, $\mathcal{A}$ is an LCA answering edge-membership query to a subgraph $T$ of $G$, such that $T$ is a (minimum) spanning tree with probability at least $\frac{2}{3}$.
    
    Here ``edge-membership query to $T$'' means that on input $(u,v)\in E(G)$, $\mathcal{A}$ returns whether $(u,v) \in E(T)$.
\end{definition}

Beyond the worst-case LCA, recently, the average-case LCA has also been introduced. These are LCAs whose input graph is from a random graph family.
In

\cite{biswas2025worstcaselocalcomputation}, 
the authors define a new scheme of local computation algorithm called ``average-case LCA'', which assumes that the input graph is drawn from some distribution and asks the LCA to succeed with probability at least $1-\frac{1}{n}$ over a random graph from this distribution.
\begin{definition}[Average-case LCA]\label{def-avg lca}
We say that $\mathcal{A}$ is an \textbf{\emph{Average-Case Local Computation Algorithm}} for a distribution $\mathcal{G}$ over objects of size $n$ for a problem $\Pi$ if, with probability at least $1 - 1/n$ over the random draw $G \sim \mathcal{G}$, the algorithm $\mathcal{A}^G$ (which has probe access to the input $G$) satisfies the requirements of an LCA.

We say that $\mathcal{A}$ has \emph{average-case probe complexity} $T(n)$ if the expected number of probes $T_{\mathcal{A}}(G)$ over $G \sim \mathcal{G}$ is at most $T(n)$. Similarly, $\mathcal{A}$ has \emph{worst-case probe complexity} $T(n)$ if the maximum 
number of probes $T_{\mathcal{A}}(G)$ over $G \sim \mathcal{G}$ is at most $T(n)$.
\end{definition}

\subsection{Our Results}
Now we state our main results. Let $G=(V,E)$ be a graph with $n$ vertices and maximum degree at most $d$ (called a \emph{$d$-bounded} graph), for $d\geq 3$.  Let $S \subseteq V$ be a vertex set.
The \emph{conductance} of $S$ is defined as $\phi_G(S):= \frac{|E(S,V\setminus S)|}{\mu_G(S)}$, where $E(S,V\setminus S)$ is the set of edges between $S$ and $V \setminus S$ and $\mu_G(S) := \sum_{v \in S} deg(v)$ is the total degree of vertices in $S$, respectively.
The conductance of $G$ is defined to be $\phi_G:= \min\limits_{\substack{\emptyset \subsetneq S \subseteq V \\ \mu_G(S)\leq \mu_G(V)/2}} \phi_G(S)$.
We informally say that $G$ is an \emph{expander} when $\phi_G$ is bounded from below by a constant. 
\paragraph{Spanning Trees in Expanders} For graph $G$ with conductance at least $\phi$ and maximum degree at most $d$, we give an LCA that provides local access to a spanning tree of $G$ with low stretch. For $G = (V, E)$, a subgraph $H$ is called a $k$-\emph{spanner} of $G$ if $\dist_H(u,v) \leq k\cdot \dist_G(u,v)$ for any vertices $u,v \in V$, where $\dist_G(u,v)$ (resp. $\dist_H(u,v)$) is the distance between $u$ and $v$ in graph $G$ (resp. $H$). We refer to $k$ as the
\emph{stretch factor}.

\begin{restatable}{theorem}{lcaexpander}\label{thm:d-expander}
    Given adjacency-list probe access
    to a $d$-bounded connected graph $G = (V,E)$ and a lower bound $\phi$ on the conductance of $G$, there is an LCA that answers edge-membership
    query to a subgraph $T=(V,E^\prime)$ of $G$ s.t. $T$ is a spanning tree with probability at least $(1 - \frac{1}{n})$. Additionally, $T$ has depth $O(\frac{\log n}{\phi^2})$, which implies that $T$ is also a $O(\frac{\log n}{\phi^2})$-spanner. The LCA has probe complexity $O(\sqrt{n}(\frac{ \log^2n}{\phi^2}+d))$.
\end{restatable}

The above result can be contrasted with the recent work of Levi et al. \cite{levi2024nearly}, who designed a local computation algorithm (LCA) for expander graphs with conductance at least $ \phi $. Their algorithm maintains a subgraph with $ n + O\left( \frac{\sqrt{n} \log^2 n}{\phi^2} \right) $ edges, achieves probe complexity $ O\left( \sqrt{n} \cdot \left( \frac{\log^2 n}{\phi^2} + d^2 \right) \right) $, and guarantees a stretch factor of $ O\left( \frac{\log n}{\phi^2} \right) $. In comparison, our algorithm maintains a spanning tree (with exactly $n - 1$ edges) and achieves lower probe complexity, thereby yielding a strict improvement over their result. Moreover, we note that the probe complexity of our LCA is nearly optimal for graphs with polylogarithmic maximum degree. In particular, a lower bound of $ \Omega(\sqrt{n}) $ probes is known for LCAs that compute sparse spanning subgraphs in bounded-degree expander graphs~\cite{reut2014local}.

\paragraph{Spanning Trees in \ER Graphs}

Now that we have an LCA for spanning trees with nearly optimal probe complexity on worst-case expanders, a natural question is to ask:

\emph{Can we go beyond worst-case expanders and achieve even better probe complexity?}

To address this question, we adopt the framework developed by Biswas et al.~\cite{biswas2024average} for analyzing \emph{local computation algorithms} over large random objects. In particular, we provide an answer in the setting where the input graph $G$ is drawn from the \ER distribution. In the \ER graph with parameters $n,p$ such that $p\in (0,1]$, a graph $G=([n],E)$ is generated from $G(n,p)$, denoted $G  \sim \Gnp$, by including each edge $(u, v)$ independently with probability $p$ for every distinct pair $u, v \in [n]$.

\begin{restatable}{theorem}{lcaERcombined}
\label{thm:lcaERcombined}
For $np = n^{\delta}$ with any constant $\delta > 0$ and a graph $G \sim \Gnp$, given access to $G$ in the \emph{general graph model}, there exists an average-case LCA with probe complexity\footnote{For simplicity, we focus on the average-case probe complexity. With a suitable cap on the number of probes and minor adaptations of the analysis, the same bound can be achieved in the worst case.} $\tilde{O}(\sqrt{n^{1-\delta}})$ that, with probability at least $1 - \frac{1}{n}$, supports edge-membership queries to a spanning tree of $G$.

\end{restatable}

Note that for any constant $\delta > 0$, the probe complexity $\tilde{O}(\sqrt{n^{1-\delta}})$ strictly improves upon the $\Omega(\sqrt{n})$ lower bound for worst-case instances. As $\boldsymbol{\delta \to 0}$, the random graph $G \sim \Gnp$ approaches the behavior of a bounded-degree expander. In this regime, the probe complexity $\tilde{O}(\sqrt{n^{1-\delta}})$ converges to $\tilde{O}(\sqrt{n})$, nearly matching the complexity of our LCA for bounded-degree expanders established in \Cref{thm:d-expander}.

On the other hand, as $\boldsymbol{\delta \to 1}$, the graph $G \sim \Gnp$ becomes increasingly dense and approaches a complete graph. Intuitively, in this case, much more efficient LCAs are possible. For example, suppose $np = \Omega(n / \log n)$. We can then sample a subset $S \subseteq V$ of size $\Theta(\log^2 n)$ such that, with high probability\footnote{Unless stated otherwise, ``with high probability'' (w.h.p.) means that the probability tends to $1$ as $n \to \infty$.}, the induced subgraph $G[S]$ is connected. A spanning tree $T_0$ of $G[S]$ can be computed. Furthermore, for each remaining vertex $u \in V \setminus S$, the choice of $p$ and the size of $S$ ensure that $u$ has at least one neighbor in $S$ with high probability. We can then connect $u$ to the lexicographically smallest such neighbor $v \in S$ by adding the edge $(u, v)$ to $T_0$. This construction yields a spanning tree of $G$ and can be simulated by an LCA with probe complexity $\poly(\log n)$. This phenomenon is consistent with the behavior of our algorithm, as the probe complexity $\tilde{O}(\sqrt{n^{1-\delta}})$ tends toward $\tilde{O}(1)$ when $\delta \to 1$.

\paragraph{Minimum Spanning Tree}
Having established an LCA for spanning trees on expander graphs, we now extend our techniques to the minimum spanning tree (MST) problem.  
We call a graph $G$ an $(n,d,\lambda)$-graph if $G$ is a $d$-regular graph on $n$ vertices whose adjacency matrix has eigenvalues
\[
d=\lambda_1 \ge \lambda_2 \ge \cdots \ge \lambda_n
\]
and satisfies $\lambda=\max\{\lambda_2,|\lambda_n|\}$.
Building on \cref{thm:d-expander} as a subroutine, we design the following LCA for computing an MST on $(n,d,\lambda)$-graphs.

\begin{restatable}{theorem}{lcaMST}\label{thm:lca mst}
Given adjacency-list probe access to a weighted graph $G=(V,E,\w)$ such that the underlying unweighted graph $\bar G=(V,E)$ is an $(n,d,\lambda)$-graph with $d=\omega_n(1)$ and $\lambda=o(d)$, and assuming that the weight function $\w$ assigns to each edge $e\in E$ an independent weight drawn uniformly from $\{1,2,\ldots,W\}$, where $W \le d/2$ and $W = o(\log n)$, then with high probability there exists an LCA for the minimum spanning tree of $G$ with probe complexity $\tilde O(\sqrt{n}\, d^2)$.
\end{restatable}

We remark that \cref{thm:lca mst} holds for any slowly growing function $d=d(n)$, for example $d=\log n$.
Moreover, since the underlying graph $\bar G$ is $d$-regular and satisfies $\lambda=o(d)$, Cheeger’s inequality implies that $\bar G$ has constant conductance, i.e., $\phi_{\bar G}=\Omega(1)$, and is therefore an expander.

To the best of our knowledge, \cref{thm:lca mst} is the first non-trivial LCA for constructing an \emph{exact} minimum spanning tree on a natural class of weighted graphs, albeit under random edge weights.
The assumption of uniformly distributed weights is made primarily for clarity of presentation and can be relaxed: our result extends to general distributions over $\{1,\ldots,W\}$ that satisfy a suitable \emph{non-criticality} condition, which we formally define and discuss in \cref{sec:beyonduniform}.

We further observe that without the random edge-weight assumption, i.e., when edge weights are assigned adversarially, it is impossible to obtain an LCA with sublinear probe complexity, even when $W=2$ and the underlying graph is an expander (see \cref{sec:lowerboundforMST}).
Finally, note that when $W=1$, the MST problem reduces to the spanning tree problem, for which an $\Omega(\sqrt{n})$ lower bound on the probe complexity of LCAs is known~\cite{reut2014local}.

\subsection{Technical Overview}
Now we provide an overview of the techniques used in the design and analysis of our LCAs.

\subsubsection{Spanning Tree in Expander Graphs}  

Our local algorithm is based on a local implementation of a carefully designed global spanning tree algorithm. The global algorithm leverages both the \emph{spectral properties} of random walks on expanders and the \emph{structural properties} of shortest paths in the graph. A key aspect of our approach is ensuring that the global algorithm can be efficiently simulated locally, which we achieve through a detailed analysis of its behavior. We note that the recent work of \cite{levi2024nearly} on LCAs for sparse spanning subgraphs also relies on a global algorithm guided by random walks. However, in one crucial subroutine, their approach constructs spanning forests using Voronoi cells, whereas our method is fundamentally based on shortest-path structures. 
We first brief our global algorithm on a expander graph $G = (V,E)$.

\paragraph{Global Algorithm Overview}

The global algorithm proceeds in three \texttt{Phases}. It begins by selecting an arbitrary vertex $s$, which then serves as both the \emph{root} of the spanning tree and the \emph{seed} for the random walks.

In \texttt{Phase 1}, starting from vertex $s$, the algorithm initiates $\tilde{\Theta}(\sqrt{n})$ \emph{lazy random walks}, each of length $\tau$, where $\tau$ denotes the \emph{mixing time} of the graph. During these walks, the algorithm faithfully records the full trajectory of each walk, including \emph{every vertex} visited and \emph{every edge} traversed. We denote the the set of vertices by $S$, and the set of edges by $H$. 

Since the vertices in $S$ and edges in $H$ are generated by random walks of length $\tau$, the graph $G_0 = (S, H)$ forms a spanning subgraph of $G$ with diameter at most $2\tau$. Moreover, $G_0$ is fully explicit to the algorithm, allowing it to access any information about $G_0$ without issuing any new probes.

In \texttt{Phase 2}, the algorithm performs a BFS from the root $s$ in the graph $G_0$ to construct a \emph{BFS tree} of $G_0$. Since $G_0$ is fully explicit, this phase does not require any probes. We denote the resulting tree by $T_0 = (S,E^\prime)$, and refer to it as the \emph{core tree}. Since the diameter of $G_0$ is at most $2\tau$, the diameter of the core tree $T_0$ is also bounded by $2\tau$.

In \texttt{Phase 3}, we connect each vertex $u \notin S$ to the core tree $T_0$, where all vertices in $S$ already lie on a spanning tree rooted at $s$. For each such vertex $u$, the algorithm identifies its \emph{anchor} -- the closest vertex to $u$ in $S$. \textbf{A key distinction from prior work \cite{levi2024nearly}} lies in how ties are broken when multiple candidate vertices in $S$ are equally close to $u$: rather than selecting the anchor based on vertex ID, we break ties using the \emph{lexicographical order} of the paths from $u$ to the candidates in $S$.
In other words, the anchor of $u$ is the vertex in $S$ for which the shortest path from $u$ is lexicographically minimal among all such paths. This structural property of the anchor plays a crucial role in improving both the probe complexity and the total number of edges in the resulting subgraph.

\vspace{1cm}
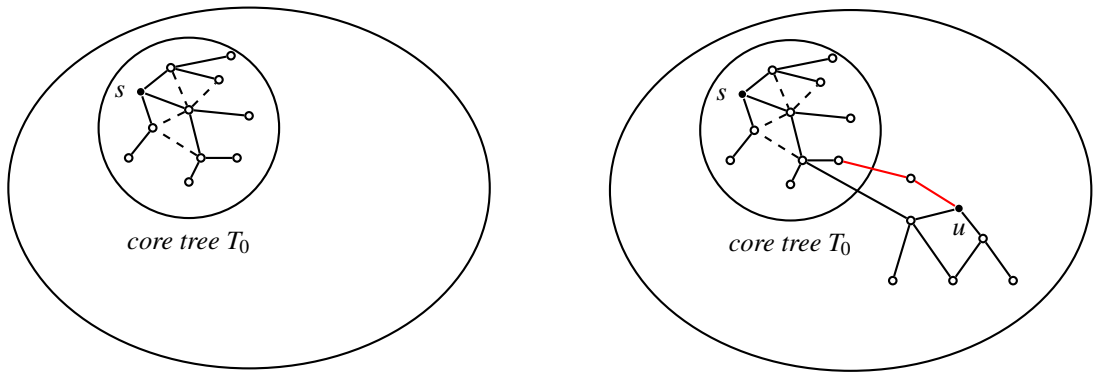
\begin{minipage}[t]{0.48\textwidth}
\centering

\begin{tikzpicture}[scale=0.8, every node/.style={font=\small}, thick]
\draw[thick] (0,0) ellipse (4cm and 3cm);
\coordinate (corecenter) at (-1,1);
\draw[thick] (corecenter) circle (1.5);
\node at ($(corecenter)+(0,-1.9)$) {\textit{core tree $T_0$}};
\node[circle, fill=black, inner sep=1pt, label=left:$s$] (s) at ($(corecenter)+(-0.8,0.6)$) {};

\node[circle, draw, inner sep=1pt] (a) at ($(corecenter)+(-0.3,1.0)$) {};
\node[circle, draw, inner sep=1pt] (b) at ($(corecenter)+(0.0,0.3)$) {};
\node[circle, draw, inner sep=1pt] (c) at ($(corecenter)+(-0.6,0.0)$) {};
\node[circle, draw, inner sep=1pt] (d) at ($(corecenter)+(0.5,0.8)$) {};
\node[circle, draw, inner sep=1pt] (e) at ($(corecenter)+(1,0.2)$) {};
\node[circle, draw, inner sep=1pt] (f) at ($(corecenter)+(-1,-0.5)$) {};
\node[circle, draw, inner sep=1pt] (g) at ($(corecenter)+(0.8,-0.5)$) {};
\node[circle, draw, inner sep=1pt] (h) at ($(corecenter)+(0.2,-0.5)$) {};
\node[circle, draw, inner sep=1pt] (i) at ($(corecenter)+(0.7,1.2)$) {};
\node[circle, draw, inner sep=1pt] (j) at ($(corecenter)+(0,-0.9)$) {};

\draw (s) -- (a);
\draw (a) -- (d);
\draw (s) -- (b);
\draw (s) -- (c);
\draw[dashed] (b) -- (d);
\draw (b) -- (e);
\draw (c) -- (f);
\draw[dashed] (a) -- (b);
\draw (a) -- (i);
\draw[dashed] (b) -- (c);
\draw (b) -- (h);
\draw[dashed] (h) -- (c);
\draw (h) -- (g);
\draw (h) -- (j);
\node at (0, -4.5) {\textbf{Phase $1 \And 2$: Refine random walks to get core tree}};
\end{tikzpicture}
\end{minipage}
\hfill
\begin{minipage}[t]{0.48\textwidth}
\centering

\begin{tikzpicture}[scale=0.8, every node/.style={font=\small}, thick]
\draw[thick] (0,0) ellipse (4cm and 3cm);
\coordinate (corecenter) at (-1,1);
\draw[thick] (corecenter) circle (1.5);
\node at ($(corecenter)+(0,-1.9)$) {\textit{core tree $T_0$}};
\node[circle, fill=black, inner sep=1pt, label=left:$s$] (s) at ($(corecenter)+(-0.8,0.6)$) {};

\node[circle, draw, inner sep=1pt] (a) at ($(corecenter)+(-0.3,1.0)$) {};
\node[circle, draw, inner sep=1pt] (b) at ($(corecenter)+(0.0,0.3)$) {};
\node[circle, draw, inner sep=1pt] (c) at ($(corecenter)+(-0.6,0.0)$) {};
\node[circle, draw, inner sep=1pt] (d) at ($(corecenter)+(0.5,0.8)$) {};
\node[circle, draw, inner sep=1pt] (e) at ($(corecenter)+(1,0.2)$) {};
\node[circle, draw, inner sep=1pt] (f) at ($(corecenter)+(-1,-0.5)$) {};
\node[circle, draw, inner sep=1pt] (g) at ($(corecenter)+(0.8,-0.5)$) {};
\node[circle, draw, inner sep=1pt] (h) at ($(corecenter)+(0.2,-0.5)$) {};
\node[circle, draw, inner sep=1pt] (i) at ($(corecenter)+(0.7,1.2)$) {};
\node[circle, draw, inner sep=1pt] (j) at ($(corecenter)+(0,-0.9)$) {};

\draw (s) -- (a);
\draw (a) -- (d);
\draw (s) -- (b);
\draw (s) -- (c);
\draw[dashed] (b) -- (d);
\draw (b) -- (e);
\draw (c) -- (f);
\draw[dashed] (a) -- (b);
\draw (a) -- (i);
\draw[dashed] (b) -- (c);
\draw (b) -- (h);
\draw[dashed] (h) -- (c);
\draw (h) -- (g);
\draw (h) -- (j);

\node[circle, fill=black, inner sep=1pt, label=below:$u$] (u) at ($(corecenter)+(2.8,-1.3)$) {};
\node[circle, draw, inner sep=1pt] (1) at ($(corecenter)+(2,-0.8)$) {};
\node[circle, draw, inner sep=1pt] (2) at ($(corecenter)+(2,-1.5)$) {};
\node[circle, draw, inner sep=1pt] (3) at ($(corecenter)+(3.2,-1.8)$) {};
\node[circle, draw, inner sep=1pt] (4) at ($(corecenter)+(2.7,-2.5)$) {};
\node[circle, draw, inner sep=1pt] (5) at ($(corecenter)+(1.7,-2.5)$) {};
\node[circle, draw, inner sep=1pt] (6) at ($(corecenter)+(3.7,-2.5)$) {};

\draw[red] (u) -- (1);
\draw (u) -- (2);
\draw (u) -- (3);
\draw[red] (1) -- (g);
\draw (2) -- (h);
\draw (2) -- (4);
\draw (3) -- (4);
\draw (2) -- (5);
\draw (3) -- (6);
\node at (0, -4.5) {\textbf{Phase 3: BFS to find anchor}};
\end{tikzpicture}
\end{minipage}

Specifically, the algorithm iterates over every vertex $u \in V \setminus S$. For each such vertex $u$, the algorithm performs a BFS\footnote{This is implemented as the subroutine \textsc{FindPath} in our algorithm. Moreover, \textsc{FindPath} actually performs \emph{lexicographical BFS} (see \Cref{sec-lexico bfs} for details).} starting from $u$ to discover the path to its anchor. Once the path is found, algorithm updates the edge set $E'$ by including all edges along the path (while keep the set $S$ invariant). After all iterations are complete, $E'$ contains a path from the root $s$ to every vertex $u \in V$, where each path consists of two segments:
\begin{enumerate}
    \item a path from $s$ to the anchor of $u$ along the core tree $T_0$;
    \item a path from the anchor to $u$ discovered during the iteration over $u$.
\end{enumerate}
Therefore, after \texttt{Phase 3}, the subgraph $T = (V, E')$ forms a connected spanning subgraph of $G$. Moreover, the first segment has length at most $2\tau$ as discussed before, and the second segment has length at most the diameter of the graph, which is also bounded by $\tau$ (shown in the proof).

Having established that $T = (V, E')$ is connected, we need to show that $T$ is in fact a tree. This follows from the structural properties of the anchor assignment, which ensure that no cycles are introduced during \texttt{Phase 3}. Additionally, we analyze the probe complexity of each BFS used to discover anchor paths, leveraging both the spectral properties of random walks and the expansion characteristics of the graph $G$.

\paragraph{Local Algorithm Implementation}
To locally determine whether an edge $(u,v)$ belongs to $E'$, we first establish a key property of the shortest path between a vertex $u \in V \setminus S$ and its anchor. Specifically, an edge $(u,v)$ is in $E'$ \textbf{if and only if} it lies on the path from $u$ to its anchor or on the path from $v$ to its anchor. 

This property enables a local algorithm to answer edge-membership queries by simulating the first two phases of the global algorithm—random walk sampling and core tree construction—and then performing BFS from both $u$ and $v$ to identify the corresponding anchor paths. 

For consistency, in \texttt{Phase 1}, the local algorithm always uses the \emph{same random bits} for random walks across all queries, ensuring that the locally reconstructed subgraph is consistent with the global construction.

\vspace{0.5em}
\noindent\textbf{\texttt{Edge-Membership Query: Is $(u,v) \in E'$?}}

\begin{itemize}
    \item \emph{Step 1: Simulate Random Walks.} Run the same random walks as performed in the global algorithm.
    
    \item \emph{Step 2: Construct Core Tree.} Extract the core tree $T_0 = (S, E')$ from the random walk trajectories.
    \begin{itemize}
        \item If $(u,v) \in E'$, return \textsc{Yes}.
    \end{itemize}
    
    \item \emph{Step 3: Identify Anchor Paths.} Otherwise, perform BFS from both $u$ and $v$ to find their respective paths to their anchors.
    \begin{itemize}
        \item If edge $(u,v)$ lies on either of these two paths, return \textsc{Yes}; otherwise, return \textsc{No}.
    \end{itemize}
\end{itemize}

\subsubsection{Spanning Tree in Random Graphs}
To break the $\Omega(\sqrt{n})$ lower bound for graphs drawn from the \ER distribution, we leverage the tools developed previously along with a key property of $G(n,p)$: each pair of vertices forms an edge independently with probability $p = n^{\delta-1}$, where $\delta > 0$ is constant.

A simple \emph{initial observation} is as follows. For a pair of non-adjacent vertices $u$ and $v$ in $G \sim G(n,p)$, independent BFS explorations from $u$ and $v$ are unlikely to intersect until each has visited roughly $\sqrt{n^{1-\delta}}$ vertices. Indeed, for two disjoint sets $V_1,V_2$ of size $o(\sqrt{n^{1-\delta}})$, the number of potential edges between them is $o(n^{1-\delta})$, and each edge exists independently with probability $p$, giving an expected number of edges $o(1)$ between the BFS frontiers.

\paragraph{Beyond the $\sqrt{n}$ Barrier} 
This observation motivates a new approach that achieves probe complexity $R = \tilde{\Theta}(\sqrt{n^{1-\delta}})$. We first perform $R$ independent random walks to construct a small core tree $T_0$ spanning a vertex set $S$. Then, for each vertex, we use at most $R$ probes to attempt to find a path to an anchor in $S$. Vertices that succeed are called \emph{good}.

Vertices that fail to find their anchors within this limited BFS are called \emph{bad}. For each bad vertex $u$, we exploit properties of the \ER distribution to define a recovery set $IP(u)$ of size up to $n^\delta$, corresponding to a long path in the subgraph. A key property (\Cref{lem-gnp bfs}) is that BFS explorations from vertices in $IP(u)$ rarely intersect, ensuring that the combined BFS from $IP(u)$ covers up to $\sqrt{n^{1-\delta}} \cdot n^\delta = \sqrt{n^{1+\delta}}$ distinct vertices. By properties of random walks, this guarantees that $IP(u)$ contains at least one good vertex, connecting every bad vertex to the core tree. Since the algorithm maintains exactly $n-1$ edges, the resulting subgraph is connected and acyclic, i.e., a spanning tree.

However, if the algorithm uses only $\tilde{\Theta}(\sqrt{n^{1-\delta}})$ probes, the guarantee that $IP(u)$ reaches size $n^\delta$ holds only when $\delta \le 1/3$. Accordingly, the analysis is divided into two cases:

\begin{itemize}
    \item For $\delta \le 1/3$, the procedure above achieves probe complexity $\tilde{\Theta}(\sqrt{n^{1-\delta}})$.
    \item For $\delta > 1/3$, each vertex $u$ is connected to a restricted set $RIP(u)$, smaller than $IP(u)$. Using vertex-pair probes in the general graph model (see \Cref{sec-preliminary}) and the fact that each pair forms an edge with probability $p$, we still achieve probe complexity $\tilde{\Theta}(\sqrt{n^{1-\delta}})$.
\end{itemize}

\paragraph{Proof Techniques}  
To analyze the algorithm rigorously, we introduce functions that characterize the structure of the sets $\{IP(u)\}_{u \in V}$ and $\{RIP(u)\}_{u \in V}$ (see \Cref{def-bc1,def-bc2}). This function-based framework for average-case LCA analysis was proposed by Biswas et al.~\cite{biswas2024average}. However, it is highly problem-specific: different algorithms require custom function definitions. Consequently, our definitions in \Cref{def-bc1,def-bc2} differ significantly from those in~\cite{biswas2024average}.

\subsubsection{Minimum Spanning Tree}

Moving from unweighted spanning trees to minimum spanning trees introduces new challenges, as the structure of the output is no longer determined solely by connectivity, but also by the relative ordering of edge weights. We focus on expander graphs with integer edge weights drawn from $\{1,2,\dots,W\}$. As shown in \Cref{sec:lowerboundforMST}, designing an LCA with sublinear probe complexity for the MST problem on expander graphs with arbitrarily assigned edge weights is impossible, even when $W=2$. 

Motivated by this impossibility result, we instead consider an average-case setting in which the weight of each edge is sampled independently and uniformly from $\{1,2,\dots,W\}$. In this setting, we show that it is indeed possible to design a local computation algorithm with sublinear probe complexity for computing an exact minimum spanning tree, as formalized in \cref{thm:lca mst}.

Here we use an observation that goes back to Kruskal’s algorithm and was exploited by Chazelle, Rubinfeld, and Trevisan in their sublinear MST weight estimation work~\cite{chazelle2005approximating}. Consider a connected graph $G$ with edge weights in $\{1,2,\dots,W\}$ for some integer $W \ge 2$. Let $G_{\le i}$ denote the subgraph of $G$ consisting of all edges of weight at most $i$.

The MST of $G$ can be constructed incrementally. First, compute a minimum spanning forest (MSF) $F_{\le 1}$ of $G_{\le 1}$, which consists of an MST for each connected component of $G_{\le 1}$. Next, add a suitable set of weight-2 edges $E_2$ to obtain
\[
F_{\le 2} = F_{\le 1} \cup E_2 ,
\]
an MSF of $G_{\le 2}$. Continuing inductively, given an MSF $F_{\le i}$ of $G_{\le i}$, we add a set of weight-$(i+1)$ edges $E_{i+1}$ to form
\[
F_{\le i+1} = F_{\le i} \cup E_{i+1},
\]
which is an MSF of $G_{\le i+1}$ and contains the MST of each connected component of $G_{\le i+1}$. Since $G$ is connected, at the final step the resulting forest $F_{\le W}$ is a minimum spanning tree of $G$.

The key technical difficulty here 
is to ensure \emph{consistency across different weight layers}, since the LCA must answer membership queries for edges without explicitly constructing forests \(F_{\le i}\). Unlike in the unweighted setting, the decision for an edge of weight \(i\) depends on the connectivity structure induced by all lower-weight edges, and directly simulating the MSF construction on \(G_{\le i}\) may incur superlinear probe complexity.

To overcome this issue, we rely on two ideas. First, we exploit structural properties induced by the random weight distribution. This weakens dependencies between decisions for edges across different weight layers and allows us to determine membership using only local information. Second, within each local weight layer, we use a spanning tree of the underlying unweighted graph as a global reference structure. The tree encodes both connectivity and acyclicity of the graph, which we leverage to help determine whether a queried edge should be retained.

\paragraph{Hierarchy by Random Weights}

Given a $d$-regular expander graph, if each edge is sampled independently with probability $p$, then as long as $p$ avoids a narrow critical interval (e.g., $(\frac{1-\varepsilon}{d}, \frac{1+\varepsilon}{d})$ for sufficiently small constant $\eps$), the resulting subgraph $G_p$ satisfies one of the following: either all connected components are small, of size $O_\varepsilon(\log n)$, or there is a single giant connected component of size $\Omega_\varepsilon(n)$ with fast mixing time $O(\log^2 n)$, while the remaining components are small, of size $O_\varepsilon(\log n)$.

This phenomenon yields a layered view of the graph as edges are revealed in increasing order of weights. Starting from \(G_{\le 1}\), with high probability the graph contains a unique giant component \(L_{\le 1}\), while all remaining components are small. As the threshold increases, the giant component gradually grows by absorbing smaller components, and each \(G_{\le i}\) continues to consist of a single giant component \(L_{\le i}\) together with only small components. We illustrate the hierarchy in following figure, and note that $L_{\leq W} = G$.

\begin{figure}[h]\label{fig:hrk}
\centering

\begin{tikzpicture}[scale=0.8, every node/.style={font=\small}, thick]

\draw[thick] (0,0) circle (1);
\node at (0,0) {$L_{\le 1}$};

\draw (2.0,0.8) circle (0.25);
\draw (1.9,-0.7) circle (0.25);
\draw (1.6,0.1) circle (0.25);
\draw (1.2,-1) circle (0.25);
\draw (1.2,1) circle (0.25);

\draw[thick] (0.8,0) circle (2);
\node at (2.6,2) {$L_{\le 2}$};

\node at (3.8,0) {$\cdots$};

\draw[thick] (2.3,0) ellipse (3.6 and 2.8);
\node at (4.5,2.8) {$L_{\le W}$};

\node at (2.3, -3.5) {\textbf{Hierarchy across giant components}};
\end{tikzpicture}

\end{figure}

We discuss this formally in \cref{sec:hierarchy}. Intuitively, for small components, we can explore the entire component with few queries and then use the full information to compute its MST or extract edges to add toward a larger MST. For the giant component $L_{\leq 1}$, since it only contains weight-$1$ edges, we can apply a random-walk–based approach, as previously described, to determine whether any given edge belongs to its MST.

Now the problem is reduced to, for each $i \in [W-1]$, select appropriate edges to connect small components between $L_{\leq i}$ and $L_{\leq i+1}$, 
i.e., those components of $G_{\le i}$ that are absorbed into $L_{\le i+1}$ when edges of weight $i+1$ are added. 
Note that every such component only contains edges with weight at most $i$. 
To this end, we apply the second idea, namely leveraging an unweighted spanning tree on $L_{\leq i+1}$, which allows us to iteratively construct a minimum spanning forest of $G_{\leq i+1}$.

\paragraph{Unweighted Spanning Tree for Reference}
Let us focus on the small components that lie between $L_{\leq 1}$ and $L_{\leq 2}$ for illustration. 
For each such small component, say $C$, our goal is to identify a weight-$2$ edge incident to $C$, such that the collection of all selected weight-$2$ edges connects these small components to $L_{\leq 1}$ without creating cycles.

To identify such edges, we use an (unweighted) spanning tree of $L_{\leq 2}$ as a reference. For intuition, fix any (unweighted) spanning tree $T$ of $L_{\leq 2}$ and contract each component $C$ into a super-vertex. The resulting graph remains connected. Moreover, every edge between super-vertices must have weight $2$ in $G$; otherwise, the corresponding super-vertices would already be merged.

This implies that the edges of $T$ provide useful structural information. For each component $C$, we select a weight-$2$ edge incident to $C$ that appears in $T$. To guarantee acyclicity, we define a partial ordering, called \emph{rank} (\cref{def-rank}), over vertices within each small component with respect to $T$. Importantly, $T$ is locally accessible using \cref{thm:d-expander}, provided that we can support adjacency-list probes to $L_{\leq 2}$ using the LCA constructed for expander graphs. Combining these ideas, we obtain \cref{alg-probbc} for identifying such edges.

\medskip
\noindent
The final LCA (\cref{alg-in mst}) integrates the above procedures to determine whether a queried edge should be included in the minimum spanning tree.

\subsection{Other Related Work}
\paragraph{LCA for Sparse Spanning Subgraphs (Spanners).}
We review the literature on Local Computation Algorithm (LCA) for two closely related problems: sparse spanning subgraphs and graph spanners. Both aim to consistently provide local access to a sparsified version of a graph while preserving certain structural properties. LCAs for sparse spanning subgraphs primarily focus on ensuring connectivity with as few edges as possible. In contrast, LCAs for spanners aim to preserve approximate pairwise distances, where a $ k $-\emph{spanner} means that it preserves all pairwise distances within a multiplicative factor $ k \geq 1 $. Since most existing works trades off between sparsity and stretch, we discuss both lines of research together.

A series of work ~\cite{reut2014local, levi2015quasi, levi2021testing} consider LCA for sparse spanning subgraph in minor-free graphs. Levi et al.~\cite{levi2015constructingnearspanningtrees} designed an LCA for hyperfinite graphs with probe complexity independent of $ n $, though super-exponential in $ d $ and $ 1/\varepsilon $. Lenzen and Levi~\cite{lenzen2018centralized} gave an LCA for $d$-bounded graphs with probe complexity $ O(n^{2/3} \cdot \poly(d/\varepsilon)) $ and stretch $ O\left( \frac{\log n \cdot (d + \log n)}{\varepsilon} \right) $. Bodwin et al.~\cite{bodwin2023} gave an adjacency oracle for a spanning subgraph with $(1 + \varepsilon)n$ edges. This model allows centralized preprocessing and answers each adjacency probe in $\tilde{O}(1)$ time. With total preprocessing time $\tilde{O}(n/\varepsilon)$, their algorithm implies an $\tilde{O}(n/\varepsilon)$-time LCA: one can construct an adjacency oracle for each query and erasing the memory after that. Levi et al.~\cite{levi2024nearly} also extended their LCA to well-clusterable graphs. Suppose $G$ is a graph that can be partitioned into a constant number of vertex sets each of which inducing an expander and has conductance at most $\phi_{\text{out}}$ in $G$. Their algorithm finds a spanning subgraph with $(1+\eps)n$ edges, achieves probe complexity $ O(\sqrt{n} + \phi_{\text{out}} n) $ and stretch $ \Theta(\log n) $.

For spanner with constant stretch, Parter et al.~\cite{parter2019local} and Arviv et al.~\cite{arviv_et_al:LIPIcs.APPROX/RANDOM.2023.42} proposed LCAs that construct $ (2k - 1) $-spanners with $ O(n^{1+1/k}) $ edges for small $ k \in \{2, 3\} $, and more generally, $ O(k^2) $-spanners with $ O(n^{1+1/k}) $ edges for arbitrary $ k \geq 1 $. In particular, Arviv et al.~\cite{arviv_et_al:LIPIcs.APPROX/RANDOM.2023.42} improved the probe complexity to $ O(n^{1 - 1/k}) $ for the $ (2k - 1) $-spanners, and to $ O(n^{2/3} d^2) $ for the $ O(k^2) $-spanners.

\paragraph{LCA for Other Problems.} 
LCAs have also been extensively studied from various perspectives. In particular, for several classical problems including maximal independent set (MIS)~\cite{ghaffari2016improved, LRY17, ghaffari2019sparsifying, Gha22},  
maximal matching~\cite{hassidim2009local, yoshida2009improved, mansour2013local, LRY17, BRR23},  
and vertex coloring~\cite{even2014deterministic, feige_et_al:LIPIcs.ICALP.2018.50, czumaj2018sublinear, chang2019complexity}.  
Recent works have also established lower bounds for LCAs~\cite{behnezhad2023local, azarmehr2025lower}.  
Despite the above, there are also interesting applications that leverage LCAs. For example,~\cite{lange2022properly, lange2025local, lange2025agnostic, lange2025robust} employ LCAs to promote the development of learning theory, while~\cite{azarmehr2025stochastic} applies LCAs to analyze an algorithm in the context of the Stochastic Matching problem.

\paragraph{Average-case LCA.}

Previously, local algorithms were studied by \cite{brautbar2010local,borgs2012power,frieze2017looking} in the preferential attachment model, a well-known random graph model characterizing real-world networks. The algorithm they studied is called ``local information algorithms'', a restricted-version of LCA. More recently, Biswas et al.~\cite{biswas2024average} initiated and formalized the study of \emph{average-case} LCAs, focusing on the construction of $k$-spanners (and sparse spanning subgraph) when the input graph is drawn from various random graph models, including Preferential Attachment and \ER graphs with certain parameters. The probe complexity of their algorithms depends on the characteristics of the specific model. In particular, for the Preferential Attachment model, they design an LCA that provides local access to an $O(\log n)$-stretch spanning tree using $O(\mu \sqrt{n})$ probes in the worst case and $O(\mu \log^3 n)$ probes in expectation over random queries, where $\mu$ denotes the expected degree. For Erdős–Rényi graphs $G \sim \mathcal{G}(n, p)$, when $np = n^\delta$, a $(2/\delta + 5)$-stretch spanner with $n + o(n)$ edges can be accessed using $O\left(\min(n^\delta, n^{1 - \delta} \log n)\right)$ probes; for $p \geq 7\log n / n$, another algorithm constructs a sparse spanning subgraph with $n + o(n)$ edges using $\tilde{O}(\Delta)$ probes, where $\Delta = np$ is the expected degree.

\paragraph{Phase Transition of Random Subgraphs.}
Percolation theory, initiated by Broadbent and Hammersley \cite{broadbent1957percolation} in 1957, studies probabilistic models of random subgraph: Given a \emph{base graph} $G$, the \emph{percolated subgraph} $G_p$ is obtained by retaining each edge of $G$ independently with probability $p$.
In this line of work, one central topic is to understand the \emph{phase transition behavior} of $G_p$ as the percolation probability $p$ increases.
Typically, there exists a critical threshold separating a \emph{subcritical regime}, in which all connected components have size $O(\log n)$ with high probability, from a \emph{supercritical regime}, in which a unique giant connected component of size $\Omega(n)$ emerges.
For the Erd\H{o}s--R\'enyi model $G(n,p)$, this threshold is $p = 1/n$ \cite{erd6s1960evolution}.
Similar phenomena have been established for other base graphs such as the hypercube \cite{ajtai1982largest,bollobas1992evolution}, general graphs under degree and spectral assumptions \cite{chung2009percolation}, and in particular for $d$-regular expanders, \cite{frieze2004emergence, diskin2024expansion} characterized the emergence and expansion of the giant component.

\paragraph{MSTs on Graphs with Randomly Assigned Edge Weights.} The study of MSTs on graphs with randomly assigned edge weights is a classical topic in theoretical computer science. Many previous works have approached this problem from various perspectives, particularly focusing on estimating the weight of the MST when the edge lengths are drawn independently from identical distributions~\cite{frieze1985value,beveridge1998random,steele2002minimal,frieze2018edge,babson2024models}.

\section{Preliminaries}\label{sec-preliminary}
In this paper, we use $[k]$ to denote set $\{1,2,\cdots,k\}$ for positive integer $k$. For input graph $G=(V,E)$, throughout this paper we assume $V=[n]$. For $u \in V$, we use $\Gamma_G(u)$ to denote the neighborhood of $u$ in $G$, i.e. $\Gamma_G(u):=\{v\in V| (u,v) \in E\}$. We sometimes omit the subscript if it causes no confusion.

We will use the following definitions of paths, path orderings and lexicographically-least shortest path. 
\begin{definition}\label{def-small}
    For a set of vertices $S$, $\smallest{S}$ (resp. $\largest{S}$) denotes the vertex with smallest (resp. largest) ID in set $S$.
\end{definition}

\begin{definition}[path]
For graph $G=(V=[n],E)$, we call a sequence of vertex $P = (u_1,u_2,\cdots,u_k)$ a \textbf{path} starting at $u_1$ and ending at $u_k$, if $u_i \in V$ for $i \in [k]$ and $(u_i,u_{i+1}) \in E$ for $i \in [k-1]$. We use following notations: (1) $V(P)$: vertex set induced by $P$, i.e. $\{u_1,u_2,\cdots,u_k\}$; (2) $E(P)$: edge set induced by $P$, i.e. $\{(u_1,u_2),(u_2,u_3)\cdots,(u_{k-1},u_k)\}$; (3) $\ell(P)$: length of $P$, i.e. $|E(P)|$; (4) $P^{(i)}$: the $i$-th vertex in $P$, i.e. $u_i$.
\end{definition}

\begin{definition}[path ordering]\label{def-comp}
    For two paths $P_1$ and $P_2$ starting from same vertex, we say \textbf{$P_1 \prec P_2$} if

        $\ell(P_1) < \ell(P_2)$; or

        $\ell(P_1) = \ell(P_2)$, and from starting vertex to ending vertex, $P_1$ has smaller lexicographical order than $P_2$.
    
\end{definition}

\begin{definition}[lexicographically-least shortest path]
    Path $(u,w_1,w_2,\cdots,v)$ is the \textbf{lexicographically-least shortest path} from $u$ to $v$ if for any other path $(u,w^\prime_1,w^\prime_2,\cdots,v)$ from $u$ to $v$, it holds that
    \[
    (u,w_1,w_2,\cdots,v) \prec (u,w^\prime_1,w^\prime_2,\cdots,v).
    \]
    We use $P_G(u,v)$ to denote the lexicographically-least shortest path from $u$ to $v$ in graph $G$. For convenience, let $P_G(u,u) = (u)$.
    
\end{definition}
Specifically, note that for any vertex $u\in V$, $V(P_G(u,u))=\{u\}$ and $E(P_G(u,u))=\emptyset$.

\paragraph{\ER graph}
We formally define the \ER graph.

\begin{definition}[\ER graphs]\label{def:ER}
    For a function $p=p(n)$, a graph $G\leftarrow{} \Gnp$ is an \textbf{\ER random graph} if for every pair $(u,v) \in [n] \times[n]$ with $u \neq v$, edge $(u,v)$ is added to $E(G)$ with probability $p$ independently.
\end{definition}

\subsection{Access Models}
We will consider the following access models. 

\textbf{Adjacency List Model.}  
In the adjacency list model, a local computation algorithm (LCA) accesses the graph $G = ([n], E)$ using two types of probes. The probe \texttt{Deg}$(v)$ returns the degree of a vertex $v$. The probe \texttt{Nbr}$(v, i)$ returns the $i^\text{th}$ neighbor of vertex $v$ according to its adjacency list if $i \le \texttt{Deg}(v)$, and returns $\bot$ otherwise. 
{For weighted graph, \texttt{Nbr}$(v, i)$ also return weight of corresponding edge.}

\textbf{General Graph Model.}  
In the general graph model, the LCA has access to the same \texttt{Deg}$(v)$ and \texttt{Nbr}$(v, i)$ probes as in the adjacency list model, and is additionally allowed to perform vertex-pair probes. Specifically, a vertex-pair probe \texttt{Exists}$(u, v)$ returns \texttt{true} if the edge $(u, v) \in E$, and \texttt{false} otherwise. This model provides more flexibility by enabling direct probes about whether a specific edge is present in the graph.

\subsection{Lexicographical Breadth-First Search}\label{sec-lexico bfs}
Throughout this paper, we frequently use \textbf{lexicographical BFS} to find lexicographically-least shortest path from a given vertex $u$, to a given set of vertices $S$, which is performed as follows:

At the $t$-th level of the BFS from $u$, we probe the neighbors of vertices discovered at level $t-1$. Among all these vertices discovered at level $t-1$, we process vertices in the increasing order of their IDs. That is to say, we first probe neighbors of the vertex with the smallest ID, then the second smallest, and so on. 

Through this way, at the moment the BFS from $u$ finds a vertex $v \notin S$ incident to $S$, $v$ must be the lexicographically first vertex incident to $S$ at distance $t$ from $u$. This means, the lexicographically-least shortest path from $u$ to $S$, must pass through $v$ to some vertex in $S$. Hence the lexicographically-least shortest path from $u$ to $S$, is a conjunction of the BFS path from $u$ to $v$ and the edge connecting $v$ to the first neighbor of $v$ incident to $S$.

\subsection{Conductance, Expander graphs and Random walks}

For graph $G=(V,E)$ with degree at most $d$, we modify graph $G$ to get $d$-regular graph $G_{reg}$ by adding half-weighted self-loops to each vertex in $G$ so that each vertex has degree $d$. In this paper, instead of directly performing random walk on $G$, we perform \textit{lazy random walk} on $G_{reg}$. A lazy random walk is a random walk which stays at current vertex w.p. $\frac{1}{2}$ at each step. Note that this is equivalent to a random walk on $G$ which stays at the current vertex $v$ w.p. $\frac{2d - deg(v)}{2d}$ and moving to each neighbor w.p. $\frac{1}{2d}$. Each step of such random walk can be implemented by uniformly sampling an integer $i$ from $[2d]$ and making a probe \Nbr$(v, i)$.

Using the notion of graph conductance, we have following corollary for lazy random walk on $G_{reg}$.

\begin{definition}[conductance]
    Let $G=(V,E)$ be a graph and $S \subseteq V$ be a vertex set. Denote $\mu_G(S) := \sum_{v \in S} deg(v)$ and $\phi_G(S):= \frac{|E(S,V\setminus S)|}{\mu_G(S)}$, where $E(A,B)$ denotes the set of edges with one vertex in $A$ and the other in $B$.

    The \textbf{conductance} of $G$ is defined by $\phi_G:= \min\limits_{\substack{\emptyset \subsetneq S \subseteq V \\ \mu_G(S)\leq \mu_G(V)/2}} \phi_G(S)$.
\end{definition}

We call a graph $G$ an expander if $\phi_G\geq \phi$ for some universal constant $\phi>0$.

\begin{corollary}[\cite{levi2024nearly}]\label{cor-lazy reg mixing time}
    Let $G=(V,E)$ be a connected $d$-bounded degree graph on $n$ vertices and let $v\in V$. If we perform a lazy random walk in $G_{reg}$ staring from $v$, of length at least $\tau(G_{reg}) := \frac{c\log (n)}{\phi^2(G)}$ for some large constant $c$, then the probability this walk ends at $u$ is at least $\left( \frac{1}{2 n} \right)$ for every $u \in V$.
\end{corollary}

\section{LCA for Spanning Tree on Expanders}\label{sec-expander}

In the following, we first present a global algorithm for constructing a spanning tree in expander graphs. Then we show how to locally implement it to obtain the LCA and prove \Cref{thm:d-expander}.

\subsection{A Global Algorithm for Constructing Spanning Trees}
The global algorithm (which runs in polynomial time) works in three phases. In \texttt{Phase 1}, it starts by sampling a vertex $s$ as a  seed and then performs $O(\sqrt{n})$ independent length-$\tau$ lazy random walks from $s$ on graph $G_{\reg}$, where $\tau \geq \tau(G_{\reg})$ is as defined in \Cref{cor-lazy reg mixing time}, which is large enough for the random walk to mix.

We denote by $S$ all the vertices met, and by $H$ the set of edges (discarding self-loops) seen  during random walks. Note that for every vertex $v \in S$, there is a path from $s$ to $v$ with length at most $\tau$, only using edges in $H$. Let $G_0 := (S,H)$ be the graph with vertex set $S$ and edge set $H$.

Then in \texttt{Phase 2}, the algorithm finds the spanning tree, denoted as $T_0=(S, H^\prime)$ of $G_0$. This can be done by performing BFS from $s$ in $G_0$. Note that the tree can be viewed as a tree rooted at $s$ of depth $O(\tau)$. In the following, we will call the tree $T_0=(S, H^\prime)$ a \textbf{core tree}. 

We have the following definition will be used in \texttt{Phase 3}.

\begin{definition}[anchor]\label{def-anc}
    Let $G=(V,E)$ and $S$ be a subset of $V$. A vertex $v \in S$ is said to be the \textbf{anchor} of $u \in V\setminus S$ if for every vertex $w\in S\setminus \{v\}$, $P_G(u,v) \prec P_G(u,w)$. Informally speaking, $P_G(u,v)$ is the the lexicographically-least shortest path among all paths from $u$ to $S$, and $v$ is the vertex in $S$ that is reached by $P_G(u,v)$.

    For $u\in V\setminus S$, we use $\anc_{S}(u)$ to denote the anchor of $u$ in the set $S$. For convenience, let $\anc_{S}(u):= u$ if $u\in S$.
\end{definition}

In \texttt{Phase 3}, we extend $T_0$ to span all remaining vertices $u$ by connecting them through their anchors $\anc_{S}(u)$, while maintaining the tree structure locally. Specifically, we employ a subroutine, \Cref{alg-find path}, which performs a BFS from a vertex $u \in V \setminus S$ in $G$, exploring vertices in lexicographical order until a vertex in $S$ is reached. 

We formally describe above procedure in \Cref{alg-global expander}.

\begin{algorithm}[htb]
\caption{Globally Computing a Spanning Tree}\label{alg-global expander}

\KwIn {$d$-bounded graph $G = (V, E)$ with conductance at least $\phi$.}
\KwOut {$T=(V, E^\prime)$ is a spanning tree of $G$.}

$E^\prime \gets \emptyset$\;

Select a seed $s \in V$ arbitrarily\;

Let $r := \Theta(\sqrt{n}\cdot \log n)$, $\tau :=  \tau(G) = \Theta(\frac{\log n}{\phi^2})$\;

\tcp{\texttt{Phase 1: perform random walks}}

Perform independent length-$\tau$ lazy random walks $\rho_1, \rho_2, \cdots \rho_r$ on $G_{reg}$, where $\rho_i = (v_1^{(i)},\cdots,v_\tau^{(i)})$ and $v_1^{(i)} = s$\;

Let $S := \bigcup_{i\in [r]} V(\rho_i)$, $H := \bigcup_{i\in [r]} E(\rho_i)$ and $G_0 := (S,H)$\;

\tcp{\texttt{Phase 2: refine $G_0$ to get core tree $T_0$}}

Perform BFS from seed $s$ in $G_0$ until all vertices are explored. $H^\prime \subseteq H$ includes every edge through which the BFS visits a vertex for the first time\;

Let $T_0 := (S,H^\prime)$, $E^\prime \gets H^\prime$\;

\tcp{\texttt{Phase 3: extend core tree to other vertices}}

\For{each $u \in V\setminus S$}{
$P \gets $ \textsc{FindPath}$(G,S,\sqrt{n},u)$\;
$E^\prime \gets E^\prime \cup E(P)$;\
}

\Return $T := (V, E^\prime)$.

\end{algorithm}

\begin{algorithm}[htb]
\caption{\textsc{FindPath}$(G,S,\theta,u)$: explore at most $\theta$ vertices to compute path $P_G(u,\anc_S(u))$ for any vertex $u$}\label{alg-find path}

\KwIn{graph $G = (V, E) \sim \Gnp$ with $np=n^\delta$; vertices set $S \subseteq V$; threshold $\theta$; $u\in V$.} 
\KwOut{the lexicographically-least shortest path from $u$ to set $S$, or $\bot$ if it fails.} 

\If {$u \in S$} {\Return $(u)$.}

Launch BFS from $u$ in $G$, until $\theta$ distinct vertices have been visited\;

\If {BFS finds a vertex in set $S$} 
{Let $v$ be the first vertex in BFS order which is incident to $S$\; 
$c \gets \smallest{\Gamma(v) \cap S}$\;
\Return the lexicographically-least shortest path from $u$ to vertex $c$.}

\Return $\bot$. \tcp{fail}

\end{algorithm}

\subsubsection{Correctness of \Cref{alg-global expander}}\label{sec-correctness}

In this section, we prove that the output subgraph $T$ is indeed a spanning tree of the input graph $G$.

\begin{theorem}\label{thm-global tree}

    If every invocation of \textsc{FindPath} does not return $\bot$ in \Cref{alg-global expander}, then output graph $T=(V,E^\prime)$ is a spanning tree of $G$.
\end{theorem}

Note that until the beginning of \texttt{Phase 3}, $E^\prime$ contains no cycle. We thus put our attention on \texttt{Phase 3}. In \texttt{Phase 3}, each of $\Theta(n)$ vertices outside the $S$ is processed by a iteration, in order to find a path to the core tree and add it to $E^\prime$. Consider that each of these $\Theta(n)$ iterations in \texttt{Phase 3} adds a path (probably of length $\log n$) to $E^\prime$, it seems that there are far more than $(n-1)$ edges included in $E^\prime$ after \texttt{Phase 3}. However, following claim tells us that most edges added to $E^\prime$ during \texttt{Phase 3}, are added for multiple times. Hence finally $T$ only contains $(n-1)$ edges.

\begin{lemma}\label{lem-subpath}
    For any set $S\subset V$ and $u \in V\setminus S$, let $v = \anc_S(u)$ be the anchor of $u$ w.r.t. $S$ and path $P_G(u,v) = (u, u_1, \cdots, u_k, v)$ be the lexicographically-least shortest path. Then for $i = 1,2,\cdots k$, we have:
    \begin{align*}
        &\anc_S(u_i) = v, \\
        &P_G(u_i,v) = (u_{i}, u_{i+1}, \cdots, u_k, v).
    \end{align*}
\end{lemma}
\begin{proof}
    First, we note that $u_i \in V \setminus S$. Otherwise, $(u, u_1, \cdots, u_i)$ becomes a path between $u$ and $S$ strictly shorter than $P_G(u,v)$, which contradicts to \Cref{def-anc}.

    Assume that $\anc_S(u_i) \neq v$, then there must be a vertex $w \in S\setminus \{v\}$ s.t. $P_G(u_i,w) \prec P_G(u_i,v)$. We denote $P_G(u_i,w)$ by $(u_i,v_{i+1},v_{i+2},\cdots,v_\ell,w)$. By \Cref{def-comp}, either $l < k$ or $l=k$ and $(u_i,v_{i+1},v_{i+2},\cdots,v_\ell,w)$ has smaller lexicographical order than $(u_{i}, u_{i+1}, \cdots, u_k, v)$. In both cases, by \Cref{def-comp} we see that $(u, u_1, \cdots, u_i, v_{i+1},v_{i+2},\cdots,v_\ell,w) \prec (u, u_1, \cdots, u_k, v)$. Thus by \Cref{def-anc} $\anc_S(u) \neq v$, which leads to a contradiction.

    Thus $\anc_S(u_i) = v$ for each $i \in [k]$. Similarly we assume that $P_G(u_i,v) \neq (u_{i}, u_{i+1}, \cdots, u_k, v)$, then there comes a contradiction to $P_G(u,v) = (u, u_1, \cdots, u_k, v)$.
\end{proof}

Now we formally prove \Cref{thm-global tree}.

\begin{proof}[Proof of \cref{thm-global tree}]
    With respect to $T_0=(S,H^\prime)$ in \Cref{alg-global expander}, for a vertex $u \in V\setminus S$, we call $u$ a \textit{``leaf''} vertex if for every neighbor $v$ of $u$ and $v \in V\setminus S$, path $P_G(u, \anc_S(u))$ is not a sub-path of $P_G(v, \anc_S(v))$. Denote the set of leaf vertices by $L$.

    In \texttt{Phase 3}, the iterations over all vertices in $V\setminus S$ can be reduced to iterations over vertices in $L$. 
    
    To see this, consider any non-leaf vertex $v_1$. By the definition of leaf vertex, there exists a neighbor $v_2$ of $v_1$ such that $P_1 := P_G(v_1, \text{anc}_S(v_1))$ is a proper sub-path of $P_2 := P_G(v_2, \text{anc}_S(v_2))$. Thus $E(P_1) \subseteq E(P_2)$.

    If $v_2$ is a leaf ($v_2 \in L$), we are done; otherwise, we can recursively apply this reasoning to $v_2$'s neighbor $v_3$, whose concerned path $P_G(v_3, \text{anc}_S(v_3))$ completely contains $P_G(v_2, \text{anc}_S(v_2))$. This process must terminate at some leaf vertex $v_k \in L$, since each time it extends the length of considered path by $1$, and all shortest paths in the graph have finite length. Thus for iterations in \texttt{Phase 3}, it is equivalent to only consider iterations over vertices in $L$.

    For convenience, we label the leaf vertices in $L$ as $\ell_1, \ell_2, \dots$ according to their processing order in \texttt{Phase 3}. For $\ell_i \in L$, define $a_i := \anc_S(\ell_i)$. Following observation is helpful:

    \begin{claim}\label{clm-subpath}
        For $\ell_i,l_j \in L$, if $P_G(\ell_i,a_i)$ and $P_G(\ell_j,a_j)$ intersect on some vertex $u$, for convenience we denote them by $(\ell_i,\cdots ,u, \cdots ,a_i)$ and $(\ell_j,\cdots ,u, \cdots ,a_j)$ respectively. Then the two paths are same after the intersection vertex $u$, i.e. $(u, \cdots, a_i) = (u, \cdots, a_j)$.
    \end{claim}
    \begin{proof}
        With respect to vertex $u$, we directly apply \Cref{lem-subpath} on $P_G(\ell_i,a_i)$ and $P_G(\ell_j,a_j)$ respectively. Then we have: 
        \begin{itemize}
            \item $\anc_s(u) = a_i, P_G(u,\anc_s(u)) = (u, \cdots, a_i);$ and
            \item $\anc_s(u) = a_j, P_G(u,\anc_s(u)) = (u, \cdots, a_j).$
        \end{itemize}
        Thus $(u, \cdots, a_i) = (u, \cdots, a_j)$.
    \end{proof}

    Now we are ready to prove by induction that $E^\prime$ is acyclic after \texttt{Phase 3}. Until the beginning of \texttt{Phase 3}, $E^\prime$ is acyclic because $T_0$ is a BFS tree. After first iteration on vertex $\ell_1 \in L$, $E^\prime$ is still acyclic because $a_1$ is the only vertex that belongs to $V(E^\prime)$ on path $P_G(\ell_1,a_1)$ and $P_G(\ell_1,a_1)$ is a shortest path between $\ell_1$ and $a_1$. 
    
    Assume that $E^\prime$ is acyclic after iteration on vertex $\ell_{i-1} \in L$, then $E^\prime$ would remain acyclic after iteration on $\ell_i$ because: if $P_G(\ell_i,a_i)$ does not intersect with path $P_G(\ell_j,a_j)$ for every $j < i$, this situation is reduced to first iteration and $E^\prime$ remains acyclic after adding $P_G(\ell_i,a_i)$. Otherwise, let $u$ be the first vertex on $P_G(\ell_i,a_i)$ that also belongs to another path $P_G(\ell_j,a_j)$ with $j < i$. Using above observation, we see that $P_G(\ell_i,a_i)$ and $P_G(\ell_j,a_j)$ are same after the intersection vertex $u$. This is equivalent to adding a sub-path $(\ell_i,\cdots,u)$ of $P_G(\ell_i,a_i)$ to current $E^\prime$ and $u$ is the only vertex on the sub-path that belongs to current $V(E^\prime)$. Thus $E^\prime$ still remains acyclic after including $P_G(\ell_i,a_i)$.

    On the other hand, after \texttt{Phase 3}, $E^\prime$ admits a path from root $s$ to any other vertex, thus $T$ is connected. Putting together, we see that $T$ is acyclic and spans every vertex in $V$, thus $T$ is a spanning tree. This finishes proof of this lemma.
\end{proof}

\subsection{Local Implementation}
We show how to implement the global \Cref{alg-global expander} locally, i.e. for the concerned edge $(u,v)$, we need some local procedure to determine whether $(u,v)$ belongs to the final output $E^\prime$. From a high-level view, it is easy to see that \texttt{Phase 1} and \texttt{Phase 2} can be implemented in $\tilde O(\sqrt{n})$ probes. Note that we always use \emph{same random bits} to perform random walks in \texttt{Phase 1}, thus it remains consistent over any sequence of queries.

However, in order to implement \texttt{Phase 3}, we should check if $(u,v)$ appears in any path added to $E^\prime$ during iterations, but directly performing these checks requires $\Omega(n)$ probes. Thanks to following lemma of locality, we could implement it with much less probes.

\begin{claim}\label{clm-local path}
    With respect to $S$ and output $T$ in \Cref{alg-global expander}, for any edge $(u,v)\notin S \times S$, edge $(u,v) \in E(T)$ if and only if $(u,v) \in P_G(u,\anc_S(u))$ or $(u,v) \in P_G(v,\anc_S(v))$.
\end{claim}
\begin{proof}
    If $(u,v) \in P_G(u,\anc_S(u))$ or $(u,v) \in P_G(v,\anc_S(v))$, then obviously we have $(u,v) \in E(T)$. In the following we consider the other direction.

    If $(u,v) \in E(T)$, then there exists some vertex $w \in V\setminus S$ s.t. $(u,v) \in P_G(w,\anc_S(w))$. As $(u,v) \notin S \times S$, edge $(u,v)$ could only be added to $E^\prime$ in \texttt{Phase 3}, hence it must belong to path $P_G(w,\anc_S(w))$ for some vertex $w$. Thus we have either 
    \[P_G(w,\anc_S(w)) = (w,\cdots,u,v,\cdots,\anc_S(w)),\] 
    or  
    \[P_G(w,\anc_S(w)) = (w,\cdots,v,u,\cdots,\anc_S(w)).\]
     For both cases, we apply \Cref{lem-subpath} on $P_G(w,\anc_S(w))$ w.r.t. the vertex of $(u,v)$ that appears earlier in the path. Thus we have either $(u,v) \in P_G(u,\anc_S(u))$ or $(u,v) \in P_G(v,\anc_S(v))$.
\end{proof}

\begin{algorithm}[htb]
\caption{\textsc{inTree}$(G,(u,v))$: locally answer whether $(u,v)$ belongs to spanning tree induced by \Cref{alg-global expander}}\label{alg-local expander}

\KwIn{$d$-bounded graph $G = (V, E)$ with conductance at least $\phi$.}

\tcp{\texttt{On query of edge $(u,v)$:}}

Select a seed $s \in V$ arbitrarily\;
Let $r := \Theta(\sqrt{n}\cdot \log n)$, $\tau := \Theta(\frac{\log n}{\phi^2})  \geq \tau(G_{reg}) $\;
Perform $r$ independent length-$\tau$ lazy random walks staring from $s$ on $G_{reg}$, and refine these walks to get its BFS tree $T_0=(S,H^\prime)$\;
\If {$(u,v) \in H^\prime$} {\Return \textsc{Yes}.}

$P_u \gets$ \textsc{FindPath}$(G,S,\sqrt{n},u)$, $P_v \gets$ \textsc{FindPath}$(G,S,\sqrt{n},v)$\;

\If {$(u,v) \in E(P_u)$ or $(u,v) \in E(P_v)$} {\Return \textsc{Yes}.}
\Return \textsc{No}.

\end{algorithm}

Now we have Local Computation \Cref{alg-local expander} to simulate \Cref{alg-global expander}, which outputs membership of any concerned edge $(u,v)$. Recall that the random tape is predetermined and thus \Cref{alg-local expander} always outputs consistent answer on every query.

\begin{lemma}\label{lem-local sim}
    For the same input graph and random tapes, \Cref{alg-local expander} return \textsc{Yes} on query of edge $(u,v)$ if and only if $(u,v) \in E(T)$ in \Cref{alg-global expander}.
\end{lemma}
\begin{proof}
    In \Cref{alg-local expander}, random tapes are only used to select seed and perform random walks in \texttt{Phase 1} and there are no randomness in \texttt{Phase 2} and \texttt{Phase 3}. Using the same random tapes, in both algorithms, the core tree $T_0=(S,H^\prime)$ are the same and thus they agree on edge $(u,v) \in S \times S$. Additionally, by \Cref{clm-local path}, the two algorithms also agree on edge $(u,v) \notin S \times S$. Hence \Cref{alg-local expander} faithfully simulates the spanning-tree membership of every edge computed by \Cref{alg-global expander}.
\end{proof}

\subsection{Probe Complexity of the LCA}
In order to bound probe complexity of \Cref{alg-local expander}, we first bound the probe complexity of \textsc{FindPath} (\Cref{alg-find path}).

\begin{lemma}\label{lem-end vertex}
    Let $G=(V,E)$ be $d$-bounded graph with conductance $\phi_G \geq \phi$. Let $\tau \geq \tau(G_{reg})$ and $r > 0$ be integers. With respect to $r$ independent length-$\tau$ lazy random walks on $G_{reg}$, let $S$ be the set of vertices in the trajectory of these random walks. Then for any given set of vertices $Q$ s.t. $|Q| \cdot r \geq 100n\log n$, we have $|Q \cap S| \geq 1$ with probability at least $1-\frac{1}{n^{10}}$. \footnote{To save random bits used by LCA, it is sufficient to let these random walks be $\log n$-wise independent.}
\end{lemma}
\begin{proof}
    Let $S_C=(v_\tau^{(1)},v_\tau^{(2)},\cdots,v_\tau^{(r)}) \subseteq S$ be the set of end-vertex of each random walk. According to \Cref{cor-lazy reg mixing time}, we have $\Pr[v_\tau^{(i)} = u] \geq \frac{1}{2n}$ for every $i\in[r]$ and $u\in V$. Let random variable $X_i$ indicate the event that the $i$-th random walk ended at $Q$ and $X:=\Sigma_{i\in [r]} X_i$. Note that $X \geq 1 \implies |Q \cap S| \geq 1$. Then we have 
    \begin{align*}
        \E[X_i]=\Pr[v_\tau^{(i)} \in Q] = \Sigma_{u\in Q} \Pr[v_\tau^{(i)} = u] \geq \frac{|Q|}{2n},
    \end{align*}
    and by linearity of expectation
    \begin{align*}
        \E[X]= \E[\Sigma_{i\in [r]} X_i] = \Sigma_{i\in [r]} \E[X_i] \geq \frac{r \cdot |Q|}{2n} \geq 50 \log n.
    \end{align*}
    Then a standard argument using Chernoff-Hoeffding bound finishes this proof.
\end{proof}

\begin{claim}\label{clm-findpath probes}
    Let $G=(V,E)$ be $d$-bounded graph with conductance $\phi_G \geq \phi$. With respect to set $S$ in \Cref{alg-local expander}, then \textsc{FindPath}$(G,S,\sqrt{n},u)$ makes $O(d\sqrt{n})$ probes for any given vertex $u$ and does not return $\bot$ with probability at least $1-\frac{1}{n^2}$.
    
\end{claim}
\begin{proof}
    Let $Q$ be the set of first $\sqrt{n}$ distinct vertices explored under the order of lexicographical BFS from $u$. Note that $r := \Theta(\sqrt{n}\cdot \log n)$ and $|Q| = \sqrt{n}$, thus $|Q|\cdot r \geq 100 \log n$ by choosing proper constant. In this scenario, \Cref{lem-end vertex} tells that, with probability at least $1-\frac{1}{n^{10}}$, the BFS in \textsc{FindPath}$(G,S,\sqrt{n},u)$ already finds a vertex in $v \in S$ at the moment that $Q$ is fully explored. Hence it does not return $\bot$ with probability at least $1-\frac{1}{n^{10}}$.
    
    Let $H:= G[Q]$ denote the subgraph induced by $Q$, then the probe complexity of BFS is $O(|V(H)| + |E(H)|)$. As the input graph is $d$-bounded, thus the subgraph induced by $Q$ contains at most $d \cdot \abs{Q} = d \sqrt{n}$ edges. Consider that computing the output path does not need more probes except for revealing all neighbors of $v$, hence for any given vertex $u \in V \setminus S$, \textsc{FindPath}$(G,S,\sqrt{n},u)$ finishes in $d\sqrt{n}+d=O(d\sqrt{n})$ probes.
\end{proof}

\begin{lemma}\label{lem-expander local probe}
    Let $G=(V,E)$ be $d$-bounded graph with conductance $\phi_G \geq \phi$. Then \Cref{alg-local expander} has probe complexity $O(\sqrt{n}(\frac{ \log^2n}{\phi^2}+d))$.
\end{lemma}
\begin{proof}
    For any edge $e=(u,v)\in E$, the LCA would firstly perform lazy random walks using at most $r \cdot \tau= O(\frac{\sqrt{n} \log^2n}{\phi^2})$ probes. Then refining the random walk trajectory to get a BFS tree $T_0$ uses no extra probes. 
    
    According to \Cref{clm-findpath probes}, a single call of \textsc{FindPath}$(G,S,\sqrt{n},w)$ makes $O(d \sqrt{n})$ probes for any given vertex $w$. Thus \textsc{FindPath}$(G,S,\sqrt{n},u)$ and \textsc{FindPath}$(G,S,\sqrt{n},v)$ together make $O(d \sqrt{n})$ probes.
    
    Thus totally, \Cref{alg-local expander} makes $O(\frac{\sqrt{n} \log^2n}{\phi^2})+O(d \sqrt{n}) = O(\sqrt{n}(\frac{ \log^2n}{\phi^2}+d))$ graph probes on any given edge. This finishes the proof.
\end{proof}
\subsection{Proof of \Cref{thm:d-expander}}
We recall \Cref{thm:d-expander} for convenience.
\lcaexpander*
\begin{proof}
    According to \Cref{clm-findpath probes}, invocation of \textsc{FindPath}$(G,S,\sqrt{n},u)$ makes $O(d\sqrt{n})$ probes for any given vertex $u$, and it does not return $\bot$ with probability at least $1-\frac{1}{n^2}$. By a union bound over all vertices, with probability at least $1-\frac{1}{n}$, every invocation of \textsc{FindPath} does not return $\bot$ in \Cref{alg-global expander}. Hence by \Cref{thm-global tree} and \Cref{lem-local sim}, the subgraph $T$ is a spanning tree of $G$. By using consistent random bits, \Cref{alg-local expander} satisfies \Cref{def:LCA} and is thus an LCA for spanning tree.
    
    Now we bound the depth of $T$. Note that \Cref{cor-lazy reg mixing time} implies that the diameter of input graph is of $\tau(G_{reg}) = O(\frac{\log n}{\phi^2})$. For every vertices $u \in S$ in core tree $T_0=(S,H^\prime)$ of \Cref{alg-local expander}, $T$ must admit a path from seed $s$ to $u$ of length $O(\frac{\log n}{\phi^2})$. This is because $T_0$ is a BFS tree on random walk trajectory, and the length of random walk if bounded by $O(\frac{\log n}{\phi^2})$. Path $P_G(u,\anc_S(u))$ is shortest between $u$ and $\anc_S(u)$, hence the length is at most the diameter of $G$. Thus together, the path from $s$ to $u$ is also of length $O(\frac{\log n}{\phi^2})$. Thus the tree $T$ has depth at most $O(\frac{\log n}{\phi^2})$. Plugging with \Cref{lem-expander local probe} of probe complexity, we finish the proof of \Cref{thm:d-expander}.
\end{proof}

\section{Average-Case LCA for Spanning Tree on $\Gnp$}
Now we present our average-case LCA for \ER graphs and prove \Cref{thm:lcaERcombined}. Our approach depends on the value of the parameter $\delta$, and we therefore design two different algorithms for the two regimes. For $\delta \leq \tfrac{1}{3}$, we first give an algorithm in the \textbf{adjacency-list model} that achieves $o(\sqrt{n})$ probe complexity.

\begin{theorem}\label{thm:avg lca 1}
Let $np = n^{\delta}$ for any constant $\delta \in (0,\tfrac{1}{3}]$, and let $G \sim \Gnp$.
Assuming probe access to $G$ in the adjacency-list model, there exists an average-case LCA that, with probability at least $1 - \tfrac{1}{n}$, supports edge-membership queries to a spanning tree of $G$. Moreover, the LCA has probe complexity $\tilde{O}(\sqrt{n^{1-\delta}})$.  
\end{theorem}

For $\delta > \tfrac{1}{3}$, a different approach is required. In this regime, even $\tilde{O}(\sqrt{n^{1-\delta}})$ probes are insufficient to fully explore the neighborhood of a single vertex, since the typical degree is $n^{\delta} = \omega(\sqrt{n^{1-\delta}})$.

\begin{theorem}\label{thm:avg lca 2}
Let $np = n^{\delta}$ for any constant $\delta > \tfrac{1}{3}$, and let $G \sim \Gnp$.
Assuming probe access to $G$ in the general-graph model, there exists an average-case LCA that, with probability at least $1 - \tfrac{1}{n}$, supports edge-membership queries to a spanning tree of $G$. Moreover, the LCA has probe complexity $\tilde{O}(\sqrt{n^{1-\delta}})$.  
\end{theorem}

\Cref{thm:lcaERcombined} follows directly by combining \Cref{thm:avg lca 1} and \Cref{thm:avg lca 2}.

In both algorithms underlying \Cref{thm:avg lca 1} and \Cref{thm:avg lca 2}, each vertex proposes one incident edge for inclusion in the spanning tree. An edge $(u,v)$ is included if it is proposed by either endpoint. Rather than directly implementing an oracle that answers edge-membership queries, we describe how to implement an oracle that computes the proposal edge of any given vertex; edge membership can then be determined from these proposals. 

We begin by introducing the necessary notation, and then proceed to the proofs of both theorems.

\subsection{Notation}

We fix $\delta$ to be a constant in $(0, 1]$. Since the input graph $G$ is drawn from $\Gnp$ with $np = n^\delta$, it holds \emph{with high probability}\footnote{Throughout this section, ``with high probability'' means with probability at least $1 - \frac{1}{n^c}$ for some sufficiently large constant $c$.} that the maximum degree of $G$ is at most $2n^\delta$, and that $G$ is a good expander with conductance $\phi_G = \Omega(1)$. Rather than conditioning explicitly on these events, we simply define parameters $d = 2n^\delta$ to be upper bound of vertex degree and $\phi = \Omega(1)$ to be lower bound of conductance, which are used in the lazy random walk described in \Cref{sec-preliminary}. If these properties do not hold for $G$, the lazy random walks simply fail. Note that the lazy random walk is only used in the LCA described in \Cref{thm:avg lca 1}.

For every vertex pair $(u, v)$, let $X_{u,v}$ be the indicator variable for the event that $(u, v) \in E(G)$; under $\Gnp$, these variables are mutually independent and equal to $1$ with probability $p$. When we refer to the probability of an event occurring ``over the randomness of the input,'' we mean that the event is determined by a set of these random variables $\{X_{u,v}\}$.

In the proofs of both \Cref{thm:avg lca 1} and \Cref{thm:avg lca 2}, by saying that a vertex is \emph{small}, we mean that its ID is small. The notations $\smallest{\cdot}$ and $\largest{\cdot}$ refer to the smallest and largest IDs in a given set, as defined in \Cref{def-small}.

Rather than directly answering membership queries on edges, the algorithms from \Cref{thm:avg lca 1} and \Cref{thm:avg lca 2} return a \emph{proposal edge} for each queried vertex. Note that, for any vertex $u \in V \setminus \{n\}$, the algorithm determines an incident edge of $u$ to be the proposal edge of $u$. In this way, membership query on any edge $(u,v)$ could be answered by checking if it is proposed by $u$ or $v$.

\subsection{Proof of \Cref{thm:avg lca 1}}

\subsubsection{Algorithm Overview}

Using following definition, we present our algorithm of \Cref{thm:avg lca 1} in \Cref{alg-gnp prop1}. 

\begin{definition}
    For any vertex $u \in V$, let {\em $\Gamma^+(u):= \{v \mid v > u,v\in \Gamma(u)\}$}, i.e. $\Gamma^+(u)$ denotes the set of neighbors of $u$ with ID larger than $u$.
\end{definition}

\begin{algorithm}[htb]
\caption{\textsc{Prop}$(G,u)$: compute the proposal edge of vertex $u$}\label{alg-gnp prop1}

\KwIn{graph $G = (V, E) \sim \Gnp$ with $np=n^\delta$ for $\delta \leq \frac{1}{3}; u\in V\setminus\{n\}.$}

\tcp{On query of vertex $u$:}

\If{$u == n$}{\Return $\bot$}

Let $r := \Theta(\sqrt{n^{1-\delta}} \log n)$, $\tau := \Theta(\log n)$\;

\tcp{\texttt{Phase 1: random walk and core tree}}
Perform $r$ independent length-$\tau$ lazy random walks staring from vertex $n$ on $G_{reg}$, and refine these walks to get its BFS tree $T_0=(S,H^\prime)$\;
\If {$u \in S$} {\Return the edge through which the BFS visits $u$ for the first time.}

\tcp{\texttt{Phase 2: find path to core tree}}
Let $\theta := \Theta(\sqrt{n^{1-\delta}}\log^2 n)$ be an integer;

$P_u \gets$ \textsc{FindPath}$(G, S, \theta, u)$\;

\If{$P_u \neq \bot$}{\Return the first edge on $P_u$}

\tcp{\texttt{Phase 3: maintain increasing path}}
\If {$\Gamma^+(u) \neq \emptyset$}{\Return $(u,\smallest{\Gamma^+(u)})$.}
\Else {\Return $(u,\smallest {\Gamma(u)})$.}
\end{algorithm}

From a high-level perspective, \Cref{alg-gnp prop1} use a similar idea to \Cref{alg-local expander}: on query proposal edge of any vertex $u$, we first use the \emph{same random bits} to perform random walks, and refine these walks to maintain a core tree in subgraph. If $u$ is visited by random walk, then it already belongs to a tree rooted at $n$; otherwise, we perform lexicographical BFS to find path from $u$ to the core tree. However, to reduce the probe complexity, we reduce the number of random walks and limit the budget of BFS. This definitely leads to a problem: a vertex may fail to find a path to core tree.

For vertex $u$, if \textsc{FindPath}$(G,S,\theta,u)$ does not return $\bot$, then we call $u$ a \textbf{good vertex}. For a good vertex $u$, the first $\theta$ distinct vertices visited by BFS from $u$, must contain at least one vertex in $S$. Let $V_{good}$ denote the set of all good vertices, note that $S \subseteq V_{good}$. We have following observation.
\begin{observation}\label{obs-good vertex}
    All proposal edges of vertices in $V_{good}$ by \Cref{alg-gnp prop1} form a spanning tree of $V_{good}$.
\end{observation}
To see this, we note that a good vertex $u$ is either a vertex in $S$ (\texttt{Phase 1}), thus it belongs to the BFS tree $T_0$; or \textsc{FindPath}$(G,S,\theta,u)$ finds the anchor of $u$ in $S$ (\texttt{Phase 2}), thus every edge on $P_G(u,\anc_{S}(u))$ is retained in $H$, as we discussed before in \Cref{clm-local path}. Hence every good vertex is directly connected to the core tree by proposal edges. 

On the other hand, for a vertex $u \notin V_{good}$, our algorithm tries to ensure that $u$ proposes an edge with large potential to connect $u$ and $n$. Presented in \texttt{Phase 3} of \Cref{alg-gnp prop1}, the rules maintain a ``path'' from bad vertex $u$, at each step traversing to the smallest neighbor among its larger neighbors. We use following definition of ``increasing path'' to characterize the path.
\begin{definition}[Increasing path]\label{def-inc path}
    A path $(v_1,v_2,\cdots,v_\ell)$ is called an \textbf{increasing path} if for every $i \in [l-1]$, $v_{i+1} = \smallest{\Gamma^+(v_i)}$. We use   $IP(v)$ to denote the longest increasing path starting from vertex $v$.
\end{definition}
Notice that once two increasing paths intersect at a vertex, they continue in exactly the same manner after the intersection vertex. Hence such  paths  do not cause cycles.

Note that the algorithm may not maintain all edges in $IP(u)$ for a bad vertex $u$, because a good vertex on $IP(u)$ determines its proposal edge in \texttt{Phase 1} and \texttt{Phase 2}. However, if this happens, such vertex $u$ must be connected to the core tree in $H$ by following Observation.
\begin{observation}\label{obs-bad to good}
    If any good vertex appears in $IP(u)$, then the bad vertex $u$ is connected to the core tree in $H$.
\end{observation}
\begin{proof}
    Without loss of generality, let $v$ be the first good vertex on $IP(u)$. Then we have
    \begin{enumerate}
        \item every edge on sub-path from $u$ to $v$ on $IP(u)$ is proposed in \texttt{Phase 3}, thus the sub-path is retained in $H$;
        \item as $v$ is a good vertex, it is connected to the core tree in $H$ by \Cref{obs-good vertex}.
    \end{enumerate}
\end{proof}

We have not yet ruled out the possibility  that a bad vertex does not contain a good vertex through its increasing path. To fix this, \texttt{Phase 3} ensures that the proposal edge of a last vertex on increasing path, is toward its smallest neighbor. Through this edge, the potential of being connected to core tree, is delivered to the smallest neighbor. Therefore, it suffices to prove that every ``small'' vertex is connected to a good vertex in $H$.

A small vertex has a 

good probability of being spanned.
This is because, the increasing path of a small vertex is long; hence there are a lot of vertices on the path launching BFS in \texttt{Phase 2}; all these BFS together visit many distinct vertices, thus among them w.h.p. there is an end-vertex of random walk. This implies that w.h.p. a small vertex meets a good vertex through its increasing path.

\subsubsection{Probe Complexity of \Cref{alg-gnp prop1}}
\begin{claim}\label{clm-gnp findpath probe}
    With high probability, for every vertex $u$, invocation of \textsc{FindPath}$(G,S,\theta,u)$ in \Cref{alg-gnp prop1} finishes using $O(\sqrt{n^{1-\delta}}\log^2 n)$ probes.
\end{claim}
\begin{proof}
    The probe complexity of \textsc{FindPath} is dominated by performing BFS until $\theta$ distinct vertices being visited. Assume that BFS in \textsc{FindPath}$(G,S,\theta,u)$ has already explored a set of distinct vertices $K$ with $\abs{K} < \theta$. Next algorithm continues BFS by fully exploring neighborhood of a vertex $v \in K$ and add all of $v$'s neighbors to $K$. We upper bound the number of steps of exploration before one expects an increase in $\abs{K}$.

    Let random variable $X_u$ be indicator of the event that $(v,u) \in E(G)$, and $X = \sum_{u \in V\setminus K}$ is the number of vertices incident to $v$ and not in $K$. As $\mu = \E[X] = (n-|K|)\cdot p \geq 0.99 n^\delta$, by Chernoff-Hoeffding bound, we have
    \[
    \Pr[X\leq 0.98 n^\delta] \leq \Pr[X-\mu \leq -0.01 n^\delta] \leq e^{-\Omega(n^\delta)}.
    \]
    This means, after $\Theta(n^\delta)$ (maximum degree) probes in BFS to fully explore neighborhood of a vertex, the size of $K$ grows by $\Theta(n^\delta)$. This always holds when $|K| < \theta$.

    Hence an inductive argument beginning at $K = \{u\}$, shows that the BFS only needs to fully explore neighborhoods of $\frac{\theta}{\Theta(n^\delta)}$ distinct vertices before $K$ reaches size $\theta$. After that, if $K$ contains at least $1$ vertex in $S$, the algorithm will use $\Theta(n^\delta)$ probes to fully explore neighborhood of some vertex. Note that $\Theta(n^\delta) \leq \sqrt{n^{1-\delta}}\log^2 n$ when $\delta \leq \frac{1}{3}$. This implies that the probe complexity of \textsc{FindPath}$(G,S,\theta,u)$ is $O(\theta) = O(\sqrt{n^{1-\delta}}\log^2 n)$.
\end{proof}

\begin{theorem}\label{thm-gnp lca probe}
    The probe complexity of \Cref{alg-gnp prop1} is $O(\sqrt{n^{1-\delta}}\log^2 n)$.
\end{theorem}
\begin{proof}
In \texttt{Phase 1}, performing $r$ independent length-$\tau$ random walks using probes $r \cdot \tau = O(\sqrt{n^{1-\delta}}\log^2 n)$. Computing $T_0$ does not cost extra probe. In \texttt{Phase 2}, by \Cref{clm-gnp findpath probe}, every invocation of \textsc{FindPath} finishes in $O(\sqrt{n^{1-\delta}}\log^2 n)$ probes. Additionally, \texttt{Phase 3} cost $O(n^\delta)$ probes. Hence together, the probe complexity is $O(\sqrt{n^{1-\delta}}\log^2 n)$.
\end{proof}

\subsubsection{Correctness of \Cref{alg-gnp prop1}}

We first use following definition to characterize the behavior of \Cref{alg-gnp prop1}.

\begin{definition}\label{def-bc1}
    A \textbf{baseline condition} is an event parameterized as $B(f,b)$ where
    \begin{itemize} 
        \item (Forward function) $f: V \to V\cup\{\bot\}$, is a function mapping vertex $v$ to a neighbor $u$, such that $u = \smallest{\Gamma^+(v)}$; or to $\bot$ if no such neighbor exists;
        \item (Backward function) $b: \{v \mid f(v) = \bot\} \to V\cup\{\bot\}$, is a function mapping a vertex $v$ to $\smallest{\Gamma(v)}$ i.e. $v$'s smallest neighbor; or to $\bot$ if no such neighbor exists.
    \end{itemize}

    We say a graph $G$ satisfies $B(f,b)$ if it is compatible with both $f$ and $b$. 
\end{definition}

\begin{lemma}\label{lem-unq cond1}
    For every pair of distinct baseline conditions $B_1 = B(f_1,b_1)$ and $B_2 = B(f_2,b_2)$, there is no graph $G$ that satisfies both.
\end{lemma}
\begin{proof}
    Consider that $B_1$ and $B_2$ are distinct, there exists at least one vertex $v$ s.t. either $f_1(v) \neq f_2(v)$ or $b_1(v) \neq b_2(v)$. Recall that function $f$ maps $v$ to $\smallest{\Gamma^+(v)}$ or $\bot$, and function $b$ maps $v$ to $\smallest{\Gamma(v)}$ if $f(v) = \bot$. Thus if $\Gamma(v)$ is compatible with $f_1$ and $b_1$, then it is not compatible with $f_2(v)$ and $b_2(v)$.
\end{proof}

\begin{lemma}\label{lem-uniqIP}
    Every baseline condition $B(f,b)$ uniquely specifies $IP(v)$ for any vertex $v$.
\end{lemma}
\begin{proof}
    If a graph satisfies $B(f,b)$, then for any vertex $v$, $\smallest{\Gamma^+(v)} = f(w)$ is determined by the function $f$, and it is exactly next vertex incident to $v$ on path $IP(v)$. Iteratively, every vertex on $IP(v)$ is determined by $f$, thus $IP(v)$ is uniquely determined.
\end{proof}

For graph conditioned on a baseline condition $B(f,g)$, we have following observation.
\begin{observation}\label{obs-pair random1}
    For $G \sim \Gnp$, if $G$ satisfies a baseline condition $B(f,g)$, then we have:
    \begin{itemize}
        \item if $f(v) \neq \bot$, then pair $(v,f(v)) \in E(G)$, and pair $(v,u) \notin E(G)$ for vertex $v < u < f(v)$;
        \item if $f(v) = \bot$ and $b(v) \neq \bot$, then pair $(v,b(v))  \in E(G)$, and pair $(v,u) \notin E(G)$ for vertex $u$ with $v < u \leq n$ or $ 1 \leq u < b(v)$;
    \end{itemize}
    and each remaining pair of vertices is connected by an edge independently with probability $p$. Or equivalently speaking, for each remaining pair $(w,z)$, $X_{w,z}$ remains random.
\end{observation}

Slightly abusing, let $f^{-1}(v)$ (resp. $b^{-1}(v)$) denote the set of vertices $u$ such that $f(u) = v$ (resp. $b(u)=v$). The following definition characterizes a desirable condition under which the algorithm behaves well.

\begin{definition}\label{def-fav cond1}
     Let integer $K = c n^{1-\delta} \log n$, where $c$ is a large enough constant. A baseline condition $B(f,b)$ is called a \textbf{favorable condition} if all of the following hold:
     \begin{enumerate}
         \item[(1)] for every $v$ such that $1 \leq v \leq n - K$, it holds that $f(v) \neq \bot$ and $f(v) \leq v + K$;
         \item[(2)] for every $v$ such that $f(v) = \bot$, it holds that $b(v) \neq \bot$ and $1 \leq b(v) \leq K$;
         \item[(3)] for every $v$, $|f^{-1}(v)| = O(\log n)$ and $|b^{-1}(v)| = O(\log n)$.
     \end{enumerate}
\end{definition}

We have following claim to relate favorable condition and graph property. For convenience, we first introduce a notation $V_I$ to denote the set of vertices induced by an interval $I$. Formally, the set $V_I$ consists of all integers modulo $n$ for elements in $I$, defined as follows.
\begin{definition}\label{def-int induced set}
    For any interval $I \subseteq \mathbb{R}$, let $V_I$ denote the set of vertices induced by $I$, defined as:
    \[
    V_I := \{ i \bmod n \mid \text{$i \in I$ and $i \in \mathbb{Z}$}\}.
    \]
\end{definition}

\begin{claim}\label{clm-fav graph1}
    If $G \sim G(n,p)$ satisfies following conditions simultaneously:
        \begin{enumerate}
            \item[(a)] every vertex $v$ is incident to at least one vertex in $V_{(v,v+K]}$;
            \item[(b)] every vertex $v$ is incident to at most $2c\log n$ vertices in $V_{[v-K,v)}$;
        \end{enumerate}
    then the baseline condition $B(f,g)$ that $G$ satisfies is a favorable condition. We call such $G$ a favorable graph.
\end{claim}
\begin{proof}
    Condition $(a)$ implies that for $1 \leq v \leq n - K$, $v$ is incident to 
    \[
    V_{(v,v+K]} = \{v+1, v+2, \cdots, v+K\}.
    \]
    Hence $f(v) = \smallest{\Gamma^+(v)} \leq v + K$. This is exactly Item $(1)$ in \cref{def-fav cond1}.
    
    For vertex $v$ s.t. $f(v) = \bot$, according to above argument, we have $v > n - K$. Applying Condition $(a)$, $v$ must be incident to
    \[
    V_{(v,v+K]} = \{v+1, v+2, \cdots, n, 1,\cdots, v+K-n\}.
    \]
    As $f(v) = \bot$, we have $\Gamma^+(v) = \emptyset$. On the one hand, this means $v$ is not incident to any vertex in $\{v+1, v+2, \cdots, n\}$; on the other hand, as $v$ is incident to $V_{(v,v+K]}$, it must be incident to one vertex in $\{1,\cdots, v+K-n\}$. Thus we have $b(v) = \smallest{\Gamma(v)} \leq v+K-n \leq K$, which is exactly Item $(2)$ in \cref{def-fav cond1}.

    Fix a vertex $v$, we see that 
    \[
    \text{$f(u) =v $ or $ b(u) =v$} \implies u \in V_{[v-K,v)}.
    \]
    This is because, if $u \notin V_{[v-K,v)}$, then $v \notin V_{(u,u+K]}$. Applying Condition $(a)$ on vertex $u$, this further means $f(u) \neq v$ and $ b(u) \neq v$, and leads to a contradiction. As Condition $(b)$ tells that $v$ is incident to at most $2c\log n$ vertices in $V_{[v-K,v)}$, the number of vertices $u$ s.t. $f(u) =v $ or $ b(u) =v$ is at most $2c\log n$. This proves that $B(f,g)$ satisfies Item $(3)$ in \cref{def-fav cond1}.
\end{proof}

Now we prove that for graph $G \sim \Gnp$, it satisfies some favorable condition with high probability.
\begin{lemma}\label{lem-fav graph 1}
    For $G \sim G(n,p)$:
    \begin{equation*}
        \sum_{\text{fav. } B(f,b)} \Pr[G \text{ satisfies $B$}] \geq 1- \frac{1}{n^2}
    \end{equation*}
    over randomness of input.
\end{lemma}
\begin{proof}
    We first show that $G$ is a favorable graph with high probability, then directly applying \Cref{clm-fav graph1} finishes the proof. Recall that we use random variable $X_{u,v}  \in \{0,1\}$ to indicate the event that $(u,v) \in E(G)$. 
    
    For every vertex $v$, we use random variables $Y_v$ and $Z_v$ to count neighbors of $v$, where
    \begin{align*}
        Y_v := \sum_{u \in V_{(v,v+K]}}X_{u,v}, \\
        Z_v := \sum_{u \in V_{[v-K,v)}}X_{u,v}.
    \end{align*}
    Note that if for every $v$, we have $Y_v \geq 1$ and $Z_v \leq 2c\log n$, then $G$ is a favorable graph. Thus we only need to prove with high probability, $Y_v \geq 1$ and $Z_v < 2c\log n$.

    As $\mu = \E[Y_v] = K \cdot p = c\log n$, by Chernoff bound, we have 
    \[
    \Pr[\text{$Y_v \leq 0$ or $Y_v \geq 2\mu$}] = \Pr[\abs{Y_v - \mu} \geq \mu] \leq e^{-\frac{\mu}{100}} \leq \frac{1}{n^{100}}.
    \]
    By a union bound over all vertices, we have $1 \leq Y_v < 2c\log n$ holds with high probability for every vertex $v$. Similar argument also works for $Z_v$, and it finishes the proof.
\end{proof}

Following lemma tells that if $G \sim \Gnp$ satisfies a favorable condition, then \Cref{alg-gnp prop1} works with high probability. The proof is deferred in \Cref{sec-proof1}.
\begin{restatable}{lemma}{lemfavtotreeA}\label{lem-fav2tree 1}
    For every favorable condition $B(f,b)$,
    \begin{equation*}
        \Pr_{G \sim G(n,p)}[H \text{ is a spanning tree} \mid G \text{ satisfies } B] \geq 1- \frac{1}{n^2},
    \end{equation*}
    where $H$ is the subgraph consists of all proposal edges determined by \Cref{alg-gnp prop1}.
\end{restatable}

Now we are ready to prove main theorem of this section.
\begin{theorem}\label{thm-lca1 is tree}
    For $np = n^\delta$ and $\delta \leq \frac{1}{3}$, let $H$ be the subgraph consists of all proposal edges determined by \Cref{alg-gnp prop1} on $G \sim \Gnp$. Then $H$ is a spanning tree with high probability.
\end{theorem}

\begin{proof}
    We have
    \begin{align*}
        \quad &\Pr_{G \sim G(n,p)}[\text{$H$ is a spanning tree}] \\
        &= \sum_{B(f,g)} \Pr[\text{$H$ is a spanning tree} \mid \text{$G$ satisfies $B$}] \cdot \Pr[\text{$G$ satisfies $B$}] \tag{\Cref{lem-unq cond1}}\\
        &\geq \sum_{\text{fav. $B(f,g)$}} \Pr[\text{$H$ is a spanning tree} \mid \text{$G$ satisfies $B$}] \cdot \Pr[\text{$G$ satisfies $B$}] \\
        &\geq (1-\frac{1}{n^2}) \sum_{\text{fav. $B(f,g)$}} \Pr[\text{$G$ satisfies $B$}] \tag{\Cref{lem-fav2tree 1}}\\ 
        &\geq (1-\frac{1}{n^2})(1-\frac{1}{n^2}) \tag{\Cref{lem-fav graph 1}}\\
        &\geq 1-\frac{1}{n}.
    \end{align*}
\end{proof}
By using consistent random bits over all possible queries, \Cref{thm-lca1 is tree} and \Cref{thm-gnp lca probe} together imply that \Cref{alg-gnp prop1} is an LCA for spanning tree satisfying \Cref{def:LCA} with high probability over randomness of input. Therefore, \Cref{alg-gnp prop1} is an average-case LCA and this finishes the proof of \Cref{thm:avg lca 1}.

\subsubsection{Proof of \Cref{lem-fav2tree 1}}\label{sec-proof1}
In this section, we prove \Cref{lem-fav2tree 1} by first showing a property of BFS on increasing path (see \Cref{lem-gnp bfs}), and then use it to prove the main lemma.

For better understanding, note that the neighborhood of a vertex $u$, is fully depending on events $\{X_{u,w}\}$ for all $w \in V\setminus \{u\}$, where $X_{u,w}$ indicates the event $(u,w) \in E$. If we condition on that $G$ satisfies some $B(f,b)$, then a part of the events $\{X_{u,w}\}$ are already determined, while most still remain random (see \Cref{obs-pair random1}). Thus analysis of neighborhood is naturally divided into determined part and undetermined part, and randomness of the undetermined part is key to our proof. Formally, we have following definition for the determined part.

\begin{definition}\label{def-det1}
    Given vertex $u \in V$, let $\Det(u)$ denote the set of vertices $w$ such that event $X_{u,w}$ is determined by $B(f,b)$, i.e.
    \[
     \Det(u):=\{w \mid \text{$X_{u,w}$ is determined by $B(f,g)$}\}.
    \]
    Similarly, let $ \Det(S) := \bigcup_{s \in S}  \Det(s)$ for any set $S$.
\end{definition}

Intuitively, if $v \in  \Det(u)$, then this means the condition $B(f,b)$ forces pair $(u,v)$ either to exist or not to exist in edge set $E(G)$. For example, if $B(f,b)$ satisfies $b(u) = v$, this means in graph $G$ the smallest neighbor of $u$ is $v$. So $(u,v) \in E(G)$ and $(u,w) \notin E(G)$ for $1\leq w < v$ is determined by $B(f,g)$. Then we $v \in  \Det(u)$ and  for $1\leq w < v$,  $w \in  \Det(u)$. Recall that we characterize all such pairs in \Cref{obs-pair random1}, and we use set $ \Det(u)$ here simply to characterize from the view of vertex $u$.\label{para-det}

The set $ \Det(u)$ is an important definition in our analysis. The reason why we define set $\Det(u)$ for vertex $u$, is because when we analyze the behavior of BFS from $u$, it comes a natural question like ``$u$ is connected to which vertex?''. To figure out this, we consider over every vertex $v \in V \setminus \{u\}$, and see if $v$ is connected to $u$. If $G \sim \Gnp$ is not conditioned on any event, it is easy to conclude that every pair forms an edge with mutually independent probability $p$. However, to analyze the behavior of algorithm, we assume that $G\sim\Gnp$ already satisfies a favorable condition $B(f,g)$. Hence intuitively, we must specify for each vertex $v$, whether pair $(u,v)$ is ``influenced'' by $B(f,g)$. This motivates us to define the set $\Det(u)$ to contain all such vertices ``influenced'' by $B(f,g)$.

We first show some properties of $\Det(u)$.

\begin{lemma}\label{lem-det size}
    Conditioned on $G \sim \Gnp$ satisfies a favorable condition $B(f,b)$, then for any vertex $u \in V(G)$, we have
    \begin{enumerate}
        \item $\abs{\Det(u)} = O(n^{1-\delta} \log n)$;
        \item $\abs{\{w \in \Det(u) \mid (u,w) \in E(G)\}} = O(\log n)$, i.e. $(u,w) \in E(G)$ holds for $O(\log n)$ vertices $w \in \Det(u)$.
    \end{enumerate}
\end{lemma}
\begin{proof}
    Let $K = c n^{1-\delta} \log n$ be the integer defined in \Cref{def-fav cond1}. Note that it satisfies for any $v \in V(G)$ that 
    \begin{itemize}
        \item $f(v) \leq v+K$ if $f(v) \neq \bot$; and
        \item $1 \leq b(v) \leq K$ if $f(v) = \bot$; and
        \item if $f(v) = \bot$, then $v > n-K$ is a large vertex.
    \end{itemize}

    Combining with \Cref{obs-pair random1}, given vertex $u$, a vertex $w \in \Det(u)$ falls into one of following cases:
    \begin{enumerate}
        \item[$(1)$] if $f(u) \neq \bot$, and $u < w \leq u+K$;
        \item[$(2)$] if $f(u) = \bot$, and $u < w \leq n$ or $1 \leq u \leq K$;
        \item[$(3)$] if $f(w) \neq \bot$, and $w < u \leq w+K$;
        \item[$(4)$] if $f(w) = \bot$, and  $w < u \leq n$ or $1 \leq u \leq w$.
    \end{enumerate}
    
    Case $(1)$ and $(3)$ together contributes at most $K + K = 2K$ possible choices of $w$. In Case $(2)$, as $f(u) = \bot$ and thus $u > n-K$, it contributes at most $(n-u) + K \leq 2K$. Similarly, Case $(4)$ contributes at most $(w-u) + K \leq 2K$. Together, we see that the total number of choices of $w$ is at most 
    \[
    6K = O(n^{1-\delta} \log n)
    \]
    Hence we prove $\abs{\Det(u)} = O(n^{1-\delta} \log n)$.

    Note that among all $w \in \Det(u)$, $(u,w) \in E(G)$ holds iff 
    \[
    \text{$f(u) = w$, or $b(u) = w$, or $f(w) = u$, or $b(w) = u$.}
    \]
    Thus $\abs{\{w \in \Det(u) \mid (u,w) \in E(G)\}}$ is bounded by $1+1+f^{-1}(u)+b^{-1}(u) = O(\log n)$, according to \Cref{def-fav cond1}.
    
\end{proof}

Now we are ready to prove \cref{lem-gnp bfs}.
\begin{lemma}\label{lem-gnp bfs}
    Conditioned on $G \sim \Gnp$ satisfies a favorable condition $B(f,b)$, let $k$ be an integer s.t. $n^{k\delta} \leq \sqrt{n^{1-\delta}}$ and $n^{(k+1)\delta} > \sqrt{n^{1-\delta}}$.
    For any vertex $u$, let $N_t(u)$ be the set of vertices visited by a $t$-level BFS starting from $u$, and denote $N_t(R) := \bigcup_{r \in R} N_t(r)$ for a set of vertices $R$.
    
    Then for any integer $1 \leq t \leq k$,
    \begin{itemize}
        \item for any given vertex $u \in V$, we have $$\Pr[|N_t(u)| = \Theta(n^{t\delta})] \geq 1-\frac{1}{n^{100}};$$
        \item for any given vertex $x$, let $R$ be a set of $O(n^\delta/\log n)$ vertices such that $x$ and vertices in $R$ belong to the same increasing path. Then we have $$\Pr[\abs{N_t(x) \cap N_t(R)} = O(n^{(t-1)\delta}\log n)] \geq 1-\frac{1}{n^{100}}.$$

    \end{itemize}
    Above events hold with randomness over input.
\end{lemma}
\begin{proof}

    We prove by induction.
    
    \paragraph{Base Case: $t=1$.}
    \begin{itemize}
        \item According to \Cref{lem-det size}, for any vertex $w$, there are at most $O(\log n)$ vertices $w$ in $\Det(u)$ s.t. $(u,w) \in E(G)$, which means that there are at most $O(\log n)$ vertices deterministically being a neighbor of $u$. 
    
        For $w \notin \Det(u)$, we have $\Pr[(u,w) \in E(G)] = p = \frac{n^\delta}{n}$. Using Chernoff-Hoeffding bound, it is easy to see that there are $\Theta(n^\delta)$ neighbors of $u$ in the undetermined part with high probability. Summing up these two parts, we have $|N_1(u)| = \Theta(n^\delta)$.

        \item For the intersection of $N_1(x)$ and $N_1(R)$, we first consider the edges between exactly $x$ and $R$. As $x$ and $R$ belong to the same increasing path $P$ determined by function $f$, there are at most $2$ vertices in $R$ deterministically incident to $x$, which are the vertices before and after $x$ on $P$. For other $r \in R$, it is connected to $x$ with independent probability $p$. This is because either $r > f(x)$ or $f(r) < x$, thus pair $(x,r)$ is not determined by function $f$ and similar for function $b$ by assuming that $x$ is not too large or small.
        
        Thus it is easy to see that in expectation there are $O(\frac{n^{2\delta}}{n\log n}) \leq 1$ vertices in $R$ incident to $x$. By Chernoff-Hoeffding bound, there are at most $O(\log n)$ vertices in $R$ incident to $x$.

        For vertices $w \notin \{x\} \cup R$, we partition them into $4$ cases:
        \begin{enumerate}
            \item $w \in \Det(x) \cap \Det(R)$. According to \Cref{lem-det size}, there are at most $O(\log n)$ vertices in $\Det(x)$ incident to $x$. Thus this part contributes at most $O(\log n)$ vertices to $N_t(x) \cap N_t(R)$.
            \item $w \in \Det(x) \setminus \Det(R)$. As above, this part contributes at most $O(\log n)$ vertices to $N_t(x) \cap N_t(R)$.
            \item $w \in \Det(R) \setminus \Det(x)$. According to \Cref{lem-det size}, there are at most $O(\log n)$ vertices deterministically incident to each vertex in $R$. Thus there at most $|R|\cdot \log n = O(n^{1-\delta})$ vertices in $\Det(R) \setminus \Det(x)$ incident to $R$, denoted by $N$. For $w \in N$, we have $w \notin \Det(x)$ and $\Pr[(w,x)\in E(G)] = p$ independently for every $w$. Thus in expectation, there are $|N| \cdot p = O(1)$ vertices in $N$ incident to $x$. By Chernoff-Hoeffding bound, this part contributes at most $c'\log n$ vertices to $N_t(x) \cap N_t(R)$ with probability at least $1-\Omega(\frac{1}{n^{c'}})$ for any constant $c'$.
            \item $w \notin \Det(x) \cup \Det(R)$. For every such vertex $w$, $w$ is incident to every vertex in $\{x\} \cup R$ with independent probability $p$. This is similar to a binomial distribution, thus in expectation there are 
            \[
            n \cdot p \cdot (1-(1-p)^{|R|}) \leq n \cdot p \cdot |R|p = O(1/\log n)
            \]
            vertices incident to $x$ and at least one vertex in $R$ at the same time, using Bernoulli's inequality. Still, by Chernoff-Hoeffding bound, this part contributes at most $O(\log n)$ vertices with high probability.
        \end{enumerate}
        Putting together, with high probability, we have $\abs{N_t(x) \cap N_t(R)} = O(\log n)$ for $t=1$.
    \end{itemize}

    \paragraph{Inductive Step.}
    Assume it holds for any $t \leq k-1$, we show that it also holds for $t + 1$.

    \begin{itemize}
        \item After a $t$-level BFS starting from $u$, every vertex with distance at most $t$ from $u$ has been visited. By inductive hypothesis, we have $|N_t(u)| = \Theta(n^{t\delta})$ and $|N_{t-1}(u)| = \Theta(n^{(t-1)\delta})$. Let $\partial_t(u) := N_t(u) \setminus N_{t-1}(u)$, which is the set of vertices at distance exactly $t$ from $u$. Then we have $|\partial_t(u)| = |N_t(u)|-|N_{t-1}(u)| = \Theta(n^{t\delta})$. To bound the size of $N_{t+1}(u)$, we need to specify neighbors of vertex in $\partial_t(u)$.
        \begin{definition}
            A vertex $w$ is called \textbf{heavily determined} by a vertex set $S$ if
            \[
            |\{ z \in S \mid w \in \Det(z) \}| \geq \frac{|S|}{2}.
            \]
            In other words, $w$ belongs to $\Det(z)$ for at least half of the vertices $z \in S$.
        \end{definition}

        \begin{observation}
            For any vertex set $S$, the number of vertices heavily determined by $S$, is $O(n^{1-\delta} \log n)$.
        \end{observation}
        \begin{proof}
            For $S = \{z_1,z_2,\cdots,z_{|S|}\}$, let $M$ be the multi-set of the union of all $\Det(z_i)$. Thus $|M| = |S| \cdot O(n^{1-\delta} \log n)$ because $|\Det(z_i)| = O(n^{1-\delta} \log n)$. If a vertex is heavily determined by $S$, then it appears in $M$ for at least $|S|/2$ times. Thus there are at most $|M|/(\frac{|S|}{2}) = O(n^{1-\delta} \log n)$ heavily determined vertices.
        \end{proof}

        \begin{observation}
            For every vertex $w$ that is not heavily determined by $S$, it is connected to at least one vertex in $S$ with at least $\frac{p|S|}{4}$ independent probability .
        \end{observation}
        \begin{proof}
            Easy to see that $w$ is connected to $S$ if $(w,z) \in E(G)$ for any $z \in S$ s.t. $w \notin \Det(z)$.
            
            As $w$ is not heavily determined by $S$, then for at least half of the vertices $z \in S$, $w \notin \Det(z)$ and $\Pr[(w,z) \in E(G)] = p$ independently. Thus it happens with probability at least $1-(1-p)^{|S|/2} \geq \frac{p|S|}{4}$.
        \end{proof}

        Now we take $S = \partial_t(u)$. Above observation tells that for $w \notin V \setminus N_t(u)$, if $w$ is not heavily determined by $S$ then it has a large chance to be incident to $S$. We have $\Theta(n)$ such vertices, because $|N_t(u)| = \Theta(n^{t\delta}) = o(n)$ and the total number of heavily determined vertices is small. Thus with high probability, the number of distinct vertices newly visited by $(t+1)$-th level of BFS, is at least $\Theta(n) \cdot \frac{p|S|}{4} = \Theta(n^{(t+1)\delta})$. This means $|N_{t+1}(u)| = \Theta(n^{(t+1)\delta})$, thus the first item also holds for $t+1$.

        \item To compute the size of $N_{t+1}(x) \cap N_{t+1}(R)$, we first divide $N_t(x) \cup N_t(R)$ into $3$ sets: $C := N_t(x) \cap N_t(R)$, $X:= N_t(x) \setminus N_t(R)$ and $Y := N_t(R) \setminus N_t(x)$. Note that the inductive hypothesis tells that $|X| = O(n^{t\delta})$, $|Y| = O(n^{(t+1)\delta})$ and $|C| = O(n^{(t-1)\delta}\log n)$.
        
        The vertex appeared in $N_{t+1}(x) \cap N_{t+1}(R)$, is either incident to set $C$, or incident to both $X$ and $Y$ at the same time.  Thus we have following cases:
        \begin{enumerate}
            \item $w$ incident to $C$. The number of vertices $w \in N_{t+1}(x) \cap N_{t+1}(y)$ s.t. $w$ is incident to $C$, is bounded by total size of neighborhood incident to $C$ and it is $|C| \cdot \Theta(n^\delta) = O(n^{t\delta}\log n)$ by the inductive hypothesis.
            \item $w$ incident to both $X$ and $Y$. Similar as before, we partition all vertices into $4$ cases:
            \begin{enumerate}
                \item $w \in \Det(X) \cap \Det(Y)$. According to \Cref{def-fav cond1}, the set $\Det(X)$ contains at most $O(|X|\cdot \log n) = O(n^{t\delta} \log n)$ vertices incident to $X$, thus this part contributes at most $O(n^{t\delta} \log n)$ vertices.
                \item $w \in \Det(X) \setminus \Det(Y)$. As above, this part contributes at most $O(n^{t\delta} \log n)$ vertices.
                \item $w \in \Det(Y) \setminus \Det(X)$. According to \Cref{def-fav cond1}, the set $\Det(Y)$ contains at most $O(|Y|\cdot \log n) = O(n^{(t+1)\delta} \log n)$ vertices incident to $Y$, denoted by set $N$. For $w \in N$, we have $w \notin \Det(X)$ thus w is incident to any vertex in $\partial_t(x)$ independently. Thus in expectation, there are 
                \begin{align*}
                    |N| \cdot (1-(1-p)^{|\partial_t(x)|}) &\leq |N|\cdot|\partial_t(x)|\cdot p \tag{Bernoulli's inequality} \\
                    &= O(\frac{n^{2(t+1)\delta}\log n}{n}) \\
                    &= O(\frac{n^{2k\delta}\log n}{n}) \tag{by $t \leq k-1$}\\
                    &= O(\frac{\log n}{n^\delta}) \tag{by $n^{k\delta} \leq \sqrt{n^{1-\delta}}$}
                \end{align*}
                vertices in $N$ incident to $\partial_t(x)$. By Chernoff-Hoeffding bound, this part contributes at most $O(\log n)$ vertices.
                \item $w \notin \Det(X) \cup \Det(Y)$. For every such vertex $w$, $w$ is incident to every vertex in $X \cup Y$ with independent probability $p$. Similarly, in expectation there are at most
                \begin{align*}
                    n \cdot (1-(1-p)^{|X|}) \cdot (1-(1-p)^{|Y|}) &\leq n \cdot |X|p \cdot |Y|p \tag{Bernoulli's inequality}\\
                    &= O(n^{2(t+1)\delta} \cdot \frac{n^\delta}{n}) \\
                    &= O(n^{2k\delta} \cdot \frac{n^\delta}{n}) \tag{$t \leq k-1$}\\
                    &= O(1) \tag{by $n^{k\delta} \leq \sqrt{n^{1-\delta}}$}
                \end{align*}
                vertices incident to $X$ and $Y$ at the same time, using Bernoulli's inequality. Still, by Chernoff-Hoeffding bound, this part contributes at most $O(\log n)$ vertices with high probability.
            \end{enumerate}
        Note that all above events fail with arbitrarily small probability polynomial in $\frac{1}{n}$, by choosing proper constant. Putting together, with high probability, we have $\abs{N_{t+1}(x) \cap N_{t+1}(R)} = O(n^{t\delta}\log n)$. Thus the second item also holds for $t+1$.
        \end{enumerate}
    
    \end{itemize}
    The base case and inductive step together finish the proof.

\end{proof}

For convinience, we restate \Cref{lem-fav2tree 1} here and prove it.
\lemfavtotreeA*
\begin{proof}
    We prove by showing that every vertex is connected to vertex $n$ in subgraph $H$.
    
    Conditioned on the event that $G$ satisfies a favorable condition $B(f,b)$, the increasing path $IP(v)$ of any vertex $v$ is characterized by $f$ (see \Cref{lem-uniqIP}). Note that if \textsc{FindPath}$(G,S,\theta,v)$ succeeds in finding a vertex in core tree $T_0$, then $v$ is a good vertex and is connected to $n$ by \Cref{obs-good vertex}. Thus next we only consider bad vertex $w$ such that \textsc{FindPath}$(G,S,\theta,w)$ returns $\bot$. 

    According to \Cref{obs-bad to good}, if there is at least one good vertex appearing in $IP(w)$, then it is connected to $n$.

    Now, it seems that we need to argue that $IP(w)$ always contains at least one good vertex. However, if $w$ has a large ID, then its increasing path $IP(w)$ is very short, thus it does not happen with a large chance to meet a good vertex on $IP(w)$. On the other hand, if $w$ has a small ID, it has a large chance to meet a good vertex on its increasing path and be connected to $n$, stated as follows. The proof is deferred later in this section.
    \begin{claim}\label{clm-small vertex connected1}
        With high probability, for every vertex $v \leq K$, $v$ is connected to $n$ in subgraph $H$ .
    \end{claim}

We now show that every bad vertex $w$ is connected to $n$ in the subgraph $H$, even if its increasing path $IP(w)$ does not contain any good vertex. For such vertices, we make the following key observation:

\begin{observation}
If $w$ is a bad vertex and $IP(w)$ does not contain any good vertex, then every edge on $IP(w)$ is retained in $H$.
\end{observation}

This follows from the fact that, for such a $w$, every vertex on $IP(w)$ determines its proposal edge during \texttt{Phase 3}. As a result, each of them proposes an edge to its neighbor $\Gamma^+(w)$, ensuring that all edges on $IP(w)$ are retained.

This observation implies that $w$ is connected in $H$ to the last vertex on $IP(w)$. Since the input graph $G$ satisfies the favorable condition $B(f, b)$, by \Cref{def-fav cond1}, the last vertex $z$ on $IP(w)$ satisfies $z > n - K$. Moreover, as $z$ is the final vertex on the path, i.e., $\Gamma^+(z) = \emptyset$, and $f(z) = \bot$. Therefore, the proposal edge of $z$ must be $\smallest{\Gamma(z)} = b(z)$.

Since $w$ is connected to $z$, and $z$ proposes to $b(z)$, it follows that $w$ is also connected to $b(z)$. By the definition of the favorable condition, $b(z) \leq K$, so $b(z)$ is a small vertex. Applying \Cref{clm-small vertex connected1}, we conclude that $b(z)$ is connected to $n$ in $H$ with high probability. Thus, since $w$ is connected to $b(z)$, it is also connected to $n$ in $H$.

We have now established that every vertex is connected to $n$ in $H$ with high probability, by \Cref{clm-small vertex connected1}. Furthermore, since there are exactly $n - 1$ proposal edges, $H$ contains $n - 1$ edges and is therefore a spanning tree.

Finally, note that even when conditioned on $B(f,g)$, the graph $G$ still has maximum degree $d = 2np$ and conductance $\phi = \Omega(1)$ with high probability. Thus, the lazy random walk in \texttt{Phase 1} succeeds with high probability. By a union bound over this and \Cref{clm-small vertex connected1}, we conclude that $H$ is a spanning tree with probability at least $1 - \frac{1}{n^2}$.

\end{proof}

In the following, we prove \Cref{clm-small vertex connected1} to finish this section.

\begin{proof}[Proof of \Cref{clm-small vertex connected1}]
    Firstly, for small vertex, we have following observation.
    \begin{observation}
        If $v \leq K$, then the length of $IP(v)$ is at least $\Omega(n^\delta /\log n)$.
    \end{observation}
    \begin{proof}
        This is because $G$ satisfies favorable condition $B(f,g)$, thus for every $u$ on $IP(v)$, we have $f(u) \leq u + K$. It means that the ID of each vertex on $IP(v)$ is at most K greater than the previous one. Furthermore, we know that the last vertex $z$ on $IP(v)$ satisfies $z \geq n-K$. Together, we see that $IP(v)$ passes through at least 
        \[
        \frac{n-K-v}{K}=\Omega(n^\delta /\log n)
        \]
        vertices.
    \end{proof}

    For every vertex $w$ on $IP(v)$, the invocation of \textsc{FindPath}$(G,S,\theta,w)$ performs a BFS starting at $w$ to find the anchor of $w$. We next prove that the set of distinct vertices visited during all invocations of \textsc{FindPath} on $IP(v)$, denoted by $Q(v)$, is of size at least $\Omega(\sqrt{n^{1+\delta}})$. 

    \begin{proposition}
    $ |Q(v)| = \Omega(\sqrt{n^{1+\delta}}) $.
    \end{proposition}
\begin{proof}
To bound the size of $ Q(v) $, let $ k $ be the integer defined in \Cref{lem-gnp bfs}. Then, for any vertex $ w $ on $ IP(v) $, a $ k $-level BFS visits only $ |N_k(w)| = \Theta(n^{k\delta}) = O(\sqrt{n^{1-\delta}}) \ll \theta $ distinct vertices, where $ \theta = \Theta(\sqrt{n^{1-\delta}} \log^2 n) $ is the threshold defined in \Cref{alg-gnp prop1}. This implies that the BFS in \textsc{FindPath}$(G, S, \theta, w)$ completes the $ k $-th level and terminates before finishing the $ (k+1) $-th level. Most of the distinct vertices visited by $ w $ appear in this $ (k+1) $-th level.

To show that $ Q(v) $ is large, consider an arbitrary set $ R \subseteq IP(v) $ of size $ \Theta(n^\delta / \log n) $. This is feasible because $ IP(v) $ has length $ \Omega(n^\delta / \log n) $.

Fix a vertex $ x \in R $, and let $ R \gets R \setminus \{x\} $. Define the set 
\[
C := N_k(x) \cap N_k(R),
\]
which represents the intersection of the $ k $-level BFS from $ x $ and from vertices in $ R $. By \Cref{lem-gnp bfs}, we have $ |C| = O(n^{(k-1)\delta} \log n) $.

Let $ \partial_t(x) := N_t(x) \setminus N_{t-1}(x) $ denote the set of vertices at distance exactly $ t $ from $ x $. Again by \Cref{lem-gnp bfs}, we have
\[
|\partial_t(x)| = |N_t(x)| - |N_{t-1}(x)| = \Theta(n^{t\delta}).
\]
Therefore,
\[
|\partial_t(x) \setminus C| \geq |\partial_t(x)| - |C| = \Theta(n^{t\delta}),
\]
which means that a large portion of $ \partial_t(x) $ is only reached by the BFS from $ x $ and not by any BFS from vertices in $ R $.

In the $ (k+1) $-th level, vertices in $ \partial_t(x) $ are sorted, and their neighborhoods are probed in order. A worst-case scenario would be if the vertices in $ C $ appear early in the order, leading to many redundant probes. However, even if all vertices in $ C $ are fully explored, the number of distinct vertices discovered would be at most
\[
|C| \cdot \Theta(n^\delta) = O(n^{k\delta} \log n) = O(\sqrt{n^{1-\delta}} \log n).
\]
Since \Cref{alg-gnp prop1} requires each BFS to discover at least $ \theta = \Theta(\sqrt{n^{1-\delta}} \log^2 n) $ vertices, exploring only the neighborhoods of vertices in $ C $ is insufficient.

As a result, during the $ (k+1) $-th level of BFS from $ x $, there must be at least 
\[
\frac{\theta}{\Theta(n^\delta)} = \Theta\left( \frac{\sqrt{n^{1-\delta}} \log^2 n}{n^\delta} \right)
\]
vertices whose neighborhoods are fully explored only during the invocation \textsc{FindPath}$(G, S, \theta, x)$, and not by any invocation from vertex in $ R $. This argument can be applied inductively to every $ r \in R $.

Therefore, $ Q(v) $ must include the neighborhoods fully explored from 
\[
\Theta\left( \frac{\sqrt{n^{1-\delta}} \log^2 n}{n^\delta} \right) \cdot |R| = \Theta(\sqrt{n^{1-\delta}} \log n)
\]
distinct vertices. Since these vertices share few neighbors, we conclude that
\[
|Q(v)| = \Omega(n^\delta \cdot \sqrt{n^{1-\delta}}) = \Omega(\sqrt{n^{1+\delta}}),
\]
up to a logarithmic loss due to overlapping neighborhoods.
\end{proof}

    Now that $Q(v)$ is a large set of vertices, by \Cref{lem-end vertex} and choosing proper constant, we have
    \[
    |Q(v)| \cdot r = \Omega(\sqrt{n^{1+\delta}}) \cdot \Theta(\sqrt{n^{1-\delta}}\log n)  \geq 100 n\log n,
    \]
    hence $|Q(v) \cap S| \geq 1$ with probability at least $(1 - \frac{1}{n^{10}})$ (note that $r$ is the number of random walks and $S$ is the set of vertices visited during random walks in \Cref{alg-gnp prop1}). Therefore, there is at least one vertex $u$ on $IP(v)$ s.t. \textsc{FindPath}$(G,S,\theta,u)$ succeeds in finding a vertex in $S$, and such $u$ is a good vertex. Together with \Cref{obs-bad to good}, we conclude that $v$ is connected to $n$ in $H$ with high probability by a union bound over all vertices $v \leq K$.
\end{proof}

\subsection{Proof of \Cref{thm:avg lca 2}}
Note that for $\delta = 1$ and $np = n^\delta$, $G \sim G(n,p)$ is necessarily a complete graph. In this case, a trivial LCA suffices: simply maintaining all edges incident to $n$ already yields a spanning tree LCA. Therefore, in the following, we focus on the regime $1 > \delta > \frac{1}{3}$.

In this range, the situation changes because an algorithm cannot even fully explore the neighborhood of a given vertex using only $\tilde O(\sqrt{n^{1-\delta}})$ probes. Consequently, for vertices not directly connected to $n$, the method used in \cref{thm:avg lca 1} to maintain increasing paths is no longer fully feasible.

A key difference from \cref{alg-gnp prop1} is that, instead of exploring the entire neighborhood to build an increasing path, we only sample $\tilde\Theta(\sqrt{n^{1-\delta}}) = o(n^\delta)$ neighbors\footnote{Sometimes we refer to the probability over neighbor sampling.} of a vertex $v$ and select the proposal edge for $v$ among these sampled neighbors. This yields a ``compromised'' version of the increasing path, which is relatively short. On the other hand, as $\delta$ increases, a vertex is more likely to be directly connected to $n$. By leveraging both observations, even a small vertex retains a significant probability of being connected to $n$ via the compromised increasing path.

\subsubsection{Algorithm Overview}

\begin{algorithm}[htb]
\caption{\textsc{Prop}$(G,u)$: compute the proposal edge of vertex $u$}\label{alg-gnp prop2}

\KwIn{graph $G = (V, E) \sim \Gnp$ with $np=n^\delta$ for $\delta \in (\frac{1}{3},1)$; $u\in V\setminus\{n\}$.}
\KwOut{the proposal edge of $u$.}

\tcp{\texttt{Phase 1:} check if $u$ is incident to $n$}
\If{\texttt{Exists}$(u, n)$}{
\Return $(u, n)$.
}

Let $q := \frac{\sqrt{n^{1-\delta}}\log^2 n}{n^\delta}$;

Randomly probe each neighbor of vertex $n$ with probability $q$. Denote these vertices by set $S$.

\tcp{\texttt{Phase 2:} check if $u$ is incident to $S$}

\For{$v \in S$ by increasing order}{
\If{\texttt{Exists}$(u, v)$}{
\Return $(u, v)$.
}
}

\tcp{\texttt{Phase 3:} maintain restricted increasing path}

Randomly probe each neighbor of vertex $u$ with probability $q$. Denote these vertices by set $\Gamma_{q}(u)$.

\If {$\Gamma^+_q(u) \neq \emptyset$}{\Return $(u,v)$, where $v = \smallest{\Gamma^+_q(u)}$.}
\Else {\Return $(u,v)$, where $v = \smallest{\Gamma_q(u)}$.}
\end{algorithm}

\noindent
\Cref{alg-gnp prop2} outputs the proposal edge of the vertex $u$, and we denote by $H$ the subgraph induced by all proposed edges. The algorithm follows the scheme of \Cref{alg-gnp prop1} but replaces full neighborhood exploration with vertex-pair probes and neighbor sampling to overcome obstacles caused by large degrees.

Using a shared random tape $R$, neighbors of $u$ are sampled as follows: first, the degree $d_u$ is probed, then for each index $i \in [d_u]$, $\Nbr(u,i)$ is probed independently with probability $q$. The same random tape $R$ is always used to ensure consistency of the neighbor sampling.

We use the following definitions to facilitate understanding of our algorithm.

\begin{definition}
Given a random tape $R$ and a probability $0<q<1$, for any vertex $u \in V$, let $\Gamma_q(u)$ denote the set of neighbors of $u$ sampled according to $q$ and $R$. Furthermore, define
\[
\Gamma_q^+(u) := \{v \in \Gamma_q(u) \mid v > u\}.
\]
Equivalently, $\Gamma_q(u)$ is obtained by tossing a $q$-biased coin for each neighbor of $u$, using the shared random tape $R$ to determine whether the neighbor is probed.
\end{definition}

\begin{definition}[Restricted Increasing Path]
Let $q$ and $R$ be as above. A path $(v_1, v_2, \dots, v_\ell)$ is called a \emph{$q$-restricted increasing path} if for every $i \in [\ell-1]$, 
\[
v_{i+1} = \smallest{\Gamma_q^+(v_i)}.
\]
We denote by $RIP_q(v)$ the longest $q$-restricted increasing path starting from vertex $v$.
\end{definition}

In \texttt{Phase 1} of \Cref{alg-gnp prop2}, for any vertex $u$, the algorithm first checks whether $u$ is incident to $n$. If so, the proposal edge of $u$ is $(u,n)$. Otherwise, the algorithm randomly probes each neighbor of $n$ with probability $q$, denoting the set of sampled neighbors by $S$. Note that for every vertex $v \in S$, its proposal edge must be the edge $(v,n)$.

In \texttt{Phase 2}, we check whether $u$ is incident to $S$. If so, $u$ proposes its first incident edge to $S$. We call vertices incident to $n$ or to $S$ \textbf{good vertices}, denoted by $V_{good}$. The following observation formalizes their structure.

\begin{observation}\label{obs-good vertex2}
All proposal edges of vertices in $V_{good}$, as determined by \Cref{alg-gnp prop2}, form a spanning tree of $V_{good}$.
\end{observation}

\begin{proof}
After \texttt{Phase 1} and \texttt{Phase 2}, each good vertex is either directly incident to $n$ or to a neighbor of $n$, so it is connected to $n$ in $H$. Moreover, each vertex in $V_{good}$ proposes exactly one edge, resulting in $|V_{good}|-1$ edges. Hence, the proposal edges collectively form a spanning tree rooted at $n$.
\end{proof}

In \texttt{Phase 3}, for a bad vertex $w$, the algorithm randomly explores neighbors of $w$ and proposes the edge along $RIP_q(w)$. If $RIP_q(w)$ contains any good vertex, then $w$ is also connected to $n$ in $H$. Formally, we have the following observation. Its proof is similar to that of \Cref{obs-bad to good} and is omitted.

\begin{observation}\label{obs-bad to good2}
If any good vertex appears in $RIP_q(u)$, then the bad vertex $u$ is connected to $n$ in $H$.
\end{observation}

The \emph{worst-case} scenario for a vertex $w$ occurs when $RIP_q(w)$ does not contain any good vertex. In this case, \texttt{Phase 3} sets the proposal edge of the last vertex on $RIP_q(w)$ to be incident to a small vertex in $\Gamma_q^+(w)$.

Another difference from \Cref{alg-gnp prop1} is that we cannot find the smallest neighbor of $w$ using only $\tilde\Theta(\sqrt{n^{1-\delta}}) = o(n^\delta)$ probes. Instead, the algorithm samples neighbors to find, with high probability, one of the smallest neighbors; as we will show later, this suffices. Through this edge, the ``hope'' of connecting $w$ to $n$ is passed to a small neighbor. Thus, the problem reduces to proving that the restricted increasing path of a small vertex is still likely to encounter a good vertex.

It is easy to see that the query complexity of \Cref{alg-gnp prop2} is dominated by the random probing of neighbors of vertex $n$, giving a complexity of $O(\sqrt{n^{1-\delta}} \log^2 n)$ with high probability. We next show that, with high probability, the subgraph $H$ forms a spanning tree.

\subsubsection{Correctness of \Cref{alg-gnp prop2}}
In this section, let $t = \Theta(\frac{1}{q}\log n)$ be an integer. Following definition characterize neighbors of $n$ and the concerned $t$ neighbors of every other vertex. This provides us a view of the random \Cref{alg-gnp prop2}'s behavior through a a deterministic way.

\begin{definition}\label{def-bc2}
    A \textbf{baseline condition} is an event parameterized as $B(C,f^t,b^t)$ where
    \begin{itemize} 
        \item $C \in 2^V$, is the set of neighbors of vertex $n$;
        \item (Forward function) $f^t: V\setminus (C \cup \{n\}) \to V^t\cup\{\bot\}$, is a function mapping vertex $v$ to the set of $t$ smallest neighbors of $v$, such that for each $u\in f^t(v)$, $u > v$; or to $\bot$ if the number of such neighbors is less than $t$;
        \item (Backward function) $b^t: \{v \mid f^t(v) = \bot\} \to V^t \cup \{\bot\}$, is a function mapping vertex $v$ (if $f(v) = \bot$) to the set of $t$ smallest neighbors of $v$; or to $\bot$ if the number of such neighbors is less than $t$;
    \end{itemize}

    We say a graph $G$ satisfies $B(C,f^t,b^t)$ if it is compatible simultaneously with $C$, $f^t$ and $b^t$. 
\end{definition}

\begin{lemma}\label{lem-unq cond2}
    For every pair of distinct baseline conditions $B_1 = B(C_1,f^t_1,b^t_1)$ and $B_2 = B(C_2,f^t_2,b^t_2)$, there is no graph $G$ that satisfies both.
\end{lemma}
\begin{proof}
    Consider that $B_1$ and $B_2$ are distinct, if $C_1 \neq C_2$, then $G$ does not satisfies both because the neighborhood of vertex $n$ is unique.
    
    Assume $C_1 = C_2$, then there must exist at least one vertex $v$ s.t. either $f^t_1(v) \neq f^t_2(v)$ or $b^t_1(v) \neq b^t_2(v)$. Recall that function $f^t$ maps $v$ to the set of $t$ smallest larger neighbors of $v$, and function $b^t$ maps $v$ to the set of $t$ smallest neighbors of $v$. Thus if $\Gamma(v)$ is compatible with $f_1$ and $b_1$, then it must be not compatible with $f_2(v)$ and $b_2(v)$.
\end{proof}

Following observation characterizes the pairs of vertices determined by condition $B(C,f^t,b^t)$.
\begin{observation}\label{obs-pair random2}
    For $G \sim \Gnp$, if $G$ satisfies a baseline condition $B(C,f^t,b^t)$, then we have:
    \begin{enumerate}
        \item for every $v \in C$, pair $(v,n) \in E(G)$. Besides, for every $v \notin C$, pair $(v,n) \notin E(G)$;
        \item if $f^t(v) \neq \bot$, then pair $(v,u) \in E(G)$ for $u \in f^t(v)$. Besides, pair $(v,u) \notin E(G)$ for vertex $v < u <\largest{f^t(v)}$, except for $u \in f^t(v)$;
        \item if $f^t(v) = \bot$ and $b^t(v) \neq \bot$, then pair $(v,u) \in E(G)$ for $u \in b^t(v)$ and there are at most $t-1$ vertices incident to $v$ with ID larger than $v$. Besides, pair $(v,u) \notin E(G)$ for vertex $1 \leq u < \largest{b^t(v)}$, except for $u \in b^t(v)$;
    \end{enumerate}
    and each remaining pair of vertices is connected by an edge independently with probability $p$.
\end{observation}

We still use $V_I$ to denote the set of vertices induced by an interval $I$ (see \Cref{def-int induced set}). Following definition characterize desirable $B(C,f^t,b^t)$ under which the algorithm behaves well.
\begin{definition}\label{def-fav cond2}
     Let integer $K = c n^{1-\delta} \log n$, where $c$ is a large enough constant. A baseline condition $B(C,f^t,b^t)$ is called a \textbf{favorable condition} if all of the following hold:
     \begin{enumerate}
         \item[(1)] for every integer $0 \leq i \leq \frac{n}{K}-1$, it holds that $\abs{C \cap V_{[iK+1,iK+K]}} \leq 2c\log n$;
         \item[(2)] for every $v$ such that $1 \leq v \leq n - tK$, it holds that $f^t(v) \neq \bot$ and $\largest{f^t(v)} \leq v + tK$;
         \item[(3)] for every $v$ such that $f^t(v) = \bot$, it holds that $b^t(v) \neq \bot$ and $\largest{b^t(v)} \leq tK$.
     \end{enumerate}
\end{definition}

In the following, we try to argue that if a graph satisfies a favorable condition, then \Cref{alg-gnp prop2} works with high probability. We first use following claim to show that w.h.p. $G \sim \Gnp$ satisfies a favorable condition.

\begin{claim}\label{clm-fav graph2}
    If $G \sim G(n,p)$ satisfies following conditions simultaneously:
        \begin{enumerate}
            \item[(a)] for every integer $0 \leq i \leq \frac{n}{K}-1$, vertex $n$ is incident to at most $2c\log n$ vertices in $V_{[iK+1,iK+K]}$;
            \item[(b)] every vertex $v$ is incident to at least $2t$ vertex in $V_{(v,v+tK]}$.
            
        \end{enumerate}
    then the baseline condition $B(C,f^t,b^t)$ that $G$ satisfies is a favorable condition. We call such $G$ a favorable graph.
\end{claim}
\begin{proof}
    Recall that $C$ is the set of $n$'s neighbors, hence $(a)$ is exactly Item $(1)$ in \Cref{def-fav cond2}.

    As $(b)$ is satisfied, for $1 \leq v \leq n-tK$, $v$ is incident to at least $2t$ vertices in set $V_{(v,v+tK]}$ where
    \[
     V_{(v,v+tK]} = \{v+1, v+2, \cdots, v+tK\}.
    \]
    This means, $v$ has at least $2t$ neighbors larger then $v$; and the largest ID among these $2t$ neighbors is at most $v+tK$. Hence $(b)$ implies Item $(2)$ in \Cref{def-fav cond2}. Furthermore, it means that if $f^t(v) = \bot$, then $v > n-tK$.

    In a graph satisfying these conditions, for $v$ such that $f^t(v) = \bot$, $v$ must be a large vertex with ID larger than $n-tK$ due to above discussions. In this case, the set $V_{(v,v+tK]}$ satisfies
    \[
     V_{(v,v+tK]} = \{v+1, v+2, \cdots, n, 1,\cdots, v+tK-n\}.
    \]
    Thus $v$ is incident to at least $2t$ vertices in above set $V_{(v,v+tK]}$ by $(b)$. On the other hand, as $f^t(v) = \bot$, the number of $v$'s neighbors with ID larger than $v$ is at most $t-1$. Together, we see that $v$ is incident to at least $t+1$ vertices in set $\{1,\cdots, v+tK-n\}$. As we have $v+tK-n \leq tK$, the baseline condition $B(C,f^t,b^t)$ that $G$ satisfies must satisfy Item $(3)$ in \Cref{def-fav cond2}.
\end{proof}

Now we prove that for graph $G \sim \Gnp$, it satisfies some favorable condition $B(C,f^t,b^t)$ with high probability.

\begin{lemma}\label{lem-fav graph 2}
    For $G \sim G(n,p)$:
    \begin{equation*}
        \sum_{\text{fav. } B(C,f^t,b^t)} \Pr[G \text{ satisfies $B$}] \geq 1- \frac{1}{n^2}
    \end{equation*}
\end{lemma}
\begin{proof}
    We prove by showing that $G$ is a favorable graph in \Cref{clm-fav graph2} with high probability. Thus by applying \Cref{clm-fav graph1}, we see that with high probability a favorable condition is satisfied.

    We first specify the neighbors of $n$ to prove that w.h.p. $(a)$ in \Cref{clm-fav graph2}, is satisfied. This is because, for any given integer $0 \leq i \leq \frac{n}{K}-1$, a vertex $w \in V_{[iK+1,iK+K]}$ is incident to $n$ with independent probability $p$. Thus in expectation, set $V_{[iK+1,iK+K]}$ contains 
    \[
    K \cdot p = c n^{1-\delta} \log n \cdot \frac{n^\delta}{n} = c \log n
    \]
    vertices incident to $n$ for a large constant $c$. This means with probability at least $1 - \Omega(\frac{1}{n^c})$, $n$ is incident to at most $2c \log n$ vertices in $n$. A union bound over $0 \leq i \leq \frac{n}{K}-1 = O(n^\delta)$ proves that $(a)$ is satisfied.
    
    For condition $(b)$ in \Cref{clm-fav graph2}, given vertex $v$, note that $\abs{V_{(v,v+tK]}} = tK$. Thus similarly there are 
    \[
    tK \cdot p = t \cdot c n^{1-\delta} \log n \cdot \frac{n^\delta}{n} = c t\log n
    \]
    vertices incident to $v$ for large constant $c$ in expectation. By a simple argument with Chernoff bound and union bound, we see that every vertex $v$ is incident to at least $2t$ vertices in $V_{(v,v+tK]}$.

    As both conditions hold with at least $1- \Omega(\frac{1}{n^{100}})$ by choosing proper constant $c$, we finish the proof by a union bound.

\end{proof}

Following lemma tells that if $G \sim \Gnp$ satisfies a favorable condition $B(C,f^t,b^t)$, then \Cref{alg-gnp prop2} works with high probability. The proof is deferred in \Cref{sec-proof2}.
\begin{restatable}{lemma}{lemfavtotreeB}\label{lem-fav2tree 2}
    For every favorable condition $B(C,f^t,b^t)$,
    \begin{equation*}
        \Pr_{G \sim G(n,p)}[H \text{ is a spanning tree} \mid G \text{ satisfies } B] \geq 1- \frac{1}{n^2},
    \end{equation*}
    where $H$ is the subgraph consists of all proposal edges determined by \Cref{alg-gnp prop2}.
\end{restatable}

Now we prove the main theorem of this section.
\begin{theorem}
    For $np = n^\delta$ and $1 >\delta > \frac{1}{3}$, let $H$ be the subgraph consists of all proposal edges determined by \Cref{alg-gnp prop2} on $G \sim \Gnp$. Then $H$ is a spanning tree with high probability.
\end{theorem}
\begin{proof}
    We have
    \begin{align*}
        \quad &\Pr_{G \sim G(n,p)}[\text{$H$ is a spanning tree}] \\
        &= \sum_{B(f,g)} \Pr[\text{$H$ is a spanning tree} \mid \text{$G$ satisfies $B$}] \cdot \Pr[\text{$G$ satisfies $B$}] \tag{\Cref{lem-unq cond2}}\\
        &\geq \sum_{\text{fav. $B(f,g)$}} \Pr[\text{$H$ is a spanning tree} \mid \text{$G$ satisfies $B$}] \cdot \Pr[\text{$G$ satisfies $B$}] \\
        &\geq (1-\frac{1}{n^2}) \sum_{\text{fav. $B(f,g)$}} \Pr[\text{$G$ satisfies $B$}] \tag{\Cref{lem-fav2tree 2}}\\ 
        &\geq (1-\frac{1}{n^2})(1-\frac{1}{n^2}) \tag{\Cref{lem-fav graph 2}}\\
        &\geq 1-\frac{1}{n}.
    \end{align*}
\end{proof}

Using consistent random bits over all possible queries, this implies that \Cref{alg-gnp prop2} is an average-case LCA for spanning tree, and this finishes the proof of \Cref{thm:avg lca 2}.

\subsubsection{Proof of \Cref{lem-fav2tree 2}}\label{sec-proof2}
In this section, we prove \Cref{lem-fav2tree 2} using tools we developed in last section.

\lemfavtotreeB*
\begin{proof}
We prove by showing that every vertex is connected to vertex $ n $ in the subgraph $ H $. Note that every vertex $ v \in V \setminus \{n\} $ falls into one of the following three cases:
\begin{enumerate}
    \item $ v $ is incident to $ n $;
    \item $ v $ is not incident to $ n $, but is incident to a vertex in $ S $;
    \item $ v $ is not incident to either $ n $ or any vertex in $ S $.
\end{enumerate}

If $ v $ belongs to the first or second case, then by \Cref{alg-gnp prop2}, it is easy to see that $ v $ is connected to $ n $ in $ H $. We refer to such vertices as \emph{good}. For vertices of the third case (i.e., \emph{bad} vertices), we rely on the following claim, whose proof is deferred to later in this section.

\begin{claim}\label{clm-small vertex connected2}
With high probability, every vertex $ v \leq tK $ is connected to $ n $ in the subgraph $ H $.
\end{claim}

For bad vertices $ w > tK $, \Cref{alg-gnp prop2} samples neighbors of $ w $ to obtain the set $ \Gamma_q(w) $, and retains the incident edge to $ \smallest{\Gamma^+_q(w)} $, which lies on the restricted increasing path $ RIP(w) $ of $ w $. If at least one good vertex appears on $ RIP(w) $, then $ w $ is already connected to $ n $. Otherwise, we show that $w$ is connected to a small vertex with ID at most $tK$.

By construction, for such vertex $w$, every edge on $ RIP(w) $ is retained in $ H $. This yields the following observation.

\begin{observation}
For every bad vertex $ w $, either $ w $ is connected to $ n $ in $ H $ through a good vertex on $ RIP(w) $, or every edge on $ RIP(w) $ is retained in $ H $.
\end{observation}

Recall that the input graph $ G $ satisfies the favorable condition $ B(C, f^t, b^t) $, which ensures that the last vertex $ z $ on $ RIP(w) $ satisfies $ z > n - tK $ with probability at least $ 1 - \frac{1}{n^{100}} $ over randomness of neighbor sampling. Indeed, by \Cref{def-fav cond2}, if $ z \leq n - tK $, then $ \Gamma^+(z) $ must contain at least $ t $ vertices. Since each of these appears in $ \Gamma^+_q(z) $ independently with probability $ q $, and $ t = \Theta\left( \frac{1}{q} \log n \right) $, we have $ \Gamma^+_q(z) = \emptyset $ only with very small probability. Therefore, with high probability, $ z > n - tK $.

Moreover, since $ z $ is the last vertex on $ RIP(w) $, we must have $ \Gamma^+_q(z) = \emptyset $, and thus its proposal edge is $ \smallest{\Gamma_q(z)} $. With high probability, we have $ \smallest{\Gamma_q(z)} \in b^t(z) $, since $ b^t(z) $ contains the $ t = \Theta\left( \frac{1}{q} \log n \right) $ smallest neighbors of $ z $, each sampled independently with probability $ q $.

Because $ \largest{b^t(z)} \leq tK $ under the favorable condition, vertex $ z $ is connected to a vertex $ x := \largest{b^t(z)} $ with ID at most $ tK $. Since $ w $ is connected to $ z $, it is also connected to $ x $. By \Cref{clm-small vertex connected2}, $ x $ is connected to $ n $ in $ H $ with high probability. Therefore, $ w $ is also connected to $ n $ in $ H $.

Applying a union bound over \Cref{clm-small vertex connected2} and all the high-probability events described above, we conclude that every vertex is connected to $ n $ in $ H $ with probability at least $ 1 - \frac{1}{n^2} $. Additionally, since there are exactly $ n - 1 $ proposal edges, the subgraph $ H $ contains $ n - 1 $ edges and is therefore a spanning tree.

\end{proof}

Now we prove \Cref{clm-small vertex connected2} to finish this section.
\begin{proof}[Proof of \Cref{clm-small vertex connected2}]
    Conditioned on a favorable condition $B(C,f^t,b^t)$, $G$ must be compatible with $C$, $f^t$ and $b^t$ satisfying \Cref{def-fav cond2}. Our proof follows two steps: $\texttt{(1)}$ show that the restricted increasing path starting at $v$ is long enough; and $\texttt{(2)}$ show that w.h.p. there is an edge between vertex on the path and core set $S$.
    
    \paragraph{On the Length of Restricted Increasing Paths from $v$.}
    For bad vertices $v$, \Cref{alg-gnp prop2} would sample neighbors of $v$ to obtain the set $\Gamma_q(v)$ and maintain incident edge to $\smallest{\Gamma^+_q(v)}$, which is the edge on restricted increasing path for $v$. We first show that the restricted increasing path $P := RIP(v) = (v_1=v,v_2,\cdots,v_\ell)$ of $v$ satisfies: 
    \begin{proposition}\label{prop-rip}
        For every $1 \leq i \leq l-1$, $v_{i+1} \in f^t(v_i)$; and specifically, $f^t(v_\ell) = \bot$.
    \end{proposition}
    \begin{proof}
        This is directly implied by following observation.
        \begin{observation}
            With high probability, for every vertex $v \leq n-tK$, $\smallest{\Gamma^+_q(v)} \in f^t(v)$. Here $\Gamma_q(v)$ is the neighbors of $v$ sampled in \Cref{alg-gnp prop2}.
        \end{observation}
        \begin{proof}
            Note that we set $t = \Theta(\frac{1}{q} \log n) = c' \frac{1}{q} \log n$, where $c'$ is a large constant. Furthermore, we have $f^t(v) \neq \bot$, because $v \leq n-tK$ and $f^t$ satisfies \Cref{def-fav cond2}. Hence, $|f^t(v)| = t$. And $q$ is the probability that any given neighbor of $v$ is sampled. 
            
            Thus in expectation, there are $|f^t(v)| \cdot q = c'\log n$ vertices in $f^t(v)$ are sampled, and these vertices belong to $\Gamma_q(v)$. By choosing a large constant $c'$ and using Chernoff bound, it is easy to see that w.h.p. for any vertex $v$, we have $\smallest{\Gamma^+_q(v)} \in f^t(v)$. A union bound over all vertices $\leq n-tK$ finishes the proof.
        \end{proof}
        Using the observation, we have
        \begin{align*}
            & v_2 = \smallest{\Gamma^+_q(v_1)} \in f^t(v_1); \\
            & v_3 = \smallest{\Gamma^+_q(v_2)} \in f^t(v_2); \\
            & \vdots \\
            & v_\ell = \smallest{\Gamma^+_q(v_{l-1})} \in f^t(v_{l-1});
        \end{align*}
        until the last vertex $v_\ell$ such that $f^t(v_\ell) = \bot$. Hence this proves our \Cref{prop-rip}.
    \end{proof}

    As $B(C,f^t,b^t)$ is a favorable condition, we have $\largest{f^t(v)} \leq v + tK$ hold for $v \leq n-tK$. Thus by \Cref{prop-rip}, we have
    \[
    \text{$v_{i+1} \leq \largest{f^t(v_i)} \leq v_i + tK$, for every $1 \leq i \leq \ell-1$.}
    \]
    Furthermore, note that $v_\ell > n - K$ because $f^t(w) \neq \bot$ for $w \leq n-tK$ under favorable condition. Putting together, we have $n - tK < v_\ell \leq ltK$, which means
    \[
    \ell \geq \frac{n}{tK}-1 = \Theta(\sqrt{n^{1-\delta}}).
    \]

    \paragraph{Probability of a Vertex Incident to $S$.}
    To analyze the probability that a vertex $ u $ is incident to at least one vertex in $ S $, we need to determine how many pairs between $ \{u\} $ and $ S $ remain random under the favorable condition $ B(C, f^t, b^t) $. Each such remaining random pair forms an edge independently with probability $ p $.

    Here, similar to \Cref{def-det1}, we define $\Det(u)$ by
    \[
    \Det(u):=\{w \mid \text{$X_{u,w}$ is determined by $B(C,f^t,b^t)$}\}.
    \]
    Here, $X_{u,w}$ is determined means that $B(C,f^t,b^t)$ forces the edge $(u,w)$ either to exist or not to exist. To see why we define this notion, see \Cref{para-det} for an intuitive explanation.
    
    Note that random variable $X_{u,w}$ is the indicator of $(u,w) \in E(G)$. Intuitively, $\Det(u)$ is a set of vertices $w$ s.t. whether $(u,w)$ is an edge is determined by baseline condition $B(C,f^t,b^t)$.
    
    We have following proposition.
    \begin{proposition}
        Conditioned on a favorable condition $B(C,f^t,b^t)$, with high probability over randomness of neighbor sampling, for every vertex $u \notin C\cup\{n\}$, $\abs{\Det(u) \cap S} \leq \frac{1}{2} |S|$.
    \end{proposition}
    \begin{proof}
        As $u \notin C\cup\{n\}$ and $B$ is a favorable condition, we have:
        \begin{itemize}
            \item Since $\largest{f^t(u)} \leq u+tK$, if $(u,w)$ is determined by $f^t$ in \Cref{obs-pair random2} (Item $2$), then there must be $u-tK \leq w \leq u + tK$.
            \item Since $b^t(u) \neq \bot$ only when $u > n-tK$ and $\largest{b^t(u)} \leq tK$, if $(u,w)$ is determined by $b^t$ in \Cref{obs-pair random2} (Item $3$), then there must be $w \in [n-tK,n] \cup [1,tK]$.
        \end{itemize}
        Together, we see that $\Det(u) = \Theta(tK) = \Theta(\frac{n}{\sqrt{n^{1-\delta}}})$.

        To compute $\Det(u) \cap S$, we first compute $\Det(u) \cap C$. Note that $C$ is the neighbors of $n$, fully determined by condition $B$; and $S \subseteq C$ is a set of uniform samples of $C$.

        Due to Item $(1)$ in \cref{def-fav cond2}, for every integer $0 \leq i \leq \frac{n}{K}-1$, the interval $[iK+1,iK+K]$ contains at most $2c\log n$ vertices in $C$. That is to say, every continuous length-$K$ interval contains at most $2c\log n$ vertices in $C$. Thus there are at most 
        \[
        \frac{|\Det(u)|}{K} \cdot 2c\log n = \Theta(\frac{n^\delta}{\sqrt{n^{1-\delta}}})
        \]
        vertices appeared in $\Det(u) \cap C$. Consider that every vertex in $\Det(u) \cap C$ is retained in $\Det(u) \cap S$ with independent probability $q = \frac{\sqrt{n^{1-\delta}}}{n^\delta} \log^2 n$, we have
        \[
        \E[\abs{\Det(u) \cap S}] = \Theta(\frac{n^\delta}{\sqrt{n^{1-\delta}}}) \cdot q = \Theta(\log^2 n).
        \]
        Note that $\E[|S|] = \Theta(n^\delta) \cdot q = \sqrt{n^{1-\delta}} \log^2 n \gg \log^2 n$, as $\delta < 1$. Thus a Chernoff bound for $\abs{\Det(u) \cap S}$ together with a union bound over all vertices finish the proof.
    \end{proof}
    This proposition tells that, for every vertex bad $u$, at least half of pairs between $u$ and $S$ still remain independently random.

    \paragraph{Existence of an Edge Between $P$ and $S$.}
    There are at least $l = \Theta(\sqrt{n^{1-\delta}})$ vertices on $P$ as we discussed before. Hence together, the total number of pairs between $P$ and $S$ that still remain random is at least
    \[
    \frac{1}{2} \cdot \Theta(\sqrt{n^{1-\delta}}) \cdot \sqrt{n^{1-\delta}} \log^2 n = \Theta(n^{1-\delta}\log^2 n).
    \]
    Further more, each of these pairs form an edge between $P$ and $S$ with probability $p = \frac{n^\delta}{n}$. Thus in expectation, there are $\Theta(n^{1-\delta}\log^2 n) \cdot p = \Theta(\log^2 n)$ such edges. Thus it is easy to see that w.p. $\geq 1 - \frac{1}{n^{100}}$ there is at least one vertex $v_j$ on $P$ incident to set $S$. Thus a path from $v$ to $v_j$ through edges on $P$, together with the edge from $v_j$ to $S$, show that $v$ is connected to $n$ in $H$. Above argument holds for any $v \leq tK$, thus a union bound over $tK$ vertices finishes the proof.
\end{proof}

\section{LCA for Minimum Spanning Tree}\label{sec:mst}

We restate the main result of this section for convenience.
\lcaMST*

Throughout this section, we assume that $G=(V,E,\w)$ is a weighted graph satisfying the assumptions of \cref{thm:lca mst}. 
We present an LCA for the minimum spanning tree of $G$ and prove \cref{thm:lca mst}.
We begin by introducing notation and auxiliary tools used in the analysis in \cref{sec:mstnotation}.
We then illustrate the main ideas in the simplified setting $W=2$ in \cref{subsec-W=2}, before completing the proof of \cref{thm:lca mst} for general $W$ in \cref{sec:mstgeneralcase}.

\subsection{Notation and Tools}\label{sec:mstnotation}

We will make use the following structural properties of  random subgraphs of expander graphs from percolation theory. 

\paragraph{Random subgraph of $d$-regular expander.}
For base graph $G$ being an $(n,d,\lambda)$-graph with $d \to \infty$ and $\lambda = o(d)$, let $G_p$ be a random subgraph of $G$ formed by retaining each edge of $G$ independently with probability $p$. The behavior of such random graphs has been studied (\cite{frieze2004emergence, diskin2024expansion}). In particular, there is a threshold of $p$ for the emergence of a giant component in $G_p$.

\begin{lemma}[\cite{frieze2004emergence, diskin2024expansion}]\label{lem-gcc}
Let $G$ be an $(n,d,\lambda)$-graph with $d \to \infty$ and $\lambda = o(d)$, and form $G_p$ by retaining each edge of $G$ independently with probability $p$. 
For any given constant $\varepsilon> 0$, we have:
\begin{itemize}
    \item (subcritical regime) If $p \leq \frac{1-\varepsilon}{d}$, then with probability\footnote{The original statements only claim that the probability tends to $1$ as $n \to \infty$, but the bound $1 - O(1/\log n)$ can be read from their proofs.} at least $1-O(\frac{1}{\log n})$, all connected components of $G_p$ have size $O(\log n)$.
    \item (supercritical regime) If $p \geq \frac{1+\varepsilon}{d}$, then with probability at least $1-O(\frac{1}{\log n})$, $G_p$ contains a unique giant component $L$ of size $\Theta(n)$, while all remaining components (if any) have size $O(\log n)$.
\end{itemize}
\end{lemma}

Note that in the original result of \cite{diskin2024expansion}, it is assumed that $\frac{\lambda}{d} \le \varepsilon^4$. Since we assume $\frac{\lambda}{d} = o(1)$ for simplicity, it suffices to take $\varepsilon>0$ to be an arbitrary constant, which merely excludes the regime where $p$ is extremely close to the critical threshold $1/d$. Beyond establishing the existence of a giant component $L$, Diskin et al.~\cite{diskin2024expansion} also study its internal expansion properties and the behavior of random walks on it.
In particular, they prove that random walks restricted to the giant component mix rapidly.

\begin{lemma}[\cite{diskin2024expansion}]\label{lem-gccmix}
Let $G$ be an $(n,d,\lambda)$-graph with $d \to \infty$ and $\lambda = o(d)$, and let $p \geq \frac{1+\varepsilon}{d}$ belong to supercritical regime in \cref{lem-gcc}.
Then on graph $G_p$, with probability at least $1-O(\frac{1}{\log n})$, the lazy random walk on the giant component $L$ has mixing time at most $O(\log^2 n)$.
\end{lemma}

We remark that in \cite{frieze2004emergence, diskin2024expansion}, the authors primarily focus on the technically most delicate regime, namely when $p$ lies in a small neighborhood of the critical probability $p=1/d$. Nevertheless, the same conclusions extend to a broader range of values of $p$, and in particular the above lemmas continue to hold. As $p$ increases beyond the critical threshold, the giant component absorbs an increasing fraction of the smaller components, while the remaining vertices exhibit behavior increasingly similar to the subcritical regime, in which all connected components have logarithmic size. On the other hand, the mixing-time bound relies on a lower bound on the vertex expansion of the giant component, which becomes more robust as $p$ increases.

\subsubsection{The Layered Structure of $G$}
Recall that $G$ is obtained from an expander $\bar G$ by assigning each edge an independent weight drawn uniformly at random from $[W]$.

For any $i \in [W]$, let $p_i$ denote the probability that an edge is assigned weight $i$.
In the setting of \cref{thm:lca mst}, we have $p_i = \frac{1}{W}$ for all $i$.
Let $G_{\le i}$ denote the subgraph of $G$ consisting of all edges with weight at most $i$, and let $\bar G_{\le i}$ denote its underlying unweighted graph.

Fix any $i \in [W]$.
The graph $\bar G_{\le i}$ can be viewed as a random subgraph of $\bar G$, obtained by retaining each edge independently with probability
\[
p_{\le i} := \sum_{j=1}^{i} p_j,
\]
namely, the probability that a given edge is assigned a weight at most $i$.

Equivalently, the graphs $\bar G_{\le i}$ induced by the random edge weights correspond exactly to bond percolation on the expander $\bar G$ with parameter $p_{\le i}$.
Throughout the remainder of this section, we condition on the event that each $\bar G_{\le i}$ lies in the supercritical regime.
This holds with probability at least $1-o(1)$ by a union bound over all $i \in [W]$, since
$p_{\le i} = \frac{i}{W} > \frac{1+\varepsilon}{d}$ when $W \le d/2$ and $W = o(\log n)$.

\subsection{Warm-up: $W = 2$}\label{subsec-W=2}

In this section, we consider the case $W=2$ (i.e., $p_1=p_2=1/2$) to illustrate our ideas.

\paragraph{Overview.}
From a high-level view, in order to construct an exact MST of $G=(V,E,\w)$, we aim to exploit edges of small weight as much as possible. In the case of $W=2$, the idea is simple: first build a minimum spanning forest (MSF) (denoted by $F_{\leq 1}$) on $G_{\leq1}$, then add proper weight-$2$ edges (denote these edges by $E_2$) to connect all connected components in $F_{\leq 1}$ while avoiding cycles. Clearly, the edge set $F_{\leq 1} \cup E_2$ forms an MST of $G$, since adding any more weight-$1$ edge to $F_{\leq 1}$ will cause a cycle. 

This reduces the problem of providing local access to an MST of $G$ to the following two \emph{tasks}:
\begin{enumerate}
    \item providing local access to an MSF (i.e. $F_{\leq 1}$) on $G_{\leq1}$; and
    \item providing local access to a set of edges that connects all components in $F_{\leq 1}$ into a single tree while remaining acyclic (i.e. $E_2$).
\end{enumerate}
We first deal with the simple one.

\subsubsection{Task~1}\label{sec:task1}

Recall that by \cref{lem-gcc} and that $p_1 = \frac{1}{2} \gg \frac{1}{d}$, there is a unique giant connected component (denoted by $L_{\leq 1}$) and every other component is of size $O(\log n)$. 

\paragraph{A subroutine \textsc{FindCC}$(G,u)$} We use a subroutine \textsc{FindCC}$(G,u)$ to test whether $u$ lies in a small connected component. It performs a BFS from $u$ in $G_{\le 1}$ with a size threshold and satisfies: \textbf{(1)} If $u$ lies in a small component of $G_{\le 1}$, it returns the entire component $C(u)$. \textbf{(2)} Otherwise, it returns $\bot$, indicating that $u$ lies in the large component $L_{\le 1}$. The subroutine runs in $O(d \log n)$ probes.

\paragraph{Local access to $F_{\leq 1}$} We are now able to provide local access to the edge set $F_{\leq 1}$ by combining the subroutine \textsc{FindCC} with a slightly modified version of \cref{alg-local expander} (still denoted by \textsc{inTree}).

\medskip
\noindent\textit{Query: Is an edge $(u,v)$ contained in $F_{\leq 1}$?}
\begin{itemize}
    \item If \textsc{FindCC}$(G,u)$ and \textsc{FindCC}$(G,v)$ return the same connected component $H$, then:
    \begin{enumerate}
        \item Construct a spanning tree $T_H$ of $H$ in a consistent manner;
        \item Return \textsc{Yes} if $(u,v) \in E(T_H)$, and \textsc{No} otherwise.
    \end{enumerate}
    \item If \textsc{FindCC}$(G,u)$ and \textsc{FindCC}$(G,v)$ both return~$\bot$, then return \textsc{inTree}$(L_{\leq 1}, (u,v))$.
    
    \item Otherwise, return \textsc{No}.
\end{itemize}

\begin{remark}\label{rmk-intree}
    Here, the subroutine \textsc{inTree} is adapted from \cref{alg-local expander}, which still answers if the queried edge should be include in spanning tree. 

    In particular, we can invoke \textsc{inTree}$(L_{\leq 1}, (u,v))$ treating $L_{\leq 1}$ as the input graph (ignoring edge weights), because neighbor and degree probes on $L_{\leq 1}$ can be simulated using at most $d$ invocations of \textsc{FindCC}.
    
    Moreover, since the lazy random walk on $L_{\leq 1}$ has mixing time at most $O(\log^2 n)$ (\Cref{lem-gccmix}), we can modify the length of lazy random walk in \cref{alg-local expander} to adapt it.
\end{remark}

Easy to see that above procedure is consistent with an MSF on $G_{\leq 1}$, using public random bits. Since $L_{\leq 1}$ is a $d$-bounded graph with lazy random walk mixing in polylogarithmic time, according to \cref{lem-expander local probe}, this procedure terminates after at most
\[
\tilde O\!\left(\sqrt{n} \cdot  d\right)\cdot d \cdot O(d \log n)
= \tilde O_{\phi}\!\left(\sqrt{n}\, d^3\right)
\]
probes. Finally, we note that the exponent of $d$ can be reduced to $2$ by further modifying \textsc{FindCC} to support neighbor and degree queries more efficiently.

\subsubsection{Task~2}

We now describe how to select an appropriate set of weight-$2$ edges $E_2$ to extend $F_{\le 1}$ into a (minimum) spanning tree of $G$. To this end, we make use of a spanning tree of the original graph. 

Specifically, fix\footnote{Such a vertex $s$ can be found by sampling a constant number of vertices and checking whether \textsc{FindCC}$(G,s)$ returns $\bot$.} $s$ to be a vertex in $L_{\le 1}$. Let $T$ be a (not necessarily minimum) spanning tree of $G$, obtained by invoking \cref{alg-global expander} on $G$ with root $s$. We then define a notion of \emph{rank} for each vertex with respect to the tree $T$.

\begin{definition}[$rank$]\label{def-rank}
For each vertex $u \in V$, let $P(u,s)$ denote the unique path from $u$ to $s$ in $T$.

We define $\mathrm{rank}(u)$ 
with respect to $T$ as the hop distance (i.e., the number of edges) from $u$ to the first vertex on $P(u,s)$ that lies in the set $L_{\le 1}$.
\end{definition}

Note that $\mathrm{rank}(u)$ is not the distance from $u$ to vertex set $L_{\leq 1}$ on $T$, since it only considers vertices of $L_{\leq 1}$ lying on the path from $u$ to the root. Since $s$ always lies in $L_{\leq 1}$, $P(u,s)$ must intersect with $L_{\leq 1}$, hence this is well-defined. The following observation says that we can compute the ranks locally.

\begin{observation}\label{obs-rank}
    For any $u\in V$, $\mathrm{rank}(u)$ can be computed locally using $\tilde O(\sqrt{n}\, d)$ adjacency-list probes.
\end{observation}
\begin{proof}
    Since the underlying graph of $G$ is an expander, we have already established local access to $T$ in \cref{alg-local expander} via the procedure \textsc{inTree}$(G,\cdot)$, which finishes in $\tilde O(\sqrt{n}\, d)$ probes according to \cref{lem-expander local probe}. 
    
    In particular, a single invocation of \textsc{inTree}$(G,(u,v))$ reveals the entire path $P(u,s)$—by combining the subgraph explored during the random walk with the lexicographical BFS rooted at $u$.

    Moreover, we note that $P(u,s)$ contains at most $O(\log^2 n)$ vertices. 
    Indeed, $P(u,s)$ is composed of a random walk of length $O(\log^2 n)$ together with BFS edges in $L_{\le 1}$. 
    Since the diameter of $L_{\le 1}$ is at most its mixing time, it follows that identifying the first vertex in $L_{\le 1}$ along $P(u,s)$ requires only an additional $d \cdot O(\log^2 n)$ probes, by invoking \textsc{FindCC} on each vertex.

\end{proof}

\begin{remark}\label{rmk-rank}
    Throughout the following, we slightly abuse procedure \textsc{inTree} to compute $\mathrm{rank}$ for convenience. 
    
    This is because we can always compute the unique path $P(u,s)$ in $\tilde O(\sqrt{n}\, d)$ probes using \textsc{inTree}, whenever the underlying graph is $d$-bounded with lazy random walk mixing in polylogarithmic time. We will later use this observation again w.r.t. giant component $L_{\leq i}$ as underlying graph, since it satisfies these conditions.
\end{remark}

\paragraph{Local access to $E_2$} Now we explain how to identify $E_2$. Globally, for each small component $C$ in $F_{\leq 1}$, exactly one edge of weight $2$ will be \emph{proposed by $C$} and included in $E_2$.
Since each component (except $L_{\leq 1}$) will propose exactly one edge in $T$, by a simple counting on all maintained edges, it suffices to argue that these proposed edges together connect all components in $F_{\leq 1}$. 

Next we present the LCA to identify the proposed edge for $C$, and why this edge connects $C$ to $L_{\leq 1}$. Here, for a component $C$, we slightly abuse notation by also using $C$ to denote the corresponding set of vertices. 

The algorithmic idea is as follows. For any vertex $u$, we use $C(u)$ to denote the small component containing $u$ in $G_{\leq 1}$. Consider $C(u)$ in the unweighted spanning tree $T$. Among all edges incident to $C(u)$, we are only interested in edges incident to the vertex of minimum $\mathrm{rank}$ in $C(u)$, denoted by $x$.
Our rule proposes for $C(u)$ the first edge on the path $P(x,s)$ in $T$.
The following LCA checks whether the queried edge $(u,v)$ coincides with the edge proposed by either $C(u)$ or $C(v)$. 

\medskip
\noindent\textit{Query: Is an edge $(u,v)$ contained in $E_2$?}
\begin{itemize}
    \item For each $r \in \{u,v\}$:
    \begin{enumerate}
        \item Let $C(r) \gets$ \textsc{FindCC}$(G,r)$. If $C(r) = \bot$, continue.
        \item Invoke \textsc{InTree} to compute $\mathrm{rank}(x)$ for every vertex $x \in C(r)$, with root fixed at $s \in L_{\leq 1}$;
        \item Let $x \gets \arg\min_{x \in C(r)} \mathrm{rank}(x)$ (breaking ties by vertex IDs).
        \item If $(u,v)$ is the first edge on $P(x,s)$, then return \textsc{Yes}.
    \end{enumerate}
    \item Return \textsc{No}.
\end{itemize}

Note that the query can be answered in
\[
\bigl(|C(u)| + |C(v)|\bigr)\cdot \tilde O_\phi(\sqrt{n}\, d)
= \tilde O_\phi(\sqrt{n}\, d)
\]
probes, since the probe complexity is dominated by the cost of computing $\mathrm{rank}$.

Now we explain why these proposed edges connect $F_{\leq 1}$ into a tree. 
\begin{lemma}\label{lem-e2connect}
    Proposal edges $E_2$ connect all small components $C(u)$ to $L_{\leq 1}$.
\end{lemma}
\begin{proof}
Let $C$ be any component (except for $L_{\leq 1}$) in $F_{\leq 1}$. Denote by $x$ the vertex in $C$ with minimum $\mathrm{rank}$. Let $(x,y)$ be the edge proposed by $C$, that is, the first edge on the path $P(x,s)$ in $T$.

If $\mathrm{rank}(x) = 1$, then without doubt $(x,y)$ connects $C$ to $L_{\leq 1}$, according to the definition of $\mathrm{rank}$. 
Otherwise, we have $\mathrm{rank}(x) \geq 2$ and $\mathrm{rank}(y) = \mathrm{rank}(x) - 1$. This means the proposal edge $(x,y)$ connects $C$ to another component (i.e. $C(y)$). Besides, among vertices in $C(y)$ the minimum $\mathrm{rank}$ is smaller than $\mathrm{rank}(x)$ by at least $1$. We can continue this inductive analysis by setting $C \gets C(y)$. Since $\mathrm{rank}$ is well-defined and bounded by diameter of $T$, this induction must end and it reaches some component incident to $L_{\leq 1}$.
\end{proof}

Since we use shared randomness to select $s$ and compute $\mathrm{rank}$, the output remains consistent with some set $E_2$. 

\paragraph{Final LCA for $W=2$.}
On query of any edge $(u,v)$, we check whether $(u,v) \in F_{\leq 1} \cup E_2$ by invoking above two LCAs respectively.

\subsection{General Case}\label{sec:mstgeneralcase}

In this section, we extend the ideas developed for the case $W=2$ to a general maximum edge weight $W$ via a recursive, layer-by-layer construction. Note that here $W$ satisfies $W\leq d/2$ and $W = o(\log n)$.

For any $i \in [W]$, let $G_{\leq i}$ denote the percolated subgraph of $G$ consisting of all edges of weight at most~$i$. For edge weight chosen uniformly from $[W]$, each percolated subgraph $G_{\leq i}$ lies in the supercritical regime (see \cref{lem-gcc}) since $\frac{1}{W} \geq \frac{2}{d} > \frac{1+\varepsilon}{d}$.

\subsubsection{A Global Construction}\label{sec:hierarchy}

\paragraph{Hierarchy across Giant Component.}
We begin by presenting a structural property of the input graph $G$. For the percolated graph $G_{\leq i}$, let $L_{\leq i}$ denote its unique linear-size (giant) component. We refer to all the remaining connected components of $G_{\leq i}$ as \emph{small} components.  Note that $L_{\leq W} = G$ and all the small components have size at most $O(\log n)$. The following observation, derived from \cref{lem-gcc}, formalizes a natural hierarchy among these components.

\begin{observation}\label{obs-hrk}
For $i = 1, \dots, W-1$, we have
\[
V(L_{\leq i}) \subseteq V(L_{\leq i+1}).
\]
\end{observation}

This observation exhibits a hierarchy of the input graph: from $L_{\leq 1}$ to $L_{\leq W}$, each giant component is fully contained in the one at the next level. We use this hierarchy to define layers of vertices.

\begin{definition}[layer]\label{def-layer}
The vertices at layer $i$ are denoted by $V_i$, defined as follows:
\begin{itemize}
    \item $V_1 := V(L_{\leq 1})$;
    \item $V_i := V(L_{\leq i}) \setminus V(L_{\leq i-1})$, for $i = 2, \dots, W$.
\end{itemize}
We use $layer(u)$ to denote the layer index of $u$, i.e. $layer(u):=i$ if $u \in V_i$.
\end{definition}

This definition immediately yields the following properties.

\begin{claim}\label{clm-vi}
For any vertex $u \in V$, the following hold:
\begin{enumerate}
    \item[(1)] For any $i \in [W]$, if $u \in V_i$, then $u \in L_{\leq i}$.
    \item[(2)] For any $i \in [W] \setminus \{1\}$, if $u \in V_i$ then $u$ lies in a small component $C(u)$ of $G_{\leq i-1}$. Moreover, $C(u)$ contains only vertices from $V_i$.
\end{enumerate}
\end{claim}

\begin{proof}
Statement (1) and the first part of (2) follow directly from the definition.
    
For the second part of (2), note that $C(u)$ only contains edges of weight at most $(i-1)$. As a result, any other vertex $w$ in $C(u)$ must be connected to $u$ through edges with weight at most $(i-1)$, we denote the path by $P(u,w)$. 

If $w \in V_j$ for some $j \leq i-1$, then by (1) and \Cref{obs-hrk} we have
\[
w \in V(L_{\leq j}) \subseteq V(L_{\leq i-1}).
\]
Since the path $P(u,w)$ survives in $G_{\leq i-1}$, this would imply $u \in V(L_{\leq i-1})$, contradicting the assumption that $u$ lies in a small component of $G_{\leq i-1}$.

Similarly, if $w \in V_j$ for some $j \geq i+1$, then $w \notin V(L_{\leq j-1}) \supseteq V(L_{\leq i})$, and hence $w \notin V(L_{\leq i})$. Therefore, every path from $w$ to $V(L_{\leq i})$ must contain an edge of weight at least $i+1$. Since $u \in V_i \subseteq V(L_{\leq i})$, this contradicts the existence of the path $P(u,w)$ consisting of edges of weight at most $i-1$.
\end{proof}

The claim implies that if we perform a BFS from each $u \in V_i$ in $G_{\leq i-1}$ to explore its component $C(u)$, then all these components together cover $V_i$ and contain no vertices outside $V_i$.

\paragraph{Recursive Construction of the MST.}

We start by constructing the MST $T_{\leq 1}$ of $L_{\leq 1}$. Then, for each $i \in [W] \setminus \{1\}$, we describe how to construct a minimum spanning tree $T_{\leq i}$ of $L_{\leq i}$ recursively from $T_{\leq i-1}$. Clearly, the final tree $T_{\leq W}$ is the desired MST of $G$.

To construct $T_{\leq i}$, we consider the percolated subgraph $G_{\leq i-1}$. By \Cref{obs-hrk} and \Cref{def-layer}, the vertices in $L_{\leq i}$ can be partitioned into two categories: vertices in $L_{\leq i-1}$ and vertices in $V_i$.

In $G_{\leq i-1}$, the subgraph induced by $L_{\leq i-1}$ is connected (with suitable expansion properties by \cref{lem-gccmix}), while each vertex in $V_i$ lies in a small component. When edges of weight $i$ are added back, these two parts merge into the connected component $L_{\leq i}$. By an argument analogous to \cref{subsec-W=2}, the MST $T_{\leq i}$ consists of:
\begin{itemize}
    \item a minimum spanning tree on $L_{\leq i-1}$, namely $T_{\leq i-1}$;
    \item a minimum spanning forest formed by the MSTs of the small components induced by $V_i$ in $G_{\leq i-1}$;
    \item edges proposed by each small component induced by $V_i$ in $G_{\leq i-1}$, denoted by $E_i$.
\end{itemize}
Consequently, $T_{\leq W}$ can be constructed recursively; and the edges fall into \textbf{3 categories} accordingly:
\begin{enumerate}
    \item Spanning tree of $L_{\leq 1}$;
    \item MST of every small component;
    \item Edges proposed by small components.
\end{enumerate}
Therefore, in order to design an LCA, it suffices to identify these cases for queried edge respectively.

\subsubsection{LCA}

\paragraph{Key Subroutine: \textsc{Layer-Comp}.} Before describing the final LCA, we first present a key subroutine \textsc{Layer-Comp}, in \cref{alg-layer comp}. For a given vertex $u$, procedure \textsc{Layer-Comp}$(G,u)$ returns the layer index of $u$ and corresponding small component $C(u)$.  This subroutine performs a role similar with \textsc{FindCC} in last section.

\begin{algorithm}[htb]  
\caption{\textsc{Layer-Comp}$(G,u)$: determine the layer of $u$ and return its component}\label{alg-layer comp}
\KwIn{A weighted graph $G=(V,E,\w)$ and a vertex $u\in V$.}
\KwOut{The layer index $i$ of $u$, and the small component $C(u)$.}

Let $\theta \gets c \log n$ for a sufficiently large constant $c$. 

\tcp{$\theta$ upper bounds the size of every small component.}

Initialize $C(u)$ to be graph containing isolated vertex $u$.

\For{$i = 1,2,\ldots,W$}{
    $H \gets C(u)$\;
    
    \While{there exists an \emph{unexplored} edge $e$ with $w_e \leq i$
           incident to $H$}{
        Update $H$ by exploring $e$\;
        \If{$|V(H)| > \theta$}{
            \Return{$(i, C(u))$}. \tcp{At this point $u$ enters the giant component of $G_{\le i}$; return the fully explored small component in $G_{\le i-1}$.}
        }
    }
    $C(u) \gets H$\;
}
\end{algorithm}

\begin{lemma}\label{lem-lc}
Suppose that $u \in V_i$. The procedure \textsc{Layer-Comp}$(G,u)$ returns the value of $i$ and a component $C(u)$ containing $u$, satisfying
    \begin{itemize}
        \item if $i \geq 2$, then $C(u)$ is the small component containing $u$ in $G_{\leq i-1}$.
        \item \textsc{Layer-Comp}$(G,u)$ finishes in $O(d \log n)$ probes.
    \end{itemize}
\end{lemma}
\begin{proof}
    
For the first statement, whenever $|V(H)| > \theta$, the vertex $u$ must belong to the giant component of $G_{\le i}$. In this case, the algorithm returns the connected component $C(u)$ fully explored in last iteration through edges of weight at most $i-1$, which is exactly the small component in $G_{\le i-1}$ by definition.

For the second statement, note that the procedure finishes once it explores a subgraph $H$ with more than $\theta$ vertices. Since the input graph is $d$-bounded, the procedure visits at most $|V(H)| \cdot d = O(d\log n)$ edges. Hence it finishes in $O(d\log n)$ probes.
\end{proof}

\begin{corollary}\label{cor-adj to Li}
Any adjacency-list probe to $L_{\le i}$ can be answered using $O(d \log n)$ probes to $G$.
\end{corollary}

\begin{proof}
For any vertex $u \in V(L_{\le i})$, the adjacency list of $u$ in $L_{\le i}$ consists exactly of those vertices incident to $u$ via edges of weight at most~$i$. 
Hence, we first determine whether $u \in V(L_{\le i})$ using $O(d \log n)$ probes by invoking \textsc{Layer-Comp}$(G,u)$. 
If so, we then use at most $d$ extra probes to get all neighbors of $u$ in $L_{\le i}$.
\end{proof}

\paragraph{Locally Compute Proposal Edge of Small Component.}

In this section, we use \emph{rank} to compute the proposal edge of a small component, presented in \cref{alg-probbc}. Here, the notion $\mathrm{rank}$ is generalized from \cref{def-rank}.
For vertices in any layer $V_i$, $\mathrm{rank}$ is computed w.r.t. an (unweighted) spanning tree of $L_{\leq i}$. Specifically, \cref{def-rank} is defined for layer $V_2$.

Note that all these spanning trees of $\{L_{\leq i}\}$ share a common root $s$, which is chosen arbitrarily from $L_{\leq 1}$ at first. That is to say, when we invoke LCA \textsc{InTree}$(L_{\leq i}, \cdot)$ for any $i$, we force \cref{alg-local expander} to choose $s$ as root. Since root can be chosen arbitrarily in \cref{alg-local expander}, this assumption does not affect correctness. Note that $s \in V(L_{\leq 1})$ implies $s \in V(L_{\leq i})$ for every $i \in [W]$, according to \cref{obs-hrk}.

Instead of defining $\mathrm{rank}$ separately for every layer like \cref{def-rank}, we show the details of computing $\mathrm{rank}$ in \cref{alg-probbc}, because they share the same ideas.

\begin{algorithm}[htb]  
\caption{\textsc{PropByComp}$(G,u)$: compute the proposal edge of $C(u)$}\label{alg-probbc}
\KwIn{A weighted graph $G=(V,E,\w)$ and a vertex $u \in V$.}

$i,C(u) \gets$ \textsc{Layer-Comp}$(G,u)$\;

\If {$i = 1$}{\Return $\bot$.}

Let $s$ be the predetermined root in $L_{\leq 1}$.

\For{each vertex $x$ in $C(u)$}{
\If{$\exists$ $e$ incident to $x$ with $w_e=i$}{
\tcp{Compute $\mathrm{rank}(x)$}
Invoke \textsc{InTree}$(L_{\leq i}, e)$ to get the unique path $P(x,s)$ from $x$ to $s$ on the tree\;

\For{each vertex $y$ along path $P(x,s)$}{
$i_y,C(y) \gets$ \textsc{Layer-Comp}$(G,y)$\;
\If{$i_y \leq i-1$}{\tcp{$y$ lies in $L_{\leq i-1}$}
$\mathrm{rank}(x) \gets$ hop-distance between $x$ and $y$ on $P(x,s)$\;
\textbf{break}.}
}

}
}

$x \gets \arg\min_{x \in C(u)} \mathrm{rank}(x)$. \tcp{Vertex with minimum rank, breaking ties by IDs.}

\Return first edge on $P(x,s)$.

\end{algorithm}

\begin{lemma}\label{lem-prop by comp}
    \Cref{alg-probbc} satisfies the following:
    \begin{enumerate}
        \item\label{lemmapropbycomp_item1}(Consistency) If $u \in V_i$, then for any vertex $x \in C(u)$, \textsc{PropByComp}$(G,x)$ outputs a same edge $e$ with $w_e=i$;
        \item\label{lemmapropbycomp_item2} (Connectivity) Denote by $E_i$ set of edges output by \textsc{PropByComp}$(G,u)$ for all $u \in V_i$. Then $E_i$ connects $C(u)$ to $L_{\leq i-1}$ for every $u \in V_i$.
        \item\label{lemmapropbycomp_item3} (Complexity) It finishes in $\tilde O(\sqrt{n}d^2)$ adjacency-list probes to $G$.
    \end{enumerate}
\end{lemma}
\begin{proof}
For \Cref{lemmapropbycomp_item1}: Since for any $x \in C(u)$, \textsc{Layer-Comp}$(G,x)$ returns $C(u)$ by \cref{lem-lc}. Hence \textsc{PropByComp}$(G,x)$ always outputs a same edge $e=(x,y)$ incident to $C(u)$, where $x$ is the vertex with minimum $\mathrm{rank}$. We further have $w_e=i$; otherwise if $w_e \leq i-1$, both $x$ and $y$ will be contained in $C(u)$. Since $(x,y)$ is the first edge on path $P(x,s)$, this means $P(y,s)$ is a sub-path of $P(x,s)$ by \cref{clm-subpath}, hence we have $\mathrm{rank}(y) < \mathrm{rank}(x)$. This contradicts with $x$ being vertex in $C(u)$ with minimum $\mathrm{rank}$. 

For \Cref{lemmapropbycomp_item2}: According to the same reason in \cref{rmk-intree} and \cref{rmk-rank}, \textsc{InTree}$(L_{\leq i}, e)$ returns the unique path $P(x,s)$ and $\mathrm{rank}(x)$ is the hop-distance from $x$ to the first vertex in $L_{\leq i-1}$ on $P(x,s)$. Thus \Cref{lemmapropbycomp_item2} comes from simply adapting the proof of \cref{lem-e2connect}, by substituting $E_2$ with $E_i$, and $L_{\leq 1}$ with $L_{\leq i-1}$.

For \Cref{lemmapropbycomp_item3}, it is easy to see that most of the probe complexity in \Cref{alg-probbc} arises from computing $\mathrm{rank}$ for vertices in $C(u)$, where $|C(u)| = O(\log n)$. By \cref{lem-expander local probe}, an invocation of \textsc{InTree}$(L_{\leq i}, e)$ use $O(\sqrt{n}d)$ adjacency-lsit probes to $L_{\leq i}$. Hence the total costs is bounded by
\[
\tilde O(\sqrt{n}d) \cdot O(d\log n) \cdot  |C(u)|= \tilde O(\sqrt{n}d^2)
\]
using \cref{cor-adj to Li}.

\end{proof}

\paragraph{Final LCA: Local Access to MST of $G$.}
We describe the final LCA for MST in \Cref{alg-in mst}.

\begin{algorithm}[htb]  
\caption{\textsc{InMST}$(G,(u,v))$: edge-membership of $(u,v)$ w.r.t. an MST of $G$}\label{alg-in mst}
\KwIn{A weighted graph $G=(V,E,\w)$.}

\tcp{On query of edge $(u,v)$:}

$i_u,C(u) \gets$ \textsc{Layer-Comp}$(G,u)$\;

$i_v,C(v) \gets$ \textsc{Layer-Comp}$(G,v)$.

\If{$i_u=i_v=1$}{ \tcp{Both Lie in $L_{\leq 1}$}
\Return \textsc{InTree}$(L_{\leq 1}, (u,v))$.
}

\If{$C(u) = C(v)$}{ \tcp{Lie in same small component}
        Construct a minimum spanning tree $T$ of $C(u)$ in a consistent manner\;
        \If{$(u,v) \in E(T)$}{\Return \textsc{Yes}.}
        \Else{\Return \textsc{No}.}
}

\If{$(u,v) =$\textsc{PropByComp}$(G,u)$ or $(u,v) =$\textsc{PropByComp}$(G,v)$}{
\Return \textsc{Yes}.
}
\Return \textsc{No}.
\end{algorithm}

\begin{lemma}
    \cref{alg-in mst} is consistent with a minimum spanning tree of $G$. Moreover, it finishes in $\tilde O(\sqrt{n}d^2)$ probes.
\end{lemma}
\begin{proof}
We prove by induction that: the edges maintained by \textsc{InMST}$(G,\cdot)$ induce a minimum spanning tree on every $L_{\leq i}$.
\begin{itemize}
    \item \textbf{Base case:} For $L_{\leq 1}$, \textsc{InMST}$(G,\cdot)$ maintains a spanning tree by Line 3-4. Since $L_{\leq 1}$ only contains weight-$1$ edges, this is also a minimum spanning tree.
    \item \textbf{Inductive Step:} Assume it holds for $1 \leq i \leq W-1$. This means on graph $L_{\leq i}$, the maintained edges form a minimum spanning tree. We show this holds for $i+1$.

    This is because, for every small component $C$ induced by $V_{i+1}$ on $G_{\leq i}$ (see \cref{def-layer} and \cref{clm-vi}), the edges (maintained by Line 5-10) induce an MST on $C$. Moreover, \cref{clm-vi} tells that these small components exactly cover $V_{i+1}$. Hence we can only use weight-$(i+1)$ edges to connect these small components to $L_{\leq i}$. According to \cref{lem-prop by comp}, Line 11-12 maintains a set edges $E_{i+1}$ to achieve this. Since every small component only proposes $1$ edge in $E_{i+1}$, these edges induce a minimum spanning tree on $L_{\leq i+1}$.
\end{itemize}
As a result, we conclude that \textsc{InMST}$(G,\cdot)$ maintain a minimum spanning tree on $L_{\leq W}=G$.

Now we bound the probe complexity. Note that \textsc{InTree}$(L_{\leq 1}, (u,v))$ finishes in $\tilde O(\sqrt{n} d^2)$ probes, combining \cref{lem-expander local probe} and \cref{cor-adj to Li}. As a result, the total number of probes is at most $\tilde O(\sqrt{n}d^2)$, by \cref{lem-prop by comp}.
\end{proof}

This completes the proof of \cref{thm:lca mst}.

\bibliographystyle{alpha}
\bibliography{citations}

\newpage
\appendix

\section{More discussions on LCAs for MST}
\subsection{Hardness of MST LCAs with Arbitrary Edge Weights}\label{sec:lowerboundforMST}

We establish the following lower bound, showing that the random edge-weight assumption in \Cref{thm:lca mst} is necessary.

\begin{restatable}{theorem}{weightedLB}\label{thm:weightedLB}
Given adjacency-list access to a graph $G=(V,E,\w)$ with arbitrary edge weights, there does not exist an LCA for the minimum spanning tree problem with $o(n)$ probe complexity, even when the underlying graph $(V,E)$ is a $d$-bounded expander and all edge weights are restricted to $\{1,2\}$.
\end{restatable}

\begin{proof}
We prove by contradiction. Suppose there exists an LCA $\mathcal{A}$ satisfying the stated conditions. We show that this would imply the existence of a sublinear-probe LCA for constructing a spanning tree in \emph{any} $d$-bounded connected graph, contradicting known lower bounds for spanning trees~\cite{reut2014local,levi2020local}.

Let $H=(V,E_H)$ be an arbitrary $d$-bounded connected graph. Let $G_0=(V,E_0)$ be a fixed $d$-bounded expander graph, independent of $H$ and known explicitly to the algorithm. We construct a weighted graph $G=(V,E,\w)$, where $E = E_0 \cup E_H$, and the edge weights are defined as
\[
w_e =
\begin{cases}
1, & \text{if } e \in E_H, \\
2, & \text{if } e \in E_0 \setminus E_H.
\end{cases}
\]

Observe that:
\begin{enumerate}
    \item Since $H=(V,E_H)$ is connected and all edges in $E_H$ have weight $1$, the graph $G$ contains a spanning subgraph consisting entirely of weight-$1$ edges. Consequently, any minimum spanning tree $T$ of $G$ must consist solely of weight-$1$ edges.
    \item By construction, an edge has weight $1$ if and only if it belongs to $E_H$. Therefore, the edges of $T$ form a spanning tree of $H$.
\end{enumerate}

Moreover, any probe to the adjacency list of $G$ can be simulated using at most one probe to $H$. Hence, if $\mathcal{A}$ provides local access to an exact MST of $G$ using $o(n)$ probes, then we obtain an $o(n)$-probe LCA for determining edge membership in a spanning tree of $H$. This contradicts the known lower bounds for LCAs for spanning trees.

It remains to verify that $\mathcal{A}$ is applicable to the constructed graph $G$. Let $\bar G=(V,E)$ denote the underlying unweighted graph. By the definition of conductance, we have
\[
\phi_{\bar G} \ge \tfrac{1}{2}\,\phi_{G_0},
\]
and hence $\bar G$ is an expander. Furthermore, all edge weights of $G$ lie in $\{1,2\}$. Therefore, by assumption, $\mathcal{A}$ correctly answers edge-membership queries with respect to an exact MST of $G$ using $o(n)$ probes, completing the contradiction.

\end{proof}

\subsection{Beyond Uniform Distribution}\label{sec:beyonduniform}

In this section, we explain why the algorithm naturally extends to more general distributions on edge weights. The key observation is that the algorithm does not depend on the uniformity of the weight distribution, but only on whether the induced percolated subgraphs lie in the subcritical or supercritical regime.

Formally, let $\mathcal{D}_W$ be a discrete distribution supported on $[W]$, and define
\[
p_{\le i} := \Pr_{w \sim \mathcal{D}_W}[\, w \le i \,].
\]
We say that $\mathcal{D}_W$ is \emph{non-critical} for $G$ if for every $i \in [W]$, 
either $p_{\le i} \le \frac{1-\varepsilon}{d}$ or $p_{\le i} \ge \frac{1+\varepsilon}{d}$, where $\varepsilon$ can be any small constant. 

Assume that the edge weights of $G$ are drawn independently from a non-critical distribution.
If the subgraph $G_{\le i}$ lies in the supercritical regime for any weight $i$, then all arguments remain identical to the case of the uniform distribution. 

Otherwise, let $t$ be the largest weight such that $G_{\le t}$ lies in the subcritical regime. 
By Lemma~\ref{lem-gcc}, all connected components of $G_{\le t}$ have logarithmic size, and hence the minimum spanning forest of $G_{\le t}$ can be computed efficiently. 
As a result, we can still locally access a minimum spanning tree of $L_{\le t+1}$ built upon an unweighted spanning tree of $L_{\le t+1}$ and the notion of \emph{rank}. 
The only difference is that here \emph{rank} is defined with respect to the small component containing the root $s$, rather than a giant component. The minimum spanning tree of $G$ then follows accordingly.

\end{document}